\documentclass[11pt]{article}
\usepackage[T1]{fontenc}
\usepackage[utf8]{inputenc}
\usepackage{lmodern}
\usepackage[margin=1in]{geometry}
\usepackage{microtype}
\usepackage{amsmath,amssymb,amsthm,mathtools,bm}
\usepackage{stmaryrd}
\usepackage{dsfont}
\usepackage{booktabs,longtable,multirow,array,tabularx}
\usepackage{xcolor}
\usepackage{graphicx}
\usepackage{caption}
\usepackage{subcaption}
\usepackage{enumitem}
\usepackage{algorithm}
\usepackage{algpseudocode}
\usepackage{listings}
\usepackage{tikz}
\usetikzlibrary{arrows.meta,positioning,shapes.geometric,shapes.misc,fit,calc,backgrounds,decorations.pathreplacing,patterns}
\usepackage{pgfplots}
\pgfplotsset{compat=1.18}
\usepgfplotslibrary{groupplots,fillbetween}
\usepackage[numbers,sort&compress]{natbib}
\usepackage{xurl}
\usepackage[colorlinks=true,linkcolor=blue!50!black,citecolor=green!40!black,urlcolor=blue!60!black]{hyperref}
\usepackage[capitalise,nameinlink]{cleveref}

\theoremstyle{plain}
\newtheorem{theorem}{Theorem}
\newtheorem{proposition}{Proposition}

\theoremstyle{definition}
\newtheorem{definition}{Definition}
\newtheorem{assumption}{Assumption}
\newtheorem{example}{Example}
\theoremstyle{remark}
\newtheorem{remark}{Remark}
\crefname{assumption}{Assumption}{Assumptions}

\newcommand{\E}{\mathbb{E}}
\newcommand{\Q}{\mathbb{Q}}
\newcommand{\Prob}{\mathbb{P}}
\newcommand{\Var}{\operatorname{Var}}

\newcommand{\R}{\mathbb{R}}

\newcommand{\ind}{\mathds{1}}
\newcommand{\dd}{\,\mathrm{d}}
\newcommand{\VaR}{\operatorname{VaR}}
\newcommand{\CVaR}{\operatorname{CVaR}}
\newcommand{\lfp}{\operatorname{lfp}}
\newcommand{\gfp}{\operatorname{gfp}}

\newcommand{\LN}{\mathrm{LN}}
\newcommand{\PERT}{\mathrm{PERT}}
\newcommand{\Pois}{\mathrm{Pois}}
\newcommand{\musd}[1]{\ensuremath{\$#1\,\mathrm{m}}}
\newcommand{\Always}{\mathop{\square}}
\newcommand{\Eventually}{\mathop{\lozenge}}

\newcommand{\pathto}{\rightsquigarrow}
\newcommand{\code}[1]{\texttt{#1}}
\newcommand{\Adm}{\mathsf{Adm}}

\lstdefinestyle{py}{language=Python,basicstyle=\ttfamily\footnotesize,keywordstyle=\color{blue!60!black}\bfseries,
  commentstyle=\color{green!35!black}\itshape,stringstyle=\color{red!50!black},showstringspaces=false,
  frame=single,framerule=0.3pt,rulecolor=\color{black!30},breaklines=true,numbers=left,numberstyle=\tiny\color{black!50},
  xleftmargin=1.5em,framexleftmargin=1.5em,columns=fullflexible,keepspaces=true}

\pgfplotsset{every axis/.append style={font=\footnotesize, tick label style={font=\scriptsize}, label style={font=\footnotesize},
  legend style={font=\scriptsize, draw=black!30}}}
\definecolor{cA}{RGB}{31,119,180}
\definecolor{cB}{RGB}{214,39,40}
\definecolor{cC}{RGB}{44,160,44}
\definecolor{cD}{RGB}{255,127,14}
\definecolor{cE}{RGB}{148,103,189}
\definecolor{cF}{RGB}{140,86,75}
\definecolor{cG}{RGB}{227,119,194}
\definecolor{cH}{RGB}{127,127,127}

\newcommand{\rBankMean}{463.8}

\newcommand{\rBankSkew}{5.17}

\newcommand{\rBankPnn}{1,226.0}

\newcommand{\rBankPlugRank}{36.3}
\newcommand{\rBankPlugGap}{22.6}
\newcommand{\rBankAleaShare}{25.4}
\newcommand{\rBankETR}{2.165}
\newcommand{\rBankTdMode}{2.61}
\newcommand{\rBankTdMeanPath}{2.67}
\newcommand{\rBankTdMC}{2.93}
\newcommand{\rBankTdPninety}{3.91}

\newcommand{\rBankFulkGap}{9.9}
\newcommand{\rBankJswMed}{15.3}

\newcommand{\rBankJswE}{36.8}

\newcommand{\rBankJswRatio}{2.41}
\newcommand{\rBankJincRatio}{1.92}
\newcommand{\rBankVstar}{148.4}
\newcommand{\rBankVO}{400}
\newcommand{\rBankBeta}{3.068}

\newcommand{\rBankArchMult}{2.53}

\newcommand{\rBankArchRbf}{0.9219}
\newcommand{\rBankArchRfb}{0.8143}
\newcommand{\rBankArchDTbf}{178}
\newcommand{\rBankArchDTfb}{467}
\newcommand{\rBankRuntime}{17.2}
\newcommand{\rBankSzeroMean}{603.9}

\newcommand{\rBankSzeroTau}{2.17}

\newcommand{\rBankSthreeMean}{463.8}

\newcommand{\rBankSthreePTd}{69.8}
\newcommand{\rBankSthreeTau}{1.50}

\newcommand{\rBankSobTopST}{0.64}

\newcommand{\rBankSobSumS}{0.92}
\newcommand{\rBankBondAlpha}{0.856}
\newcommand{\rBankBondB}{662.7}
\newcommand{\rBankBondRatio}{1.43}

\newcommand{\rTelMean}{446.8}

\newcommand{\rTelSkew}{0.87}

\newcommand{\rTelPnn}{805.6}

\newcommand{\rTelPlugRank}{36.0}
\newcommand{\rTelPlugGap}{12.1}
\newcommand{\rTelAleaShare}{1.4}

\newcommand{\rTelTdMode}{1.87}
\newcommand{\rTelTdMeanPath}{1.92}
\newcommand{\rTelTdMC}{2.10}
\newcommand{\rTelTdPninety}{2.77}

\newcommand{\rTelFulkGap}{9.4}

\newcommand{\rTelJswRatio}{1.85}
\newcommand{\rTelJincRatio}{1.44}
\newcommand{\rTelVstar}{448.2}
\newcommand{\rTelVO}{700}
\newcommand{\rTelBeta}{2.261}

\newcommand{\rTelDelayOptTwo}{142.0}
\newcommand{\rTelDelayDcfTwo}{162.2}
\newcommand{\rTelDelayMean}{93.9}
\newcommand{\rTelDelayDcfMean}{104.9}
\newcommand{\rTelArchMult}{1.97}

\newcommand{\rTelArchRbf}{0.8484}
\newcommand{\rTelArchRfb}{0.7227}
\newcommand{\rTelArchDTbf}{180}
\newcommand{\rTelArchDTfb}{356}
\newcommand{\rTelRuntime}{26.0}

\newcommand{\rTelSthreecMean}{446.8}
\newcommand{\rTelSthreecPnf}{669.4}

\newcommand{\rTelSthreetPnf}{699.9}

\newcommand{\rTelSobTopST}{0.24}

\newcommand{\rTelSobSecST}{0.23}
\newcommand{\rTelSobSumS}{0.92}

\newcommand{\rTelBondB}{572.4}
\newcommand{\rTelBondRatio}{1.28}

\newcommand{\rAirMean}{345.3}

\newcommand{\rAirSkew}{1.05}

\newcommand{\rAirPnn}{679.3}

\newcommand{\rAirPlugRank}{44.0}
\newcommand{\rAirPlugGap}{7.9}
\newcommand{\rAirAleaShare}{8.7}

\newcommand{\rAirTdMode}{2.45}
\newcommand{\rAirTdMeanPath}{2.50}
\newcommand{\rAirTdMC}{2.74}
\newcommand{\rAirTdPninety}{3.63}

\newcommand{\rAirFulkGap}{9.6}

\newcommand{\rAirJswRatio}{3.77}
\newcommand{\rAirJincRatio}{4.06}
\newcommand{\rAirVstar}{246.1}
\newcommand{\rAirVO}{600}
\newcommand{\rAirBeta}{2.561}

\newcommand{\rAirArchMult}{2.43}

\newcommand{\rAirArchRbf}{0.9599}
\newcommand{\rAirArchRfb}{0.9053}
\newcommand{\rAirArchDTbf}{45}
\newcommand{\rAirArchDTfb}{118}
\newcommand{\rAirRuntime}{23.1}

\newcommand{\rAirStwoMean}{345.3}

\newcommand{\rAirSthreeMean}{363.5}

\newcommand{\rAirSobTopST}{0.58}

\newcommand{\rAirSobSecST}{0.23}
\newcommand{\rAirSobSumS}{0.89}

\newcommand{\rAirBondB}{463.3}
\newcommand{\rAirBondRatio}{1.34}

\newcommand{\rEconTruth}{0.195}
\newcommand{\rEconTwfeBias}{-0.0324}
\newcommand{\rEconTwfeRmse}{0.0329}
\newcommand{\rEconTwfeRel}{-16.6}
\newcommand{\rEconCsBias}{-2.7\times 10^{-6}}
\newcommand{\rEconCsRmse}{0.0146}
\newcommand{\rEconCsCov}{99}
\newcommand{\rEconCovReps}{100}
\newcommand{\rEconR}{200}

\newcommand{\rSynthRatio}{60.0}
\newcommand{\rSynthP}{0.032}
\newcommand{\rSynthDonors}{16}

\newcommand{\rTotalRuntime}{97}
\newcommand{\rN}{200{,}000}
\newcommand{\rVBacklogMaxErr}{1.44}
\newcommand{\rVUnifErr}{9.8\times 10^{-14}}
\newcommand{\rVUnifCumErr}{5.2\times 10^{-8}}
\newcommand{\rVQcov}{95.7}
\newcommand{\rVQnom}{96.0}
\newcommand{\rVAcuteMaxErr}{5.3\times 10^{-12}}
\newcommand{\rVOptMaxZ}{1.93}
\newcommand{\rVOptQuadErr}{5.2\times 10^{-5}}
\newcommand{\rVEuler}{4.3\times 10^{-14}}
\newcommand{\rBankTornBase}{399.4}
\newcommand{\rTelTornBase}{419.2}
\newcommand{\rAirTornBase}{359.6}

\newcommand{\rBankMeanNum}{463.8}

\newcommand{\rBankPnnNum}{1226.0}
\newcommand{\rBankBondNum}{662.7}
\newcommand{\rBankPlugNum}{359.1}
\newcommand{\rBankHistMax}{0.002787}
\newcommand{\rBankHistXmax}{1394}
\newcommand{\rBankBondRatioNum}{1.429}
\newcommand{\rBankUatEL}{17.3}
\newcommand{\rBankUatB}{6.5}
\newcommand{\rBankFatEL}{0.601}
\newcommand{\rBankAttrFundMeanShare}{59.1}
\newcommand{\rBankAttrFundTailShare}{69.5}

\newcommand{\rBankAcIncMeanShare}{10.4}
\newcommand{\rBankAcIncTailShare}{15.5}

\newcommand{\rBankMeanHW}{0.22}
\newcommand{\rBankPnfHW}{0.46}
\newcommand{\rTelMeanNum}{446.8}

\newcommand{\rTelPnnNum}{805.6}
\newcommand{\rTelBondNum}{572.4}
\newcommand{\rTelPlugNum}{392.9}
\newcommand{\rTelHistMax}{0.003925}
\newcommand{\rTelHistXmax}{864}
\newcommand{\rTelBondRatioNum}{1.281}
\newcommand{\rTelUatEL}{10.7}
\newcommand{\rTelUatB}{2.9}
\newcommand{\rTelFatEL}{0.551}

\newcommand{\rTelFunDelMeanShare}{63.2}
\newcommand{\rTelAcIncMeanShare}{2.7}

\newcommand{\rTelMeanHW}{0.12}
\newcommand{\rTelPnfHW}{0.27}
\newcommand{\rAirMeanNum}{345.3}

\newcommand{\rAirPnnNum}{679.3}
\newcommand{\rAirBondNum}{463.3}
\newcommand{\rAirPlugNum}{317.9}
\newcommand{\rAirHistMax}{0.004072}
\newcommand{\rAirHistXmax}{733}
\newcommand{\rAirBondRatioNum}{1.342}
\newcommand{\rAirUatEL}{13.1}
\newcommand{\rAirUatB}{3.5}
\newcommand{\rAirFatEL}{0.541}

\newcommand{\rAirFunDelMeanShare}{82.5}

\newcommand{\rAirPlugBest}{S0}
\newcommand{\rAirMeanHW}{0.15}
\newcommand{\rAirPnfHW}{0.27}
\newcommand{\rEconHazard}{0.0235}

\newcommand{\rBankEAF}{1.406}
\newcommand{\rBankEexp}{1.11}
\newcommand{\rBankFexp}{0.320}
\newcommand{\rBankKSLOC}{45}
\newcommand{\rBankSF}{20}
\newcommand{\rBankPMnom}{282.8}
\newcommand{\rBankTDEVnom}{22.3}
\newcommand{\rBankPMcr}{404.4}
\newcommand{\rBankTDEVcr}{16.8}
\newcommand{\rBankEMlist}{RELY\,=\,VH (1.26), DATA\,=\,H (1.14), CPLX\,=\,VH (1.34), TIME\,=\,H (1.11), ACAP\,=\,H (0.85), PCAP\,=\,H (0.88), APEX\,=\,H (0.88)}
\newcommand{\rTelEAF}{1.290}
\newcommand{\rTelEexp}{1.08}

\newcommand{\rTelKSLOC}{20}
\newcommand{\rTelSF}{17}
\newcommand{\rTelPMnom}{96.4}
\newcommand{\rTelTDEVnom}{15.4}
\newcommand{\rTelPMcr}{137.9}
\newcommand{\rTelTDEVcr}{11.6}
\newcommand{\rTelEMlist}{RELY\,=\,VH (1.26), TIME\,=\,VH (1.29), CPLX\,=\,VH (1.34), ACAP\,=\,H (0.85), PCAP\,=\,H (0.88), APEX\,=\,H (0.88), TOOL\,=\,H (0.90)}
\newcommand{\rAirEAF}{1.602}
\newcommand{\rAirEexp}{1.09}

\newcommand{\rAirKSLOC}{35}
\newcommand{\rAirSF}{18}
\newcommand{\rAirPMnom}{227.0}
\newcommand{\rAirTDEVnom}{20.4}
\newcommand{\rAirPMcr}{324.6}
\newcommand{\rAirTDEVcr}{15.3}
\newcommand{\rAirEMlist}{RELY\,=\,H (1.10), DATA\,=\,VH (1.28), CPLX\,=\,VH (1.34), TIME\,=\,VH (1.29), ACAP\,=\,H (0.85), PCAP\,=\,H (0.88), APEX\,=\,H (0.88)}

\newcommand{\rBankSD}{228.5}
\newcommand{\rBankBrob}{695.1}
\newcommand{\rBankBrobRatio}{1.50}
\newcommand{\rBankJrobWC}{44.18}
\newcommand{\rBankJtrueStar}{37.76}
\newcommand{\rBankJtrueRob}{37.86}
\newcommand{\rBankRobPremium}{0.3}

\newcommand{\rTelBrob}{571.3}
\newcommand{\rTelBrobRatio}{1.28}
\newcommand{\rTelJrobWC}{32.30}
\newcommand{\rTelJtrueStar}{28.70}

\newcommand{\rAirBrob}{463.6}
\newcommand{\rAirBrobRatio}{1.34}
\newcommand{\rAirJrobWC}{27.28}
\newcommand{\rAirJtrueStar}{23.65}

\title{\textbf{Restraint-Induced Loss in Software-Intensive Enterprises:}\\[2pt]
\Large A Verified Technical--Economic Impact Model of Operational Downtime and Delayed Commercialization,\\ with Tail-Risk Security Sizing and Three Enterprise Case Studies}
\author{Mahdi Eslamimehr\thanks{The case-study organisations are anonymised; all case data are synthetic and calibrated to public benchmarks (\cref{sec:cases}).}\\[2pt]
\normalsize Quandary Peak Research}
\date{Technical paper, October 2026}
\hypersetup{pdfauthor={Mahdi Eslamimehr},pdftitle={Restraint-Induced Loss in Software-Intensive Enterprises: A Verified Technical-Economic Impact Model of Operational Downtime and Delayed Commercialization},pdfsubject={Technical-economic impact modelling},pdfkeywords={cost of downtime, dependability, real options, COCOMO II, CVaR, security sizing}}

\begin{document}
\maketitle

\begin{abstract}
\noindent When an enterprise must stop using a software component (for example, because of an interim court restraint in a technology patent dispute), the economic loss passes through a chain of technical mechanisms. Withdrawing the component exposes dependent capabilities. Fallback paths degrade throughput, correctness or feature coverage. Incident rates and recovery backlogs rise. Customers respond. Remediation programmes compete against the restraint clock, and planned launches are pushed back. This paper develops and verifies a \emph{technical--economic impact model} (TEIM) that quantifies this chain end to end as a probability distribution of present-value loss rather than as a point estimate. The technical layer combines (i)~a lattice-valued, fixed-point semantics of service dependencies with degraded fallbacks, (ii)~continuous-time Markov reliability models evaluated by uniformization and specified in continuous stochastic logic, and (iii)~a fluid queueing model of post-outage backlog. These feed (iv)~convex acute-outage loss functions with regulatory and contractual thresholds, (v)~compound-Poisson excess-incident processes, (vi)~survival-weighted functional, attrition and funding channels that provably avoid double counting, (vii)~a real-options valuation of forcibly delayed launches with an explicit closed form, and (viii)~a COCOMO~II--driven stochastic activity network for design-around and substitution strategies. We prove seven theorems and eight propositions. They cover existence and computability of degraded-capability semantics, causal locality, backlog convexity, a Jensen lower bound on the bias of deterministic estimates, exact channel partition, monotonicity of forced-delay loss and the bounds $0\le\ell(\Delta)\le\ell_{\mathrm{DCF}}(\Delta)$ for launches that are optimal immediately, and a critical-fractile rule for optimal security whose minimised expected cost equals $c\cdot\CVaR_{\alpha^\star}(L)$. A Python implementation is verified against independent numerics: the closed forms agree with quadrature to $\rVAcuteMaxErr$, uniformization with matrix exponentials to $\rVUnifErr$, and order-statistic intervals reach $\rVQcov\%$ empirical coverage. We apply the model, with $\rN$ common-random-number draws per strategy, to three anonymised enterprise platforms: a core-banking ledger, a 5G converged-charging function and an airline passenger-service system. Expected losses under the loss-minimising strategy are \musd{\rBankMean}, \musd{\rTelMean} and \musd{\rAirMean}, with 99th percentiles of \musd{\rBankPnn}, \musd{\rTelPnn} and \musd{\rAirPnn}. Conventional plug-in estimates fall at the $\rBankPlugRank$th, $\rTelPlugRank$th and $\rAirPlugRank$th percentiles of the simulated distributions. Variance is driven by customer response (bank), by a mix of customer response, remediation execution and launch value (operator), and by the restraint horizon (airline). The minimising strategy differs across platforms. The plug-in also misranks the responses: for the airline it selects degrade-and-wait, which the full model identifies as the worst option, and for the operator it selects the wrong credit-control failure policy. At a 2\% premium and a 30\% probability of wrongful restraint, the security implied by the critical-fractile rule is $\rTelBondRatio$--$\rBankBondRatio\times$ the expected loss. Posting only the expected loss would leave $\rTelUatEL$--$\rBankUatEL\%$ of it uncompensated in expectation.
\end{abstract}

\noindent\textbf{Keywords:} cost of downtime; performability; fixed-point semantics; continuous stochastic logic; fluid queues; real options; COCOMO~II; stochastic activity networks; conditional value-at-risk; global sensitivity analysis; difference-in-differences; security for interim restraints.

\noindent\textbf{JEL:} C15, C44, D81, G31, K41, L86. \quad \textbf{ACM CCS:} Software and its engineering $\to$ Software creation and management; Mathematics of computing $\to$ Probability and statistics.

\tableofcontents
\clearpage

\section{Introduction}\label{sec:intro}
Large enterprises run on platforms whose components are developed, licensed and patented by many parties. A core-banking ledger posts transactions in real time through a proprietary engine. A mobile operator's 5G charging function grants quota through a vendor-specific algorithm. An airline's passenger-service system (PSS) computes real-time availability for dynamic pricing and disruption recovery. When an external event forces the enterprise to stop using such a component for a period of uncertain length, the economic consequences run through technical mechanisms. These include dependency propagation, degraded fallbacks, raised incident rates, recovery backlogs and remediation engineering. Economists and finance professionals cannot evaluate these mechanisms without models of the software. Engineers rarely translate them into present values, tail risks and decision rules. The paper is written for readers at that intersection.

Our running scenario is an \emph{interim restraint}: an order that forbids use of a component (here, one alleged to practise a patent) until a later decision. The enterprise may later be found to have been restrained wrongfully. In common-law systems that use security undertakings, it then recovers from a security posted by the party that obtained the restraint. In United States federal practice, Rule~65(c) of the Federal Rules of Civil Procedure requires the movant to give security ``in an amount that the court considers proper to pay the costs and damages sustained by any party found to have been wrongfully enjoined or restrained'' \citep{frcp65}. Several circuits treat the bond as a ceiling on recovery \citep{buddy1976,nintendo1994,meadjohnson2000}. The \emph{distribution} of loss, and in particular its upper tail, therefore matters, not just its expectation. We do not analyse any jurisdiction, procedure or dispute. The model is a technical--economic object. Its outputs (a loss distribution, a decomposition by causal channel, and a decision-theoretic security rule) are useful whenever an ex-ante, auditable, tail-aware quantification of restraint-induced loss is needed.

\paragraph{Why existing practice is insufficient.} Practitioners typically multiply an industry survey figure for hourly downtime cost by an expected outage duration. They then add a static discounted-cash-flow (DCF) estimate of delayed revenue. Surveys are informative about magnitudes. For example, more than 90\% of mid-size and large firms report hourly downtime costs above \$300{,}000, and 41\% of large enterprises report \$1--5~million or more \citep{itic2024}. Uptime Institute reports that 54\% of significant outages cost more than \$100{,}000 and 20\% more than \$1~million \citep{uptime2025}. But the conventional calculation has five structural defects that this paper addresses formally:
\begin{enumerate}[label=(\roman*),leftmargin=2.2em,itemsep=1pt]
\item \emph{Linearity.} Outage cost is treated as linear in duration. Recovery backlogs make it convex, and contractual and regulatory thresholds make it discontinuous (\cref{thm:backlog,thm:jensen}).
\item \emph{Missing structure.} Survey averages ignore which capabilities depend on the withdrawn component and how much capability the fallback paths preserve (\cref{sec:lattice}).
\item \emph{Missing dynamics.} The length of degraded operation is the minimum of an exogenous restraint horizon and an endogenous remediation time. That remediation time has its own schedule risk, which is right-skewed (\cref{sec:remediation}).
\item \emph{Static valuation of delay.} DCF treatment of a postponed launch ignores managerial flexibility. It is an upper bound only when the launch is optimal immediately (\cref{thm:delay}).
\item \emph{No uncertainty quantification.} Point estimates cannot support tail-based decisions. Deterministic plug-ins are biased low under convexity and Fulkerson-type effects (\cref{thm:jensen,prop:fulkerson}).
\end{enumerate}

\paragraph{Contributions.}
\begin{enumerate}[label=\textbf{C\arabic*},leftmargin=2.4em,itemsep=1pt]
\item A formal problem statement. Restraint-induced loss is defined as a counterfactual present-value difference and decomposed into eight channels. Admissibility conditions are expressed in first-order and metric temporal logic (\cref{sec:formulation}).
\item A lattice-valued fixed-point semantics for service dependencies with degraded alternatives. It gives existence, finite convergence, cold-start versus warm-continuation semantics (least versus greatest fixed point), monotone comparative statics and a causal-locality theorem (\cref{thm:lattice,prop:monotone,thm:locality}).
\item Closed-form operational economics: fluid backlog area (\cref{thm:backlog}), lognormal expectations of piecewise convex/threshold loss with a Jensen lower bound (\cref{thm:jensen}), and Campbell-type moments of discounted excess-incident losses (\cref{prop:campbell}).
\item An exact partition theorem for functional, attrition and funding channels with an explicit approximation bound (\cref{thm:partition}).
\item A real-options treatment of forced delay. We prove $0\le\ell\le\ell_{\mathrm{DCF}}$ and monotonicity, and give a closed form for the delayed option value (\cref{thm:delay,prop:delayclosed}).
\item A remediation model coupling COCOMO~II effort and schedule equations with a stochastic activity network, plus a Fulkerson bias result (\cref{prop:fulkerson}).
\item A decision-theoretic security rule (critical fractile) with the identity $J^\star=c\cdot\CVaR_{\alpha^\star}(L)$, and additive Euler allocation of tail risk by channel (\cref{thm:bond,prop:euler}).
\item A verified, vectorised implementation (\cref{sec:impl}) and a full experimental study on three anonymised enterprise platforms, including global sensitivity analysis and an econometric calibration experiment (\cref{sec:cases,sec:experiments}).
\end{enumerate}

\paragraph{Organisation.} \Cref{sec:related} reviews related work. \Cref{sec:formulation} formalises the problem. \Cref{sec:model} develops the model and its theory. \Cref{sec:impl} describes the implementation. \Cref{sec:cases} presents the anonymised case studies, and \cref{sec:experiments} the experiments. \Cref{sec:discussion} discusses adversarial review, validity and limitations, and \cref{sec:conclusion} concludes. Longer proofs and notation are in the appendices.

\section{Background and Related Work}\label{sec:related}
\paragraph{Economics of IT downtime.} Survey-based estimates dominate practice. ITIC's 2024 survey reports that over 90\% of mid-size and large enterprises put the cost of one hour of downtime above \$300{,}000. It also reports that banking and finance are among the verticals where large-enterprise averages exceed \$5~million per hour, excluding litigation and penalties \citep{itic2024}. Uptime Institute's annual outage analyses report that 54\% of respondents' most recent significant outage cost more than \$100{,}000, with 16\% (2024) and 20\% (2025) above \$1~million \citep{uptime2024,uptime2025}. Splunk and Oxford Economics estimate that unplanned downtime costs the Global~2000 about \$400~billion a year, or 9\% of profits \citep{splunk2024}. At the firm level, event studies find significant negative abnormal returns around publicly reported IT failures \citep{bharadwaj2009,mackinlay1997}. Public incidents provide calibration anchors: the 2018 TSB migration \citep{fcatsb2022}, the 2012 RBS batch failure \citep{fcarbs2014}, the 2022 Rogers and 2024 AT\&T network outages \citep{fccatt2024}, and the 2016 and 2024 Delta and 2022 Southwest disruptions \citep{delta2024q3,southwest2022}. None of these sources supplies a structural model that maps a specific withdrawn component to a loss distribution.

\paragraph{Dependability, performability and formal verification.} We use the dependability taxonomy of \citet{avizienis2004}. Meyer's performability framework \citep{meyer1980} evaluates degradable systems by reward accumulated over operational states. We generalise it from numeric rewards to lattice-valued capability levels. Repairable redundancy is modelled by continuous-time Markov chains (CTMCs) \citep{trivedi2016}, solved by uniformization \citep{jensen1953,grassmann1977}. Properties are specified in continuous stochastic logic (CSL) \citep{aziz2000,baier2003}, which descends from PCTL \citep{hansson1994}, metric temporal logic \citep{koymans1990} and linear temporal logic \citep{pnueli1977}. The service-dependency semantics builds on the Knaster--Tarski and Kleene fixed-point theorems \citep{tarski1955,davey2002}, and on coherent-structure theory \citep{esary1963,barlow1975}.

\paragraph{Queueing and backlog.} Post-outage backlogs are modelled with deterministic fluid approximations \citep{newell1982,kleinrock1975}, consistent with Little's law \citep{little1961}. The resulting quadratic waiting area is the source of convexity in our acute loss.

\paragraph{Software cost and schedule estimation.} Design-around effort and schedule follow COCOMO~II.2000 \citep{boehm2000,boehm1981}, interpreted with the estimation-error evidence of \citet{jorgensen2007}. Large IT projects have fat-tailed overruns: in \citet{flyvbjerg2011}, one in six projects is a ``black swan'' with a cost overrun of about 200\% and a schedule overrun of about 70\%. Activity networks follow PERT \citep{malcolm1959,clark1962} with Monte Carlo evaluation \citep{vanslyke1963}. Their deterministic analysis is biased optimistic \citep{fulkerson1962}.

\paragraph{Real options.} Delayed commercialisation is valued as a forcibly postponed perpetual option to invest \citep{mcdonald1986,dixit1994,trigeorgis1996}, in the contingent-claims tradition of \citet{black1973,merton1973} and the growth-option view of \citet{myers1977}.

\paragraph{Risk measures, uncertainty quantification and sensitivity.} We use coherent risk measures \citep{artzner1999}, expected shortfall and its optimisation representation \citep{rockafellar2000,rockafellar2002,acerbi2002}, and Euler capital allocation \citep{mcneil2015}. Dependence is induced by a Gaussian copula \citep{sklar1959,nelsen2006} in the spirit of rank-correlation methods \citep{iman1982}. Variance-based global sensitivity analysis follows \citet{sobol2001,saltelli2008,saltelli2010}, implemented with SALib \citep{herman2017}. Quantile confidence intervals use order statistics \citep{david2003,serfling1980}. The security rule is a newsvendor critical fractile \citep{arrow1951}.

\paragraph{Customer economics and causal calibration.} Customer lifetime value follows \citet{gupta2004} and \citet{fader2005}. Behavioural parameters can be calibrated from post-event panels. Two-way fixed-effects estimators are biased under staggered adoption with dynamic effects \citep{goodmanbacon2021,dechaisemartin2020,sun2021}. Group-time estimators \citep{callaway2021} and synthetic controls \citep{abadie2003,abadie2010,abadie2021} avoid this bias. \Citet{roth2023} survey the area, and inference uses the bootstrap \citep{efron1979}.

\paragraph{Law and economics of interim restraints.} The error-cost logic of interim relief goes back to \citet{leubsdorf1978}. \textit{American Hospital Supply} expressed it as the condition $P\cdot H_p>(1-P)\cdot H_d$ \citep{americanhospital1986}. The patent hold-up literature \citep{lemley2007,shapiro2010} stresses that the threat of an injunction against a component can extract value disproportionate to the component's contribution. Courts have required evidence of proximately caused loss \citep{nokia2011} and have upheld very large securities where lost profits, lost share and relaunch costs were documented \citep{sanofi2006}. We use this literature only to motivate why a tail-aware, causally decomposed distribution is the right output. All modelling is jurisdiction-neutral.

\section{Problem Formulation}\label{sec:formulation}
\subsection{Platforms, capabilities and restraint scenarios}

\begin{definition}[Capability chain]\label{def:chain}
A \emph{capability chain} is a finite totally ordered set $\mathcal{C}=\{0=c_0<c_1<\dots<c_m=1\}\subset[0,1]$. A level $c\in\mathcal{C}$ is the fraction of nominal functional capability (throughput, correctness or feature coverage) that a service delivers. $(\mathcal{C},\le)$ is a complete lattice with $\wedge=\min$ and $\vee=\max$, and so is the product $\mathcal{C}^V$ under the pointwise order.
\end{definition}

\begin{definition}[Platform]\label{def:platform}
A \emph{platform} is a tuple $\mathcal{P}=(V,E,\{f_v\}_{v\in V},b,K)$ where:
\begin{itemize}[leftmargin=1.6em,itemsep=0pt]
\item $V$ is a finite set of services, and $E\subseteq V\times V$ contains $(u,v)$ iff $v$ depends on $u$. We write $\mathrm{dep}(v)=\{u:(u,v)\in E\}$.
\item $f_v:\mathcal{C}^{\mathrm{dep}(v)}\to\mathcal{C}$ is an aggregation operator generated by the grammar
\[ f::= 1 \mid s_u \mid f\wedge f \mid f\vee f \mid \mathrm{kth}_k(f,\dots,f)\mid q\wedge f,\qquad q\in\mathcal{C}, \]
where $\mathrm{kth}_k$ returns the $k$-th largest argument. The capped alternative $\mathrm{alt}_{q}(s_{u_1},\dots,s_{u_n})=\bigvee_i (q_i\wedge s_{u_i})$ is the derived form used for fallbacks.
\item $b\in\mathcal{C}^V$ is a vector of exogenous \emph{bases} (administrative caps), and $K\subseteq V$ is the set of revenue-bearing \emph{capabilities}.
\end{itemize}
\end{definition}

The interpretation is as follows. $f_v=\bigwedge_{u}s_u$ means that $v$ needs all of its dependencies (series logic). $\bigvee$ means any one suffices (parallel logic), and $\mathrm{kth}_k$ means $k$-of-$n$ redundancy. $\mathrm{alt}_q$ says that $v$ is delivered by its primary path at full capability, or by fallback path $i$ at a capped level $q_i<1$. For example, a stand-in processor might deliver $85\%$ of available-balance accuracy.

\begin{definition}[Restraint scenario and strategies]\label{def:scenario}
A \emph{restraint scenario} is a tuple $\mathcal{R}=(r,T_R,\mathcal{S})$. Here $r\in V$ is the restrained component, whose base is forced to $b_r=0$ during $[0,T_R)$. $T_R>0$ is a random \emph{restraint horizon}, and $\mathcal{S}$ is a finite set of response strategies. A strategy $s\in\mathcal{S}$ induces a random remediation completion time $T_d(s)\in(0,\infty]$, after which a non-restrained substitute restores the but-for capability levels. It also induces a remediation cost process. The \emph{effective degraded period} is
\begin{equation}\label{eq:tau}
\tau_s \;=\; T_R\wedge T_d(s).
\end{equation}
We study four canonical strategies. $S_0$ (\emph{degrade-and-wait}) has $T_d=\infty$. $S_1$ (\emph{nominal design-around}) builds a non-practising alternative at nominal COCOMO schedule. $S_2$ (\emph{crashed design-around}) is $S_1$ at 75\% schedule compression. $S_3$ (\emph{substitution}) licenses a third-party module and integrates it.
\end{definition}

\Cref{fig:timeline} shows the timeline. At $t=0$ the component is withdrawn and a switch-over outage of random duration $D_0$ may occur. During $[0,\tau_s)$ the platform runs on fallbacks with raised incident intensity and reduced capability. A launch planned for $t_L$ that depends on $r$ is held until $T_R\wedge T_d$.

\begin{figure}[t]
\centering
\begin{tikzpicture}[x=1.15cm,y=1cm,>=Latex,font=\footnotesize]
\draw[->,thick] (0,0) -- (13.2,0) node[right]{$t$ (years)};
\fill[red!12] (0,0.05) rectangle (8.4,1.25);
\fill[green!10] (8.4,0.05) rectangle (12.8,1.25);
\node[align=center] at (4.2,1.55) {degraded operation on fallbacks: $\ell^{\mathcal{R}}_v<\ell^{\mathrm{bf}}_v$, incident rate $m\lambda_{\mathrm{bf}}$};
\node[align=center] at (9.75,1.55) {restored capability};
\draw[thick,red!70!black] (0,0) -- (0,2.2) node[above,align=center]{restraint, $t=0$\\withdraw $r$};
\fill[red!60] (0,0.05) rectangle (0.35,1.25);
\node[red!60!black,anchor=west] at (0.38,1.02) {\scriptsize $\leftarrow$ switch-over outage $D_0$};
\foreach \x/\h in {1.6/0.6,2.9/0.75,4.4/0.5,5.1/0.8,6.8/0.7,7.6/0.55} {\draw[red!70!black,line width=1.2pt] (\x,0.05) -- (\x,\h);}
\node[red!70!black] at (5.3,-0.3) {\scriptsize excess incidents $N\sim\Pois(\lambda_{\mathrm{ex}}\tau)$, durations $D_i\sim\LN$};
\draw[thick,blue!70!black] (2.0,0) -- (2.0,-0.8);
\node[blue!70!black,anchor=east] at (1.95,-0.62) {\scriptsize $t_L$ planned launch};
\draw[thick,green!40!black] (8.4,0) -- (8.4,2.2) node[above,align=center]{$T_d(s)$\\remediation done};
\draw[thick,dashed] (11.2,0) -- (11.2,2.2) node[above]{$T_R$};
\draw[decorate,decoration={brace,amplitude=4pt,mirror},blue!60!black] (2.0,-0.85) -- (8.4,-0.85) node[midway,below=5pt]{\scriptsize forced delay $\Delta=(\tau_s-t_L)^+$};
\draw[decorate,decoration={brace,amplitude=5pt,mirror}] (0,-1.55) -- (8.4,-1.55) node[midway,below=6pt]{$\tau_s=T_R\wedge T_d(s)$};
\end{tikzpicture}
\caption{Restraint timeline with remediation completed before the restraint ends ($T_d<T_R$). Shown are the switch-over outage, the excess-incident process during the degraded period, the forced launch delay $\Delta$ and the effective degraded period $\tau_s$.}
\label{fig:timeline}
\end{figure}
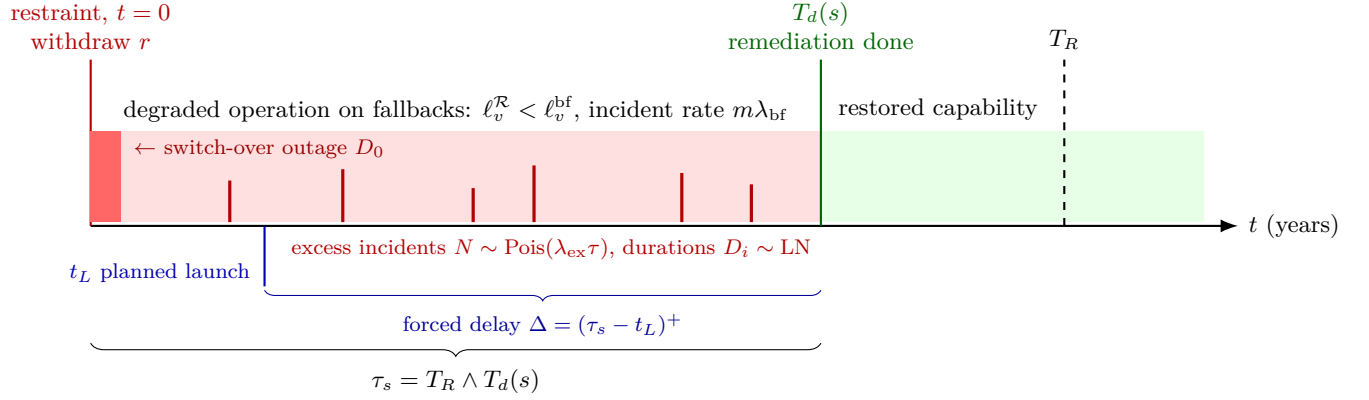

\subsection{Loss as a counterfactual present value}
Fix a filtered probability space $(\Omega,\mathcal{F},(\mathcal{F}_t),\Prob)$ carrying the restrained and the but-for (unrestrained) worlds. They are coupled by common exogenous randomness (\cref{ass:coupling}). Let $\Pi^{\mathrm{bf}}_t$ and $\Pi^{s}_t$ be the enterprise's net cash-flow rates in the two worlds under strategy $s$, and let $\rho>0$ be the continuously compounded cost of capital.

\begin{definition}[Restraint-induced loss]\label{def:loss}
\begin{equation}\label{eq:loss}
L(s)\;=\;\int_0^\infty e^{-\rho t}\left(\Pi^{\mathrm{bf}}_t-\Pi^{s}_t\right)\dd t \;=\;\sum_{k\in\mathcal{K}}L_k(s),
\end{equation}
\[\mathcal{K}=\{\mathrm{ac},\mathrm{inc},\mathrm{fun},\mathrm{ops},\mathrm{att},\mathrm{fund},\mathrm{rem},\mathrm{del}\}.\]
The channels are: acute switch-over outage; excess incidents; functional degradation; workaround operations; customer attrition; deposit funding; remediation cost; and delayed launch. We call $L^\star=\min_{s\in\mathcal{S}}\E[L(s)]$ the \emph{mitigated} expected loss, and $s^\star$ the loss-minimising strategy.
\end{definition}

Taking the minimum over strategies encodes the \emph{duty to mitigate}: the loss attributed to the restraint is the loss under the best available response, not under an arbitrary one.

\subsection{Admissibility logic}
We use a many-sorted first-order language with sorts \textsf{Svc}, \textsf{Chan} and \textsf{Atom}. Its predicates are: reachability $u\pathto v$ (the transitive closure of $E$); $\mathrm{Lvl}^{w}(v,c)$ (``in world $w\in\{\mathrm{bf},\mathcal{R}\}$, $v$ has level $c$''); $\mathrm{Aff}(k,v)$ (channel $k$ is computed from the level of $v$); and $\mathrm{Atoms}(k,a)$ (loss atom $a$ is assigned to channel $k$). Loss atoms are elementary customer--time or event cells (\cref{sec:partition}).

\begin{definition}[Admissible channel]\label{def:adm}
\begin{multline}\label{eq:adm}
\Adm(k)\;:\Longleftrightarrow\; \Big[\exists v\in K\;\big(r\pathto v\;\wedge\;\exists c,c'\,(\mathrm{Lvl}^{\mathcal{R}}(v,c)\wedge\mathrm{Lvl}^{\mathrm{bf}}(v,c')\wedge c<c')\\
\wedge\;\mathrm{Aff}(k,v)\big)\;\vee\;k\in\{\mathrm{ac},\mathrm{inc},\mathrm{rem}\}\Big]\\
\wedge\;\forall k'\neq k\;\forall a\;\neg\big(\mathrm{Atoms}(k,a)\wedge\mathrm{Atoms}(k',a)\big).
\end{multline}
\end{definition}

The first conjunct is a \emph{causal-path} condition. A functional, attrition, funding or delay channel is admissible only if it is computed from a capability that is reachable from $r$ and actually degraded. Acute, incident and remediation channels are admissible directly, because they are caused by the withdrawal itself. The second conjunct forbids double counting. \Cref{thm:locality} turns the first conjunct into a decidable graph test. \Cref{thm:partition} establishes the second for our channel definitions.

\subsection{Temporal specifications}
Operational behaviour is described by timed traces $\sigma:\R_{\ge0}\to 2^{AP}\times\mathcal{C}^V$ over atomic propositions $AP\supseteq\{\mathsf{active}(r),\mathsf{up},\mathsf{restored}\}$. Using metric temporal logic \citep{koymans1990} with bounded operators $\Always_I$ and $\Eventually_I$, we require:
\begin{align}
\varphi_{\mathrm{comp}}&:=\Always_{[0,T_R)}\,\neg\mathsf{active}(r) &&\text{(compliance)}\label{eq:mtl1}\\
\varphi_{\mathrm{deg}}&:=\Always_{[0,\tau_s)}\textstyle\bigwedge_{v\in K}\big(\mathrm{lvl}_v\le\ell^{\mathcal{R}}_v\big) &&\text{(degraded envelope)}\label{eq:mtl2}\\
\varphi_{\mathrm{rem}}&:=\Eventually_{[0,T_d(s)]}\,\mathsf{restored}\;\wedge\;\Always_{[T_d(s),\infty)}\textstyle\bigwedge_{v\in K}\big(\mathrm{lvl}_v=\ell^{\mathrm{bf}}_v\big) &&\text{(remediation)}\label{eq:mtl3}\\
\varphi_{\mathrm{SLA}}^{\theta}&:=\Always\big(\neg\mathsf{up}\rightarrow\Eventually_{[0,\theta]}\mathsf{up}\big) &&\text{(restoration within $\theta$ hours)}\label{eq:mtl4}
\end{align}
A violation of $\varphi_{\mathrm{SLA}}^{\theta_j}$ triggers the threshold penalty $\kappa_j$ in the acute loss function (\cref{sec:acute}). The reliability layer evaluates the continuous stochastic logic (CSL) property
\begin{equation}\label{eq:csl}
\Phi_{\mathrm{rel}}(t,p):=\mathcal{P}_{\ge p}\big[\Always^{\le t}\,\mathsf{up}\big],
\end{equation}
and the steady-state operator $\mathcal{S}_{\ge a}[\mathsf{up}]$ \citep{aziz2000,baier2003} for the but-for and fallback architectures. The difference between the two drives the excess-incident channel (\cref{sec:ctmc}).

\subsection{Assumptions}
\begin{assumption}[Coupling]\label{ass:coupling}
Exogenous demand, prices and but-for failure processes are identical in the two worlds. The worlds differ only through mechanisms that originate at $r$.
\end{assumption}
\begin{assumption}[Piecewise stationarity]\label{ass:stationary}
Given parameters $\theta$, rates and levels are constant on $[0,\tau_s)$, and they revert to but-for values at $\tau_s$.
\end{assumption}
\begin{assumption}[Marked Poisson incidents]\label{ass:poisson}
Given $\theta$ and $\tau_s$, the fallback incident process is a Poisson process with rate $m\lambda_{\mathrm{bf}}$, $m\ge1$. Incident durations are i.i.d.\ and independent of the arrival times.
\end{assumption}
\begin{assumption}[Proportional-hazard customer response]\label{ass:hazard}
Excess churn of a customer in segment $j$ has cumulative hazard $\Lambda_j=\beta^{\mathrm{out}}_jH+\beta^{\mathrm{deg}}_jx_j\tau_s$. Here $H$ is the number of outage hours experienced and $x_j$ is the capability loss relevant to the segment.
\end{assumption}
\begin{assumption}[Project-value dynamics]\label{ass:gbm}
Under a pricing measure $\Q$, the value of a planned launch follows $\dd V_t=(r-\delta)V_t\dd t+\sigma V_t\dd W_t$, with $r>0$, $\delta>0$ and $\sigma>0$.
\end{assumption}

\section{The Technical--Economic Impact Model}\label{sec:model}
\Cref{fig:framework} shows the model's architecture. A technical layer maps the restraint into capability levels, incident intensities and backlog dynamics. An economic layer maps these into channel cash flows. A remediation layer determines $\tau_s$. An uncertainty layer propagates parameter and event uncertainty into a loss distribution, which a decision layer converts into risk measures and a security rule.

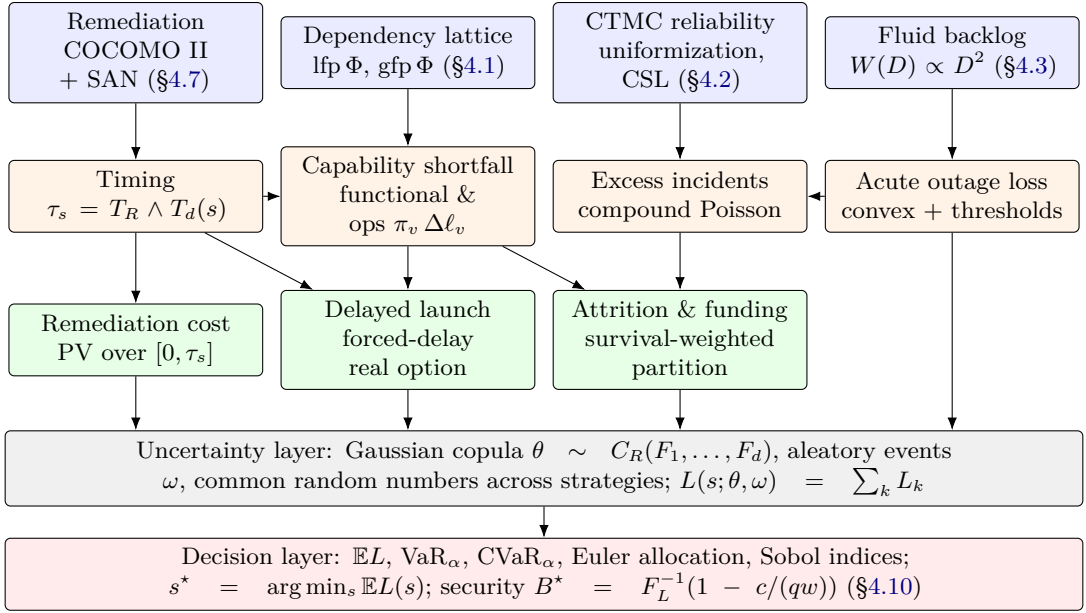
\begin{figure}[t]
\centering
\begin{tikzpicture}[font=\footnotesize,>=Latex,
  box/.style={draw,rounded corners=2pt,align=center,minimum height=0.95cm,text width=3.1cm,fill=#1},
  box/.default=white]
\node[box=blue!8] (Rm) at (0,0) {Remediation\\COCOMO II + SAN (\S\ref{sec:remediation})};
\node[box=blue!8] (G) at (3.6,0) {Dependency lattice\\$\lfp\Phi$, $\gfp\Phi$ (\S\ref{sec:lattice})};
\node[box=blue!8] (M) at (7.2,0) {CTMC reliability\\uniformization, CSL (\S\ref{sec:ctmc})};
\node[box=blue!8] (Qu) at (10.8,0) {Fluid backlog\\$W(D)\propto D^2$ (\S\ref{sec:acute})};
\node[box=orange!10] (T) at (0,-1.9) {Timing\\$\tau_s=T_R\wedge T_d(s)$};
\node[box=orange!10] (F) at (3.6,-1.9) {Capability shortfall\\functional \& ops $\pi_v\,\Delta\ell_v$};
\node[box=orange!10] (I) at (7.2,-1.9) {Excess incidents\\compound Poisson};
\node[box=orange!10] (A) at (10.8,-1.9) {Acute outage loss\\convex + thresholds};
\node[box=green!10] (Rc) at (0,-3.8) {Remediation cost\\PV over $[0,\tau_s]$};
\node[box=green!10] (D) at (3.6,-3.8) {Delayed launch\\forced-delay real option};
\node[box=green!10] (C) at (7.2,-3.8) {Attrition \& funding\\survival-weighted partition};
\node[box=gray!12,text width=14.0cm] (U) at (5.4,-5.5) {Uncertainty layer: Gaussian copula $\theta\sim C_R(F_1,\dots,F_d)$, aleatory events $\omega$, common random numbers across strategies; $L(s;\theta,\omega)=\sum_k L_k$};
\node[box=red!8,text width=14.0cm] (Dz) at (5.4,-6.9) {Decision layer: $\E L$, $\VaR_\alpha$, $\CVaR_\alpha$, Euler allocation, Sobol indices; $s^\star=\arg\min_s\E L(s)$; security $B^\star=F_L^{-1}(1-c/(qw))$ (\S\ref{sec:bond})};
\draw[->] (Rm) -- (T); \draw[->] (G) -- (F); \draw[->] (M) -- (I); \draw[->] (Qu) -- (A);
\draw[->] (T) -- (F); \draw[->] (A) -- (I);
\draw[->] (T) -- (Rc); \draw[->] (T) -- (D); \draw[->] (F) -- (D); \draw[->] (F) -- (C); \draw[->] (I) -- (C);
\draw[->] (A.south) -- ++(0,-1.2) -- ($(U.north)+(5.4,0)$);
\foreach \n in {Rc,D,C} {\draw[->] (\n) -- (\n |- U.north);}
\draw[->] (U) -- (Dz);
\end{tikzpicture}
\caption{Architecture of the technical--economic impact model (TEIM). Blue: technical layer. Orange: operational--economic layer. Green: market and remediation layer. Grey and red: uncertainty and decision layers. The degraded period $\tau_s$ enters every flow channel; for clarity only its direct arrows are drawn.}
\label{fig:framework}
\end{figure}

\subsection{Lattice semantics of dependency propagation}\label{sec:lattice}
Given bases $b$ and fallback caps $q$, define the platform operator $\Phi_{b,q}:\mathcal{C}^V\to\mathcal{C}^V$ by
\begin{equation}\label{eq:phi}
\Phi_{b,q}(s)_v \;=\; b_v\wedge f_v(s).
\end{equation}
A \emph{consistent capability state} is a fixed point $s=\Phi(s)$. Restraint sets $b_r=0$. We write $\ell^{\mathrm{bf}}$ and $\ell^{\mathcal{R}}$ for the fixed points in the but-for and restrained configurations, and $\Delta\ell_v=\ell^{\mathrm{bf}}_v-\ell^{\mathcal{R}}_v$ for the capability loss.

\begin{theorem}[Existence, computability and semantics of capability states]\label{thm:lattice}
Let $h=|V|(|\mathcal{C}|-1)$ be the height of $\mathcal{C}^V$.
\begin{enumerate}[label=(\alph*),itemsep=0pt]
\item $\Phi$ is monotone, and its fixed points form a non-empty complete lattice. In particular, a least fixed point $\lfp\Phi$ and a greatest fixed point $\gfp\Phi$ exist.
\item The Kleene iterates $s^{(0)}=\bot$, $s^{(n+1)}=\Phi(s^{(n)})$ form an ascending chain that reaches $\lfp\Phi$ after at most $h$ strict increases. Dually, iteration from $\top$ descends to $\gfp\Phi$ after at most $h$ strict decreases.
\item If $(V,E)$ is acyclic, $\Phi$ has a unique fixed point, $\lfp\Phi=\gfp\Phi$. It is reached from any initial state after at most $d+1$ iterations, where $d$ is the length of the longest path.
\end{enumerate}
\end{theorem}
\begin{proof}
(a) Every production of the grammar in \cref{def:platform} is monotone in each argument, and so are compositions of monotone maps and the meet with the constant $b_v$. Hence $s\le s'\Rightarrow\Phi(s)\le\Phi(s')$. $\mathcal{C}^V$ is a finite, hence complete, lattice, so the Knaster--Tarski theorem \citep{tarski1955} applies.
(b) $s^{(0)}=\bot\le s^{(1)}$. If $s^{(n)}\le s^{(n+1)}$, monotonicity gives $s^{(n+1)}=\Phi(s^{(n)})\le\Phi(s^{(n+1)})=s^{(n+2)}$. Each strict increase raises at least one coordinate by at least one step of $\mathcal{C}$, so there are at most $h$ strict increases. The stationary value $s^\ast$ is a fixed point. If $x=\Phi(x)$ is any fixed point, then $\bot\le x$ and induction give $s^{(n)}\le\Phi^n(x)=x$, so $s^\ast\le x$ and $s^\ast$ is least. The argument for $\gfp\Phi$ is dual.
(c) Order $V$ topologically and let $\mathrm{depth}(v)$ be the length of the longest path ending at $v$. Coordinates with depth $0$ are leaves, with $\Phi(s)_v=b_v$ after one iteration. If all coordinates of depth $<k$ are fixed after $k$ iterations, a coordinate of depth $k$ depends only on them and is fixed after $k+1$ iterations, independently of the initial state. Hence every fixed point equals the common limit.
\end{proof}

\begin{remark}[Cold start versus warm continuation]
Cyclic dependencies arise in practice: an authentication service may read its configuration from a store that itself requires authentication. In such cases $\lfp\Phi$ is the \emph{cold-start} semantics, where a cycle without an external bootstrap cannot come up. $\gfp\Phi$ is the \emph{warm-continuation} semantics, where a running cycle sustains itself. Withdrawing a component from a running platform often requires restarts, so the lfp is the prudent choice for loss quantification. Both are computed (\cref{alg:kleene}). The three case-study graphs are acyclic, so the two coincide by \cref{thm:lattice}(c).
\end{remark}

\begin{proposition}[Monotone comparative statics]\label{prop:monotone}
If $b\le b'$ and $q\le q'$ pointwise, then $\lfp\Phi_{b,q}\le\lfp\Phi_{b',q'}$ and $\gfp\Phi_{b,q}\le\gfp\Phi_{b',q'}$. Consequently, every functional channel of the form $L_{\mathrm{fun}}=\sum_{v\in K}\pi_v\,\Delta\ell_v\,a(\tau)$ with $\pi_v\ge0$ is non-increasing in the fallback caps.
\end{proposition}
\begin{proof}
For every $s$, $\Phi_{b,q}(s)\le\Phi_{b',q'}(s)$, since caps enter only through meets and bases only through $b_v\wedge\cdot$. By induction, $\Phi^n_{b,q}(\bot)\le\Phi^n_{b',q'}(\bot)$, using monotonicity of $\Phi_{b',q'}$. Pass to the limit by \cref{thm:lattice}(b). The gfp case is dual. The consequence follows because $\Delta\ell_v=\ell^{\mathrm{bf}}_v-\ell^{\mathcal{R}}_v$ and $\ell^{\mathrm{bf}}_v$ does not depend on the fallback caps when the primary path is available.
\end{proof}

\begin{theorem}[Causal locality]\label{thm:locality}
Let $\mathrm{Reach}(r)=\{v: r\pathto v\}\cup\{r\}$. For every $v\notin\mathrm{Reach}(r)$, both $\lfp\Phi$ and $\gfp\Phi$ at $v$ are invariant to the value of $b_r$. Hence, in \eqref{eq:adm}, any channel computed only from capabilities outside $\mathrm{Reach}(r)$ violates the causal-path conjunct, and admissibility of a channel is decidable in $O(|V|+|E|)$ time.
\end{theorem}
\begin{proof}
Let $U=V\setminus\mathrm{Reach}(r)$. If $v\in U$ and $u\in\mathrm{dep}(v)$, then $u\in U$: otherwise $r\pathto u\to v$ would put $v\in\mathrm{Reach}(r)$. Hence the restriction $\Phi(s)|_U$ depends only on $s|_U$ and on $b|_U$, which does not contain $b_r$. The Kleene iterates restricted to $U$ are therefore identical for all values of $b_r$, and so are their limits. Reachability is computed by breadth-first search.
\end{proof}

\begin{example}[Cycle with break-glass bootstrap]\label{ex:toy}
Consider the seven services $V=\{\mathrm{DNS},\mathrm{AUTH},\mathrm{CFG},\mathrm{BOOT},$ $\mathrm{API},\mathrm{APP},\mathrm{CACHE}\}$ with aggregation operators
\begin{gather*}
f_{\mathrm{AUTH}}=s_{\mathrm{DNS}}\wedge s_{\mathrm{CFG}},\qquad f_{\mathrm{CFG}}=s_{\mathrm{AUTH}}\vee s_{\mathrm{BOOT}},\qquad f_{\mathrm{API}}=s_{\mathrm{AUTH}}\wedge s_{\mathrm{CFG}},\\
f_{\mathrm{APP}}=(1\wedge s_{\mathrm{API}})\vee(0.5\wedge s_{\mathrm{CACHE}}).
\end{gather*}
The other services are leaves (\cref{fig:toy}). With the break-glass bootstrap disabled ($b_{\mathrm{BOOT}}=0$), the cycle AUTH$\leftrightarrow$CFG admits three fixed points. The lfp is $\ell_{\mathrm{APP}}=0.5$ (cold start: degraded read-only service from the cache). The gfp is $\ell_{\mathrm{APP}}=1$ (warm continuation). \Cref{tab:v-fix} confirms by exhaustive enumeration that Kleene iteration returns the least and greatest fixed points.
\end{example}

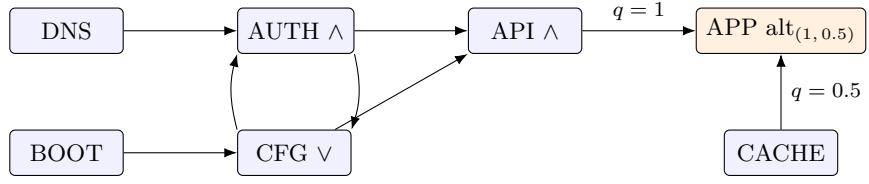
\begin{figure}[t]
\centering
\begin{tikzpicture}[font=\footnotesize,>=Latex,node distance=1.0cm and 1.5cm,
  svc/.style={draw,rounded corners=2pt,minimum width=1.5cm,minimum height=0.6cm,fill=blue!6}]
\node[svc] (dns) {DNS};
\node[svc,below=of dns] (boot) {BOOT};
\node[svc,right=of dns] (auth) {AUTH $\wedge$};
\node[svc,right=of boot] (cfg) {CFG $\vee$};
\node[svc,right=of auth] (api) {API $\wedge$};
\node[svc,right=of api,fill=orange!12] (app) {APP alt$_{(1,\,0.5)}$};
\node[svc,below=of app] (cache) {CACHE};
\draw[->] (dns) -- (auth); \draw[->] (boot) -- (cfg);
\draw[->] (cfg.north west) to[bend left=15] (auth.south west);
\draw[->] (auth.south east) to[bend left=15] (cfg.north east);
\draw[->] (auth) -- (api); \draw[->] (cfg) -- (api.south west);
\draw[->] (api) -- node[above]{\scriptsize $q=1$} (app); \draw[->] (cache) -- node[right]{\scriptsize $q=0.5$} (app);
\end{tikzpicture}
\caption{Toy platform of \cref{ex:toy}, with a dependency cycle AUTH$\leftrightarrow$CFG and a capped fallback for APP.}
\label{fig:toy}
\end{figure}

\subsection{Reliability layer: CTMC models and the incident multiple}\label{sec:ctmc}
Each architecture (but-for and fallback) is modelled as a series system of independent blocks. A hardware or instance block is a $k$-of-$n$ repairable cluster: a birth--death CTMC on the number $j\in\{0,\dots,n\}$ of failed nodes, with generator
\begin{equation}\label{eq:bd}
q_{j,j+1}=(n-j)\lambda,\qquad q_{j,j-1}=\min(j,\varsigma)\,\mu,\qquad \mathcal{U}=\{j:n-j\ge k\},
\end{equation}
where $\varsigma$ is the number of repair crews. Software and configuration common-mode failures are $1$-of-$1$ blocks. Their rate scales with operational change intensity, and manual workarounds raise it in fallback mode. For an irreducible generator $\mathbf{Q}$ with stationary law $\bm\pi$ ($\bm\pi\mathbf{Q}=0$, $\bm\pi\mathbf{1}=1$), the availability, the Markov failure frequency and the mean down time are
\begin{equation}\label{eq:avail}
A=\sum_{j\in\mathcal{U}}\pi_j,\qquad f=\sum_{i\in\mathcal{U}}\sum_{j\notin\mathcal{U}}\pi_i q_{ij},\qquad \mathrm{MDT}=\frac{1-A}{f}.
\end{equation}
For independent blocks in series, $A=\prod_iA_i$ and $f=\sum_i f_i\prod_{j\ne i}A_j$. Both identities are exact under independence, because the system leaves $\mathcal{U}$ exactly when one block fails while all others are up. The CSL interval reliability $R(t)=\Prob(\Always^{\le t}\mathsf{up})$ is obtained by making the down states absorbing ($\mathbf{Q}^{\mathcal{U}}$). It is evaluated by uniformization \citep{jensen1953,grassmann1977}:
\begin{equation}\label{eq:unif}
\bm\pi(t)=\sum_{n=0}^{\infty}e^{-\Lambda t}\frac{(\Lambda t)^n}{n!}\,\bm\pi(0)\mathbf{P}^n,\qquad \mathbf{P}=\mathbf{I}+\mathbf{Q}/\Lambda,\quad \Lambda\ge\max_i|q_{ii}|.
\end{equation}

\begin{proposition}[Uniformization error and cumulative rewards]\label{prop:unif}
Truncating \eqref{eq:unif} at $N$ terms gives $\|\bm\pi(t)-\bm\pi_N(t)\|_1\le 1-\sum_{n\le N}e^{-\Lambda t}(\Lambda t)^n/n!$. The interval occupation measure satisfies $\int_0^t\bm\pi(u)\dd u=\Lambda^{-1}\sum_{n\ge0}\Prob(N_{\Lambda t}>n)\,\bm\pi(0)\mathbf{P}^n$, with $N_{\Lambda t}\sim\Pois(\Lambda t)$.
\end{proposition}
\begin{proof}
$\mathbf{P}$ is stochastic, so $\|\bm\pi(0)\mathbf{P}^n\|_1=1$ and the tail of the series is bounded by the Poisson tail mass. Integrating term by term, $\int_0^t e^{-\Lambda u}(\Lambda u)^n/n!\dd u=\Lambda^{-1}\Prob(N_{\Lambda t}>n)$ by the Gamma--Poisson identity.
\end{proof}

The \emph{architecture-implied incident multiple} is $\hat m=f^{\mathrm{fb}}/f^{\mathrm{bf}}$. It is used as the prior median of $m$ in \cref{ass:poisson}. The excess intensity attributable to the restraint is $\lambda_{\mathrm{ex}}=(m-1)^+\lambda_{\mathrm{bf}}$. By the superposition property of Poisson processes and \cref{ass:coupling}, the fallback process is the superposition of the but-for process and an independent excess process. Only the latter is attributable.

\subsection{Acute outage economics}\label{sec:acute}
\begin{theorem}[Fluid backlog convexity]\label{thm:backlog}
Let demand arrive at rate $\lambda$. During an outage of length $D$, a fraction $\psi\in(0,1]$ of arrivals is deferred (queued) and $1-\psi$ is lost. After restoration, service capacity is $c>\lambda$. In the fluid limit, the backlog peaks at $Q^\ast=\psi\lambda D$, drains in time $T_c=\psi\lambda D/(c-\lambda)$, and accumulates the waiting area
\begin{equation}\label{eq:W}
W(D)=\int_0^{D+T_c}Q(t)\dd t=\tfrac12\,\psi\lambda D^2\Big(1+\frac{\psi\lambda}{c-\lambda}\Big),
\end{equation}
which is strictly convex in $D$. By Little's law \citep{little1961}, the mean delay per deferred transaction is $W/(\psi\lambda D)=\tfrac{D}{2}\big(1+\psi\lambda/(c-\lambda)\big)$, which is linear and increasing in $D$.
\end{theorem}
\begin{proof}
For $t\le D$, $\dot Q=\psi\lambda$, so $Q(t)=\psi\lambda t$. For $t>D$, new arrivals are served on arrival and the backlog drains at net rate $c-\lambda$, so $Q(t)=\psi\lambda D-(c-\lambda)(t-D)$ until it reaches zero at $D+T_c$. The area is that of two triangles with common height $\psi\lambda D$ and bases $D$ and $T_c$, which gives \eqref{eq:W}. $W''=\psi\lambda(1+\psi\lambda/(c-\lambda))>0$.
\end{proof}

With a delay cost $w$ per transaction-hour, the backlog contributes $bD^2$ with $b=w\cdot\tfrac12\psi\lambda(1+\psi\lambda/(c-\lambda))$. Writing $c=\kappa\lambda$ with catch-up multiple $\kappa>1$, this is $b=\tfrac12 w\psi\lambda\,(1+\psi/(\kappa-1))$. The \emph{acute loss function} of an outage of length $D$ (hours) is
\begin{equation}\label{eq:acute}
L_{\mathrm{ac}}(D)=K_0\ind\{D>0\}+aD+bD^2+\sum_{j}\kappa_j\ind\{D>\theta_j\}+\sum_{m}h_m\,(D-\eta_m)^+ .
\end{equation}
Here $K_0$ is the fixed response cost and $a$ the linear hourly loss. The $\kappa_j$ are threshold penalties: SLA credits, regulatory reporting or enforcement exposure, and goodwill credits (see \eqref{eq:mtl4}). The hinge terms $h_m$ capture losses that start only after a delay, such as customer redress after $\eta=4$\,h or cascade cancellations after $3$\,h.

\begin{theorem}[Lognormal expectation and Jensen bias of plug-in estimates]\label{thm:jensen}
Let $D\sim\LN(\mu,s^2)$ and $d_x=(\mu-\ln x)/s$. Then
\begin{equation}\label{eq:acuteexp}
\E L_{\mathrm{ac}}(D)=K_0+a\,e^{\mu+s^2/2}+b\,e^{2\mu+2s^2}+\sum_j\kappa_j\,\Phi(d_{\theta_j})+\sum_m h_m\Big[e^{\mu+s^2/2}\Phi(d_{\eta_m}+s)-\eta_m\Phi(d_{\eta_m})\Big].
\end{equation}
For any integrable $D\ge0$ with $\Prob(D>0)=1$,
\begin{equation}\label{eq:jensengap}
\E L_{\mathrm{ac}}(D)-L_{\mathrm{ac}}(\E D)\;\ge\; b\Var(D)+\sum_j\kappa_j\big[\Prob(D>\theta_j)-\ind\{\E D>\theta_j\}\big].
\end{equation}
Moreover, $L_{\mathrm{ac}}(\mathrm{med}\,D)\le L_{\mathrm{ac}}(\E D)$ whenever $\mathrm{med}\,D\le\E D$, which always holds for the lognormal. Hence a plug-in estimate evaluated at the median duration is biased downward by at least the right-hand side of \eqref{eq:jensengap}, whenever that is non-negative.
\end{theorem}
\begin{proof}
For the lognormal, $\E D=e^{\mu+s^2/2}$, $\E D^2=e^{2\mu+2s^2}$ and $\Prob(D>x)=\Phi(d_x)$. The partial first moment is $\E[D\ind\{D>\eta\}]=e^{\mu+s^2/2}\Phi(d_\eta+s)$, obtained by completing the square in the Gaussian integral. Then $\E(D-\eta)^+=\E[D\ind\{D>\eta\}]-\eta\Prob(D>\eta)$, which gives \eqref{eq:acuteexp}. For \eqref{eq:jensengap}, split $L_{\mathrm{ac}}=g+bD^2+\sum_j\kappa_j\ind\{D>\theta_j\}$ with $g(D)=K_0+aD+\sum_mh_m(D-\eta_m)^+$, which is convex on $(0,\infty)$. Jensen gives $\E g(D)\ge g(\E D)$. Also $\E[bD^2]-b(\E D)^2=b\Var D$, and the step terms contribute exactly $\kappa_j[\Prob(D>\theta_j)-\ind\{\E D>\theta_j\}]$. The final claim holds because $L_{\mathrm{ac}}$ is non-decreasing, and for the lognormal $\mathrm{med}\,D=e^\mu<e^{\mu+s^2/2}=\E D$.
\end{proof}

\subsection{Excess incidents in degraded mode}
\begin{proposition}[Campbell moments of discounted excess-incident loss]\label{prop:campbell}
Under \cref{ass:poisson}, let $\{U_i\}$ be the arrival times of excess incidents on $[0,\tau)$, with rate $\lambda_{\mathrm{ex}}$, and let $D_i$ be their durations. The discounted excess-incident loss $S=\sum_ie^{-\rho U_i}L_{\mathrm{ac}}(D_i)$ satisfies
\begin{equation}\label{eq:campbell}
\E[S\mid\theta,\tau]=\lambda_{\mathrm{ex}}\,\E L_{\mathrm{ac}}(D)\,\frac{1-e^{-\rho\tau}}{\rho},\qquad
\Var[S\mid\theta,\tau]=\lambda_{\mathrm{ex}}\,\E L_{\mathrm{ac}}(D)^2\,\frac{1-e^{-2\rho\tau}}{2\rho}.
\end{equation}
The expected excess outage hours are $\E[\sum_iD_i\mid\theta,\tau]=\lambda_{\mathrm{ex}}\tau\,e^{\mu+s^2/2}$.
\end{proposition}
\begin{proof}
$\{(U_i,D_i)\}$ is a marked Poisson process on $[0,\tau)\times\R_+$ with intensity $\lambda_{\mathrm{ex}}\dd u\otimes F_D(\dd x)$. Campbell's theorem gives $\E\sum_ig(U_i,D_i)=\int\!\!\int g\,\lambda_{\mathrm{ex}}\dd u\,F_D(\dd x)$ and $\Var\sum_ig=\int\!\!\int g^2\,\lambda_{\mathrm{ex}}\dd u\,F_D(\dd x)$. Apply this with $g(u,x)=e^{-\rho u}L_{\mathrm{ac}}(x)$.
\end{proof}

\subsection{Functional, workaround, attrition and funding channels}\label{sec:partition}
For a capability $v\in K$ with exposure $E_v$ (volume per year), unit margin $u_v$ and recapture share $\varrho_v$, the functional loss rate is
\begin{equation}\label{eq:fun}
\pi_v(\theta)=E_v\,u_v\,\Delta\ell_v\,(1-\varrho_v).
\end{equation}
The present value over the degraded period is $\mathrm{PV}_\tau[\pi]=\pi(1-e^{-\rho\tau})/\rho$. Workaround operations add $\mathrm{PV}_\tau[n_{\mathrm{FTE}}\,c_{\mathrm{FTE}}]$ plus any one-off fallback capital expenditure.

\paragraph{Customer value.} A customer of segment $j$ with annual margin $m_j$ and annual retention $\gamma_j$ has lifetime value \citep{gupta2004}
\begin{equation}\label{eq:clv}
\mathrm{CLV}_j=\frac{m_j\,\gamma_j}{1+\rho-\gamma_j}.
\end{equation}
Under \cref{ass:hazard}, the excess churn probability over the degraded period is $\chi_j=1-\exp(-\beta^{\mathrm{out}}_jH-\beta^{\mathrm{deg}}_jx_j\tau)$. The attrition channel is $L_{\mathrm{att}}=\sum_jN_j\chi_j\,\mathrm{CLV}_j\,e^{-\rho\tau/2}$, discounted at mid-period. For banks, deposit outflow is $\nu\sum_j\mathrm{Dep}_j\chi_j$, where $\nu$ is the balance-flight multiple. It removes a deposit franchise spread $\varsigma_d$ over a $T_{\mathrm{dep}}$-year life, so $L_{\mathrm{fund}}=\nu\sum_j\mathrm{Dep}_j\chi_j\,\varsigma_d\,(1-e^{-\rho T_{\mathrm{dep}}})/\rho\cdot e^{-\rho\tau/2}$. Margins $m_j$ exclude the deposit spread by construction.

\begin{theorem}[Exact partition of functional and attrition loss]\label{thm:partition}
Consider customer $i$ with but-for margin flow $\phi_i(t)$. During $[0,\tau)$ a fraction $\varepsilon_i$ of it is impaired, and the customer has an excess-churn time $\zeta_i$. Let $\mathrm{CLV}_i(t)$ be the time-$t$ value of the customer's but-for margins from $t$ onward. The customer's loss is
\begin{equation}\label{eq:partition}
L_i=\underbrace{\int_0^{\zeta_i\wedge\tau}e^{-\rho t}\varepsilon_i\phi_i(t)\dd t}_{\text{functional atom}}+\underbrace{\ind\{\zeta_i<\tau\}\,e^{-\rho\zeta_i}\,\mathrm{CLV}_i(\zeta_i)}_{\text{attrition atom}} .
\end{equation}
The two atoms are supported on the disjoint sets $\{t<\zeta_i\}$ and $\{t\ge\zeta_i\}$ of the customer's timeline, so $L_{\mathrm{fun}}+L_{\mathrm{att}}=\sum_iL_i$ involves no double counting. Integrating functional flows over $[0,\tau)$ for every customer while also charging churners' CLV overstates the loss by $\E\int_{\zeta\wedge\tau}^{\tau}e^{-\rho t}\varepsilon\phi\dd t$. With constant hazard $h$ and $x=h\tau\in(0,1]$, the undiscounted exposure satisfies
\begin{equation}\label{eq:survbound}
\E[\zeta\wedge\tau]=\tau\Big(1-\frac{\chi}{2}\Big)+\tau\,g(x),\qquad -\frac{x^2}{12}\le g(x)\le0,\qquad \chi=1-e^{-x}.
\end{equation}
\end{theorem}
\begin{proof}
In the restrained world, customer $i$ yields $(1-\varepsilon_i)\phi_i$ on $[0,\zeta_i\wedge\tau)$, nothing after $\zeta_i$ if $\zeta_i<\tau$, and $\phi_i$ after $\tau$ otherwise. Subtracting from the but-for flow and discounting gives \eqref{eq:partition}. The overstatement is the functional flow on $[\zeta_i,\tau)$, which is already contained in $\mathrm{CLV}_i(\zeta_i)$. For \eqref{eq:survbound}, $\E[\zeta\wedge\tau]=(1-e^{-x})/h$. Define $g(x)=(1-e^{-x})/x-1+(1-e^{-x})/2$. Expanding in series, $g(x)=\sum_{n\ge2}(-1)^{n+1}\frac{n-1}{2(n+1)\,n!}x^n$. The term magnitudes decrease for $x\le1$, since their ratio is $x\,n/((n-1)(n+2))\le\tfrac12$. The series is therefore alternating with first term $-x^2/12$, which gives $-x^2/12\le g\le0$.
\end{proof}
The implementation uses the survival weight $(1-\chi/2)$ on functional flows. By \eqref{eq:survbound}, it overstates exposure by a relative amount of at most $x^2/12$; for our calibrations $x<0.05$, so the error is below $2\times10^{-4}$.

\subsection{Delayed commercialisation as a forcibly postponed option}\label{sec:delay}
A launch that depends on $r$ and is planned for $t_L$ is blocked until $\tau_s$. The forced delay is $\Delta=(\tau_s-t_L)^+$. Under \cref{ass:gbm}, the perpetual option to launch at cost $I$ has value \citep{mcdonald1986,dixit1994}
\begin{equation}\label{eq:F}
F(V)=\begin{cases}(V^\ast-I)\,(V/V^\ast)^\beta,&V<V^\ast,\\ V-I,&V\ge V^\ast,\end{cases}\qquad V^\ast=\frac{\beta}{\beta-1}I,
\end{equation}
where $\beta>1$ is the positive root of $\tfrac12\sigma^2\beta(\beta-1)+(r-\delta)\beta-r=0$:
\begin{equation}\label{eq:beta}
\beta=\tfrac12-\frac{r-\delta}{\sigma^2}+\sqrt{\Big(\frac{r-\delta}{\sigma^2}-\tfrac12\Big)^2+\frac{2r}{\sigma^2}} .
\end{equation}
A restraint that prevents launch before $\Delta$ leaves the value $G(V_0,\Delta)=\E^{\Q}[e^{-r\Delta}F(V_\Delta)]$. We define the forced-delay loss $\ell(\Delta)=F(V_0)-G(V_0,\Delta)$. The static benchmark treats the launch as committed and merely postponed:
\begin{equation}\label{eq:dcf}
\ell_{\mathrm{DCF}}(\Delta)=(V_0-I)-\E^{\Q}\big[e^{-r\Delta}(V_\Delta-I)\big]=V_0(1-e^{-\delta\Delta})-I(1-e^{-r\Delta}).
\end{equation}

\begin{theorem}[Forced-delay loss]\label{thm:delay}
Under \cref{ass:gbm}:
\begin{enumerate}[label=(\roman*),itemsep=0pt]
\item $\ell(\Delta)\ge0$, and $\ell$ is non-decreasing in $\Delta$.
\item If $V_0\ge V^\ast$ (the launch was optimal immediately), then $\ell(\Delta)\le\ell_{\mathrm{DCF}}(\Delta)$, with strict inequality for $\Delta>0$.
\item With $\tau^\ast=\inf\{t:V_t\ge V^\ast\}$,
\[\ell(\Delta)=\E^{\Q}\Big[\ind\{\tau^\ast<\Delta\}\big(e^{-r\tau^\ast}(V^\ast-I)-e^{-r\Delta}F(V_\Delta)\big)\Big]\ \text{ if }V_0<V^\ast,\]
so the loss arises only on paths on which the unrestrained firm would have launched before $\Delta$.
\end{enumerate}
\end{theorem}
\begin{proof}
$F$ is the value function of the optimal stopping problem $\sup_{\varsigma}\E^\Q[e^{-r\varsigma}(V_\varsigma-I)^+]$. By the general theory of optimal stopping for one-dimensional diffusions \citep{peskir2006,dixit1994}, $F$ is the smallest $r$-excessive majorant of $(V-I)^+$. Hence $M_t=e^{-rt}F(V_t)$ is a $\Q$-supermartingale, and $F(V)\ge V-I$ for all $V$.
(i) The supermartingale property gives $G(V_0,\Delta)=\E M_\Delta\le M_0=F(V_0)$. For $\Delta_1<\Delta_2$, $\E M_{\Delta_2}=\E[\E(M_{\Delta_2}\mid\mathcal{F}_{\Delta_1})]\le\E M_{\Delta_1}$.
(ii) Since $F(V)\ge V-I$, $G\ge\E^\Q[e^{-r\Delta}(V_\Delta-I)]=V_0e^{-\delta\Delta}-Ie^{-r\Delta}$. Because $F(V_0)=V_0-I$, $\ell\le\ell_{\mathrm{DCF}}$. The inequality is strict because $F(V)>V-I$ on $\{V<V^\ast\}$, which has positive $\Q$-probability under the lognormal law of $V_\Delta$.
(iii) On the continuation region $\{V<V^\ast\}$, $F$ solves $\tfrac12\sigma^2V^2F''+(r-\delta)VF'-rF=0$ and is bounded by $V^\ast-I$. The stopped process $M_{t\wedge\tau^\ast}$ is therefore a bounded martingale, and optional stopping gives $F(V_0)=\E M_{\tau^\ast\wedge\Delta}$. Subtracting $\E M_\Delta$, the two terms coincide on $\{\tau^\ast\ge\Delta\}$, and $F(V_{\tau^\ast})=V^\ast-I$ by path continuity.
\end{proof}

\begin{proposition}[Closed form of the delayed option value]\label{prop:delayclosed}
Let $d_2=\big[\ln(V_0/V^\ast)+(r-\delta-\tfrac12\sigma^2)\Delta\big]/(\sigma\sqrt\Delta)$ and $d_1=d_2+\sigma\sqrt\Delta$. Then
\begin{equation}\label{eq:Gclosed}
G(V_0,\Delta)=V_0e^{-\delta\Delta}\Phi(d_1)-Ie^{-r\Delta}\Phi(d_2)+(V^\ast-I)\Big(\frac{V_0}{V^\ast}\Big)^{\beta}\Phi\big(-d_2-\beta\sigma\sqrt\Delta\big).
\end{equation}
\end{proposition}
\begin{proof}
Split $F(V_\Delta)$ on $\{V_\Delta\ge V^\ast\}$ and its complement. On $\{V_\Delta\ge V^\ast\}$, the lognormal partial moments give $e^{-r\Delta}\E^\Q[(V_\Delta-I)\ind\{V_\Delta\ge V^\ast\}]=V_0e^{-\delta\Delta}\Phi(d_1)-Ie^{-r\Delta}\Phi(d_2)$. On the complement, write $V_\Delta=V_0\exp(\nu\Delta+\sigma\sqrt\Delta Z)$ with $\nu=r-\delta-\sigma^2/2$. Then $\E[V_\Delta^\beta\ind\{Z<-d_2\}]=V_0^\beta e^{\beta\nu\Delta+\beta^2\sigma^2\Delta/2}\Phi(-d_2-\beta\sigma\sqrt\Delta)$. The fundamental quadratic, rewritten as $\tfrac12\sigma^2\beta^2+\nu\beta=r$, makes the exponent equal to $r\Delta$, which cancels the discount factor.
\end{proof}

\subsection{Remediation: COCOMO~II and stochastic activity networks}\label{sec:remediation}
The design-around build is sized by the COCOMO~II.2000 post-architecture model \citep{boehm2000}:
\begin{equation}\label{eq:cocomo}
\begin{aligned}
\mathrm{PM}_{\mathrm{NS}}&=A\cdot\mathrm{Size}^{E}\cdot\prod_{i}\mathrm{EM}_i, & E&=B+0.01\sum_{j=1}^5\mathrm{SF}_j,\\
\mathrm{TDEV}&=C\cdot\mathrm{PM}_{\mathrm{NS}}^{\,F}\cdot\frac{\mathrm{SCED\%}}{100}, & F&=D+0.2(E-B),
\end{aligned}
\end{equation}
with $A=2.94$, $B=0.91$, $C=3.67$ and $D=0.28$. Size is in KSLOC and effort in person-months. Schedule compression to $75\%$ multiplies effort by $\mathrm{EM}_{\mathrm{SCED}}=1.43$. Programme completion is the longest path of the activity network in \cref{fig:san}:
\begin{equation}\label{eq:san}
T_d=d_A+\max\big(d_B+d_C,\;d_D\big)+d_E,
\end{equation}
where A is the non-practising design and clearance, B the build, C integration and regression, D data and parallel-run preparation, and E cut-over and certification. Durations $d_A,d_C,d_D,d_E\sim\PERT(a,m,b)$ are Beta-PERT with shape parameter $4$. The build is $d_B=\mathrm{TDEV}(\mathrm{Size})\cdot\epsilon_s$, with $\ln\mathrm{Size}\sim N(\ln\mathrm{Size}_{\mathrm{med}},0.35^2)$. Estimation error follows a two-regime mixture reflecting the overrun evidence of \citet{flyvbjerg2011}: with probability $5/6$, $\ln\epsilon_s\sim N(0,0.15^2)$; with probability $1/6$ (a ``black swan''), $\epsilon_s\approx1.7$ and the cost error is $\epsilon_c\approx3.0$, each with log-dispersion $0.15$ and $0.20$ respectively. Programme cost is $\mathrm{PM}\cdot\epsilon_c\cdot c_{\mathrm{PM}}(1+o)$ plus time-driven fixed costs, which scale with $T_d$ relative to plan. Cost accrues uniformly over $[0,T_d]$ and stops at $\tau_s$ if the restraint ends first. Substitution ($S_3$) replaces A and B by vendor selection and integration activities. It adds an upfront licence and an annual licence over $[T_d,T_R]$.

\begin{figure}[t]
\centering
\begin{tikzpicture}[font=\footnotesize,>=Latex,act/.style={draw,circle,minimum size=0.85cm,fill=green!8,inner sep=1pt}]
\node[act] (s) at (0,0) {start};
\node[act] (a) at (2,0) {A};
\node[act] (b) at (4,0.9) {B};
\node[act] (c) at (6,0.9) {C};
\node[act] (d) at (5,-0.9) {D};
\node[act] (e) at (8,0) {E};
\node[act] (f) at (10,0) {done};
\draw[->] (s) -- (a); \draw[->] (a) -- (b); \draw[->] (b) -- (c); \draw[->] (c) -- (e); \draw[->] (a) -- (d); \draw[->] (d) -- (e); \draw[->] (e) -- (f);
\node[below=0.15cm of b,align=center] {\scriptsize COCOMO II\\[-2pt]\scriptsize TDEV$\cdot\epsilon_s$};
\node[above=0.05cm of a] {\scriptsize PERT};
\node[below=0.05cm of d] {\scriptsize PERT};
\end{tikzpicture}
\caption{Stochastic activity network for design-around ($S_1$, $S_2$) and substitution ($S_3$; A, B become vendor selection and integration).}
\label{fig:san}
\end{figure}
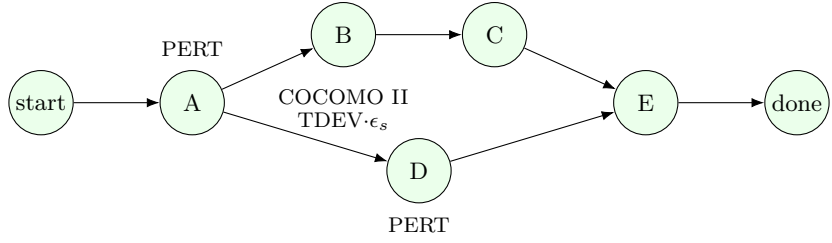

\begin{proposition}[Fulkerson bias and strategy dominance]\label{prop:fulkerson}
\begin{enumerate}[label=(\roman*),itemsep=0pt]
\item The map $\bm d\mapsto T_d(\bm d)$ in \eqref{eq:san} is convex and non-decreasing. Hence $\E T_d\ge T_d(\E\bm d)$ \citep{fulkerson1962}. If every PERT mode lies below its mean, then also $T_d(\E\bm d)\ge T_d(\mathrm{mode}\,\bm d)$.
\item Suppose all non-remediation channels are non-decreasing in $\tau$, and compare two strategies coupled by common random numbers. If $T_d(s)\le T_d(s')$ almost surely, the loss difference satisfies $L(s)-L(s')\le L_{\mathrm{rem}}(s)-L_{\mathrm{rem}}(s')$ almost surely. Faster remediation is preferred whenever its incremental cost is smaller than the reduction in degraded-period loss.
\end{enumerate}
\end{proposition}
\begin{proof}
(i) $T_d$ is a maximum of sums of coordinates, so it is convex and non-decreasing, and Jensen's inequality applies. Monotonicity gives the second inequality. (ii) $T_d(s)\le T_d(s')$ implies $\tau_s\le\tau_{s'}$, and each non-remediation channel is non-decreasing in $\tau$ given $(\theta,\omega)$.
\end{proof}

\subsection{Uncertainty model}\label{sec:uq}
We separate \emph{epistemic} parameters $\theta\in\R^d$ (calibration uncertainty) from \emph{aleatory} events $\omega$: switch-over failure, incident arrivals and durations, and IROPS event sizes. Marginals $F_i$ are lognormal, PERT, triangular or uniform, and are anchored to public benchmarks (\cref{sec:cases}). Dependence is induced by a Gaussian copula \citep{sklar1959,nelsen2006}: $\theta_i=F_i^{-1}(\Phi(Z_i))$ with $\bm Z\sim N(0,\mathbf{R})$, where $\mathbf{R}$ is checked to be positive definite. The law of total variance separates the two sources,
\begin{equation}\label{eq:totvar}
\Var L=\underbrace{\Var\big(\E[L\mid\theta]\big)}_{\text{epistemic}}+\underbrace{\E\big[\Var(L\mid\theta)\big]}_{\text{aleatory}} .
\end{equation}
$\E[L\mid\theta]$ is evaluated by the \emph{mean-field} engine, which applies \cref{thm:jensen,prop:campbell} in closed form. All strategies are evaluated on the same draws (common random numbers), so strategy comparisons are paired.

The conventional \emph{plug-in} estimator, against which we benchmark, sets every parameter to its central value (median for lognormals, mode for PERT and triangular distributions). It replaces the switch-over Bernoulli by its probability, values incidents at the median duration, and runs the activity network deterministically at modal durations.

\subsection{Risk measures and allocation}
For a loss $L$ with continuous distribution function $F_L$, we use $\VaR_\alpha(L)=F_L^{-1}(\alpha)$ and \citep{rockafellar2000,rockafellar2002}
\begin{equation}\label{eq:cvar}
\CVaR_\alpha(L)=\min_{\xi\in\R}\Big\{\xi+\frac{1}{1-\alpha}\E(L-\xi)^+\Big\}=\E[L\mid L\ge\VaR_\alpha(L)],
\end{equation}
which is coherent \citep{artzner1999,acerbi2002}.

\begin{proposition}[Euler allocation of tail risk]\label{prop:euler}
If $L=\sum_kL_k$ has a continuous distribution, then $\CVaR_\alpha(L)=\sum_k\E[L_k\mid L\ge\VaR_\alpha(L)]$. The summands are the Euler contributions \citep{mcneil2015}
\[
\E\big[L_k\mid L\ge\VaR_\alpha(L)\big]=\frac{\partial}{\partial h_k}\CVaR_\alpha\Big(\sum_j h_jL_j\Big)\Big|_{\bm h=\mathbf{1}} .
\]
\end{proposition}
\begin{proof}
The equality follows from linearity of conditional expectation applied to \eqref{eq:cvar}. The derivative representation follows from positive homogeneity of $\CVaR$, Euler's theorem for homogeneous functions, and the standard differentiability result for $\CVaR$ under continuity.
\end{proof}

\begin{proposition}[Distribution-free quantile confidence interval]\label{prop:qci}
Let $X_{(1)}\le\dots\le X_{(n)}$ be the order statistics of an i.i.d.\ sample from a continuous $F$ with $p$-quantile $\xi_p$, and let $1\le l<u\le n$. Then $\Prob(X_{(l)}\le\xi_p<X_{(u)})=\Prob(l\le\mathrm{Bin}(n,p)<u)$ \citep{david2003}.
\end{proposition}
\begin{proof}
$X_{(l)}\le\xi_p<X_{(u)}$ holds iff the number of observations $\le\xi_p$, which is $\mathrm{Bin}(n,p)$, lies in $[l,u)$.
\end{proof}

\subsection{Optimal security under asymmetric error costs}\label{sec:bond}
Suppose a security $B$ is posted to compensate the restrained party if the restraint proves wrongful, which happens with probability $q$. Holding the security costs $c$ per unit, for example the premium rate times the expected horizon. Loss in excess of the security is uncompensated, and a social or decision weight $w$ applies to it. The expected cost of a security level is
\begin{equation}\label{eq:J}
J(B)=c\,B+q\,w\,\E(L-B)^+ .
\end{equation}

\begin{theorem}[Critical-fractile security and the CVaR identity]\label{thm:bond}
Let $L\ge0$ be integrable with continuous $F_L$. If $c<qw$, then $J$ is minimised by
\begin{equation}\label{eq:Bstar}
B^\star=\VaR_{\alpha^\star}(L),\qquad \alpha^\star=1-\frac{c}{q\,w},
\end{equation}
and the minimum is $J(B^\star)=c\cdot\CVaR_{\alpha^\star}(L)$. If $c\ge qw$, then $B^\star=0$. Moreover:
\begin{enumerate}[label=(\alph*),itemsep=0pt]
\item $B^\star$ is non-decreasing in $q$ and $w$, and non-increasing in $c$.
\item The expected uncompensated loss is $\E(L-B^\star)^+=(1-\alpha^\star)\big(\CVaR_{\alpha^\star}(L)-\VaR_{\alpha^\star}(L)\big)$.
\item The rule ``security $=\E L$'' is optimal iff $F_L(\E L)=\alpha^\star$. For right-skewed losses with $F_L(\E L)<\alpha^\star$, it is too low.
\end{enumerate}
\end{theorem}
\begin{proof}
$J$ is convex, with right derivative $J'(B)=c-qw\Prob(L>B)=c-qw(1-F_L(B))$. If $c\ge qw$, then $J'\ge0$ and $B^\star=0$. Otherwise $J'(B)=0$ iff $F_L(B)=\alpha^\star$. Rewrite $J(B)=qw(1-\alpha^\star)\big[B+\E(L-B)^+/(1-\alpha^\star)\big]$. By \eqref{eq:cvar}, the bracket is minimised at $B=\VaR_{\alpha^\star}$ with value $\CVaR_{\alpha^\star}$, so $J(B^\star)=qw(1-\alpha^\star)\CVaR_{\alpha^\star}=c\,\CVaR_{\alpha^\star}$. (a) $\alpha^\star$ is increasing in $q$ and $w$ and decreasing in $c$, and $\VaR$ is monotone in $\alpha$. (b) Substitute $\E(L-\xi)^+=(1-\alpha)(\CVaR_\alpha-\xi)$ at $\xi=\VaR_\alpha$. (c) This follows from the first-order condition.
\end{proof}

\begin{remark}[Estimation-robust security]
Under sampling uncertainty, \cref{prop:qci} gives a distribution-free upper confidence bound $X_{(u)}$ for $\VaR_{\alpha^\star}$. Posting $X_{(u)}$ instead of $\widehat{\VaR}_{\alpha^\star}$ guards against Monte Carlo error. At $n=\rN$ the two differ by less than $1\%$ (\cref{tab:summary}).
\end{remark}

\section{Implementation}\label{sec:impl}
\subsection{Software architecture}
TEIM is implemented as a Python package of about 2{,}500 lines. It runs on NumPy and SciPy, and uses SALib for Sobol sampling and analysis. \Cref{tab:modules} lists the modules. The design follows three rules. (i) \emph{Vectorisation over draws}: every model quantity is a NumPy array indexed by Monte Carlo draw, including the dependency-lattice levels, so one call evaluates $\rN$ scenarios. (ii) \emph{Three evaluation modes, one code path}: \code{mc} samples aleatory events; \code{mean} replaces them by their conditional expectations (\cref{thm:jensen,prop:campbell}); \code{plug} evaluates the deterministic plug-in. All three share the same channel functions, so differences between them are attributable to uncertainty treatment alone. (iii) \emph{Separation of case data from mechanics}: a case is a subclass of \code{Case} that supplies parameters (\code{Param(spec, group, desc, anchor, unit)}), a dependency graph, an acute loss function and channel formulas. Everything else is generic.

\begin{table}[t]
\centering\footnotesize
\caption{Implementation modules and their sizes (physical lines of code).}
\label{tab:modules}
\begin{tabularx}{\textwidth}{@{}l r X@{}}
\toprule
Module & LoC & Responsibility \\
\midrule
\code{graph.py} & 132 & Service grammar (leaf/and/or/kofn/alt), vectorised Kleene iteration for $\lfp$/$\gfp$, chain construction, acyclicity test, brute-force fixed-point enumeration (verification) \\
\code{ctmc.py} & 86 & Steady state, uniformization (transient and cumulative), Markov failure frequency, $k$-of-$n$ birth--death generator \\
\code{architecture.py} & 64 & But-for vs fallback architectures as series systems; CSL interval reliability; implied incident multiple $\hat m$ \\
\code{queueing.py} & 66 & Fluid backlog (\cref{thm:backlog}); stochastic time-stepped verifier with common random numbers \\
\code{acute.py} & 77 & Acute loss function \eqref{eq:acute}; closed-form lognormal expectation \eqref{eq:acuteexp}; Jensen gap \\
\code{options.py} & 85 & $\beta$, $V^\ast$, $F$; closed form \eqref{eq:Gclosed}; Gauss--Hermite and Monte Carlo verifiers; DCF benchmark \\
\code{project.py} & 123 & COCOMO~II.2000; Beta-PERT; vectorised longest path; design-around and substitution samplers \\
\code{stats\_utils.py} & 162 & Gaussian copula; inverse CDFs; CLV; VaR/CVaR; order-statistic CIs; uncompensated-loss curve; critical fractile \\
\code{engine.py} & 206 & Generic engine: timing, switch-over, incidents (three modes), delay channel, evaluation, plug-in \\
\code{econometrics.py} & 167 & Panel simulator; TWFE; Callaway--Sant'Anna ATT$(g,t)$ with cluster bootstrap; synthetic control with placebo inference \\
\code{cases/*.py} & 479 & Bank~A, Operator~T, Carrier~P (parameters, graphs, channels) \\
\code{verify.py}, \code{experiments.py} & 391 & Verification suite (\cref{sec:verif}); experimental pipeline (\cref{sec:experiments}) \\
\bottomrule
\end{tabularx}
\end{table}

\subsection{Algorithms}
\Cref{alg:kleene} computes capability levels for all draws simultaneously. Caps $q$ and bases $b$ may be arrays of shape $(n,)$. Convergence is tested by exact equality, which is valid because all values lie in a finite chain (\cref{thm:lattice}).

\begin{algorithm}[t]
\caption{Vectorised Kleene iteration for $\lfp\Phi$ (resp.\ $\gfp\Phi$)}\label{alg:kleene}
\begin{algorithmic}[1]
\Require services $V$ in declaration order, operators $f_v$, bases $b$ (arrays of shape $(n,)$), caps $q$, start $\in\{\bot,\top\}$
\State $s_v\gets 0$ (or $1$) for all $v\in V$
\Repeat
  \State $s'_v\gets \min\big(b_v,\,f_v(s)\big)$ for all $v$ \Comment{elementwise over $n$ draws}
  \State $\mathit{changed}\gets\exists v:\ s'_v\ne s_v$
  \State $s\gets s'$
\Until{not $\mathit{changed}$} \Comment{at most $|V|(|\mathcal{C}|-1)$ strict changes per draw}
\State \Return $s$
\end{algorithmic}
\end{algorithm}

\begin{algorithm}[t]
\caption{Uniformization with cumulative occupation (\cref{prop:unif})}\label{alg:unif}
\begin{algorithmic}[1]
\Require generator $\mathbf{Q}$, initial law $\bm\pi_0$, horizon $t$, tolerance $\varepsilon=10^{-12}$
\State $\Lambda\gets1.02\max_i|q_{ii}|$; $\mathbf{P}\gets\mathbf{I}+\mathbf{Q}/\Lambda$; $N\gets F^{-1}_{\Pois(\Lambda t)}(1-\varepsilon)+10$
\State $\bm v\gets\bm\pi_0$; $\bm\pi\gets\mathbf{0}$; $\bm c\gets\mathbf{0}$
\For{$k=0,\dots,N$}
  \State $\bm\pi\gets\bm\pi+\Pois(k;\Lambda t)\,\bm v$; \quad $\bm c\gets\bm c+\Prob(N_{\Lambda t}>k)\,\bm v$; \quad $\bm v\gets\bm v\mathbf{P}$
\EndFor
\State \Return $\bm\pi$, $\bm c/\Lambda$
\end{algorithmic}
\end{algorithm}

\begin{algorithm}[t]
\caption{TEIM Monte Carlo engine with common random numbers}\label{alg:mc}
\begin{algorithmic}[1]
\Require case $\mathcal{P}$, strategies $\mathcal{S}$, sample size $n$, seeds $(\varsigma_1,\varsigma_2)$
\State $\bm Z\sim N(0,\mathbf{R})^{n}$ with seed $\varsigma_1$; $U\gets\Phi(\bm Z)$ clipped to $[10^{-12},1-10^{-12}]$; $\theta_i\gets F_i^{-1}(U_i)$
\For{$s\in\mathcal{S}$} \Comment{same $\theta$ and same aleatory seed $\varsigma_2$ for every $s$}
  \State $(T_d,\text{cost})\gets$ SAN sampler (\eqref{eq:cocomo}--\eqref{eq:san}) driven by $U_{\text{size}},U_{\text{err}},U_{\text{swan}},U_A,\dots,U_E$
  \State $\tau\gets\min(T_R,T_d)$; $\ell\gets$ \Call{Kleene}{$b_r=0$, caps$(\theta)$}
  \State switch-over: $\ind\{U_{\text{sw}}<p_0\}\,L_{\mathrm{ac}}(D_0)$
  \State incidents: $N\sim\Pois(\lambda_{\mathrm{ex}}\tau)$; flatten all events with \code{np.repeat}; $D\sim\LN$; $U\sim\mathcal{U}(0,\tau)$; aggregate $e^{-\rho U}L_{\mathrm{ac}}(D)$ by \code{np.bincount}
  \State $H\gets$ outage hours; $\chi_j\gets1-\exp(-\beta^{\mathrm{out}}_jH-\beta^{\mathrm{deg}}_jx_j\tau)$
  \State functional, ops, attrition, funding via \eqref{eq:fun}--\eqref{eq:clv} with survival weights $(1-\chi/2)$
  \State delay: $\Delta\gets(\tau-t_L)^+$; $\ell(\Delta)\gets F(V_0)-G(V_0,\Delta)$ by \eqref{eq:Gclosed}; discount by $e^{-\rho t_L}$
  \State $L(s)\gets\sum_kL_k$
\EndFor
\State $s^\star\gets\arg\min_s\overline{L(s)}$; compute moments, quantiles with CIs (\cref{prop:qci}), $\CVaR$, Euler contributions, $U(B)$
\end{algorithmic}
\end{algorithm}

\paragraph{Uncompensated-loss curve.} $U(B)=\E(L-B)^+$ is computed for a whole grid of $B$ in $O(n\log n)$. Sort the sample once, precompute suffix sums $S_i=\sum_{j\ge i}X_{(j)}$, and evaluate $U(B)=n^{-1}\big(S_{i(B)}-B\,(n-i(B))\big)$, where $i(B)$ is found by binary search.

\paragraph{Sensitivity analysis.} Grouped first-order and total Sobol indices are computed with Saltelli's scheme \citep{saltelli2010} in SALib. We use $N=4096$ base samples and parameters grouped by mechanism, which gives $N(g+2)$ model evaluations. Sobol indices are computed on the mean-field response $\E[L\mid\theta]$, i.e.\ the epistemic term in \eqref{eq:totvar}. A one-at-a-time tornado varies each parameter between its 10th and 90th percentiles, holding the others at their medians.

\subsection{Numerical safeguards}
All normal CDF evaluations use \code{scipy.special.ndtr}. Power terms in $F$ are guarded with $\max(V,10^{-300})$, and $\Delta\le0$ is handled by an explicit branch to avoid $0/0$ in $d_{1,2}$. Uniforms are clipped before inversion. The copula correlation matrix is checked for positive definiteness through its eigenvalues before the Cholesky factorisation. CTMC steady states are solved as the augmented least-squares system $[\mathbf{Q}^\top;\mathbf{1}^\top]\bm\pi=[\mathbf{0};1]$. The PERT inverse uses the Beta quantile with $\alpha=1+4(m-a)/(b-a)$ and $\beta=1+4(b-m)/(b-a)$. In the event-driven incident simulation, total memory is proportional to $\sum_i N_i$, about $n\,\E[\lambda_{\mathrm{ex}}\tau]$.

\subsection{Complexity}
Per strategy and draw, the cost is $O(h\cdot|V|)$ for the lattice, where $h$ is the depth for acyclic graphs. The remaining costs are $O(d)$ for parameter sampling, $O(1)$ for the SAN, $O(\E N)$ for incidents and $O(1)$ for the closed-form option value. With $n$ draws, $|\mathcal{S}|$ strategies and $n_p$ policies, the experiment costs $O\big(n\,|\mathcal{S}|\,n_p\,(h|V|+d+\E N)\big)$, plus $O(n\log n)$ for the risk statistics. CTMC solution is $O(Nk^2)$ for $N$ uniformization steps on $k$ states. It is performed once per architecture, not per draw.

\subsection{Reproducibility}
All experiments were run with Python~3.12.3, NumPy~2.5.3, SciPy~1.18.1 and SALib~1.6.0 on a 6-vCPU Intel Xeon (2.5\,GHz) virtual machine. Seeds are fixed: $\varsigma_1=20261019$ for parameter sampling and $\varsigma_1+1$ for aleatory events; seed $7$ for Sobol sampling and analysis; seed $99$ for econometric replications and $2026$ for the illustrative panel. The complete pipeline (three cases, all strategies, Sobol, tornado, convergence and econometrics) runs in $\rTotalRuntime$\,s. All tables and figures in \cref{sec:experiments}, and the numerical results quoted in the text, are generated by the pipeline and inserted into the manuscript as macros and data files. \Cref{lst:acute,lst:option} show the two closed forms at the core of the mean-field engine.

\begin{lstlisting}[style=py,caption={Closed-form lognormal expectation of the acute loss (\cref{thm:jensen}).},label={lst:acute}]
def expect_lognormal(self, mu, s):
    m1 = np.exp(mu + 0.5 * s**2)                  # E[D]
    m2 = np.exp(2 * mu + 2 * s**2)                # E[D^2]
    E = self.K0 + self.a * m1 + self.b * m2
    for th, K in self.steps:                      # K * P(D > th)
        E = E + K * ndtr((mu - np.log(th)) / s)
    for eta, h in self.hinges:                    # h * E[(D - eta)^+]
        d2 = (mu - np.log(eta)) / s
        E = E + h * (m1 * ndtr(d2 + s) - eta * ndtr(d2))
    return E
\end{lstlisting}

\begin{lstlisting}[style=py,caption={Closed-form value of a forcibly delayed perpetual option (\cref{prop:delayclosed}).},label={lst:option}]
def delayed_value(V0, Delta, I, r, delta, sigma):
    b  = beta_root(r, delta, sigma); Vs = b / (b - 1) * I
    D  = np.maximum(Delta, 1e-12); sq = sigma * np.sqrt(D)
    d2 = (np.log(V0 / Vs) + (r - delta - 0.5 * sigma**2) * D) / sq
    d1 = d2 + sq
    G  = (V0 * np.exp(-delta * D) * ndtr(d1) - I * np.exp(-r * D) * ndtr(d2)
          + (Vs - I) * (V0 / Vs)**b * ndtr(-d2 - b * sq))
    return np.where(Delta <= 0, F(V0, I, r, delta, sigma), G)
\end{lstlisting}

\section{Case Studies on Three Enterprise Platforms}\label{sec:cases}
\subsection{Confidentiality and data protocol}\label{sec:protocol}
The three case studies are kept confidential. The organisations are referred to only as \emph{Bank~A}, \emph{Operator~T} and \emph{Carrier~P}. Their names, jurisdictions, vendors, product names and the identity of any patent or counterparty are withheld. No confidential data are used. Each case is a \emph{synthetic, calibrated} replica of a platform of the stated scale. The architecture is a stylised composite of publicly documented designs. Every parameter is either anchored to a public benchmark (cited in the calibration tables) or explicitly labelled as an assumption and given a wide prior. The results therefore show what the model does on realistic inputs; they are not estimates for any real firm. Common settings: the restraint horizon is $T_R\sim\LN(2\text{ yr},0.40)$ (median two years, $\E T_R=\rBankETR$\,yr); the risk-free rate is $r=4\%$; and the discount rates are industry costs of capital (January~2026: regional banks $4.98\%$, wireless telecom $5.48\%$, air transport $6.72\%$) \citep{damodaran2026}. In the tables, $\LN(m,s)$ denotes a lognormal with median $m$ and log-standard deviation $s$; $\PERT(a,m,b)$ is Beta-PERT with minimum $a$, mode $m$ and maximum $b$; and $\mathrm{Tri}$ and $\mathcal{U}$ are triangular and uniform.

\subsection{Case A: core-banking platform (Bank~A)}\label{sec:bank}
\paragraph{Platform.} Bank~A is a regional bank with $6.0$\,m retail and $0.30$\,m SME customers, \$90\,bn of retail and \$45\,bn of SME deposits. It posts about $450{,}000$ transactions per hour at peak. Its core ledger includes a \emph{real-time posting and available-balance engine} (RTPE). The RTPE supplies authoritative balances to card authorisation, mobile and P2P payments, instant payments (RTP/FedNow) and liquidity reservation, and it is the foundation of a planned real-time treasury product.

\paragraph{Restraint and fallback.} The RTPE is withdrawn. The fallback is a memo-post \emph{stand-in} (shadow balances) fed by end-of-day batch posting. Available-balance quality is capped at $q_{\mathrm{AVB}}\sim\PERT(0.75,0.85,0.95)$ and liquidity reservation at $q_{\mathrm{LIQ}}\sim\PERT(0.25,0.40,0.60)$. \Cref{fig:graph-bank} shows the dependency graph with lfp levels at central parameters. Instant payments fall to $\ell_{\mathrm{IP}}=0.40$, limited by liquidity reservation. Card authorisation and mobile fall to $0.85$. ACH is unaffected ($\ell=1$, by \cref{thm:locality}: ACH depends only on the batch path and the gateway). The planned treasury product is blocked ($\ell=0$).

\begin{figure}[t]
\centering
\resizebox{\textwidth}{!}{%
\begin{tikzpicture}[>=Latex,font=\scriptsize,
  svc/.style={draw,rounded corners=2pt,text width=2.45cm,minimum height=0.66cm,align=center,inner sep=1.5pt},
  infra/.style={svc,fill=gray!12}, core/.style={svc,fill=blue!6}, chan/.style={svc,fill=blue!3},
  restr/.style={svc,fill=red!25,draw=red!70!black,very thick}, fb/.style={svc,fill=orange!18},
  capnode/.style={svc,fill=green!12}]
\node[infra] (DC) at (0.0,6.2) {Data centre\\[-1pt]{\color{black!60}$\ell=1$}};
\node[infra] (NET) at (0.0,1.0) {Network\\[-1pt]{\color{black!60}$\ell=1$}};
\node[infra] (IAM) at (0.0,3.6) {Identity\\[-1pt]{\color{black!60}$\ell=1$}};
\node[core] (CIF) at (3.375,4.6) {Customer master\\[-1pt]$\wedge$, {\color{black!60}$\ell=1$}};
\node[core] (LDB) at (3.375,6.6) {Ledger DB\\[-1pt]$\wedge$, {\color{black!60}$\ell=1$}};
\node[restr] (RTPE) at (6.75,4.8) {RT posting engine\\[-1pt]$\wedge$, {\color{red!70!black}$\ell=0$}};
\node[fb] (MEMO) at (6.75,7.8) {Memo-post\\[-1pt]$\wedge$, {\color{black!60}$\ell=1$}};
\node[fb] (BATCH) at (6.75,6.6) {EOD batch\\[-1pt]$\wedge$, {\color{black!60}$\ell=1$}};
\node[fb] (STANDIN) at (10.125,7.2) {Stand-in\\[-1pt]$\wedge$, {\color{black!60}$\ell=1$}};
\node[core] (AVB) at (13.5,5.4) {Avail. balance\\[-1pt]alt$_{(1,0.85)}$, {\color{red!70!black}$\ell=0.85$}};
\node[core] (LIQ) at (13.5,7.0) {Liquidity resv.\\[-1pt]alt$_{(1,0.40)}$, {\color{red!70!black}$\ell=0.4$}};
\node[core] (FRAUD) at (16.875,4.4) {Fraud engine\\[-1pt]$\wedge$, {\color{red!70!black}$\ell=0.85$}};
\node[chan] (IPGW) at (6.75,2.2) {RTP/FedNow GW\\[-1pt]$\wedge$, {\color{black!60}$\ell=1$}};
\node[chan] (CARDGW) at (6.75,1.1) {Card GW\\[-1pt]$\wedge$, {\color{black!60}$\ell=1$}};
\node[chan] (ACHGW) at (6.75,0.0) {ACH GW\\[-1pt]$\wedge$, {\color{black!60}$\ell=1$}};
\node[capnode] (CAPIP) at (20.25,5.9) {Instant payments\\[-1pt]$\wedge$, {\color{red!70!black}$\ell=0.4$}};
\node[capnode] (CAPCARD) at (20.25,4.2) {Card authorisation\\[-1pt]$\wedge$, {\color{red!70!black}$\ell=0.85$}};
\node[capnode] (CAPACH) at (20.25,0.6) {ACH\\[-1pt]$\wedge$, {\color{black!60}$\ell=1$}};
\node[capnode] (CAPMOB) at (20.25,2.5) {Mobile/P2P\\[-1pt]$\wedge$, {\color{red!70!black}$\ell=0.85$}};
\node[capnode] (CAPTREAS) at (20.25,7.6) {RT treasury (planned)\\[-1pt]$\wedge$, {\color{red!70!black}$\ell=0$}};
\draw[->,black!55] (DC.east) -- (CIF.west);
\draw[->,black!55] (IAM.east) -- (CIF.west);
\draw[->,black!55] (DC.east) -- (LDB.west);
\draw[->,black!55] (LDB.east) -- (RTPE.west);
\draw[->,black!55] (CIF.east) -- (RTPE.west);
\draw[->,black!55] (LDB.east) -- (MEMO.west);
\draw[->,black!55] (LDB.east) -- (BATCH.west);
\draw[->,black!55] (MEMO.east) -- (STANDIN.west);
\draw[->,black!55] (BATCH.east) -- (STANDIN.west);
\draw[->,red!70!black,thick] (RTPE.east) -- (AVB.west);
\draw[->,densely dashed,orange!70!black,thick] (STANDIN.east) -- (AVB.west);
\draw[->,red!70!black,thick] (RTPE.east) -- (LIQ.west);
\draw[->,densely dashed,orange!70!black,thick] (STANDIN.east) -- (LIQ.west);
\draw[->,black!55] (AVB.east) -- (FRAUD.west);
\draw[->,black!55] (IAM.east) -- (FRAUD.west);
\draw[->,black!55] (NET.east) -- (IPGW.west);
\draw[->,black!55] (NET.east) -- (CARDGW.west);
\draw[->,black!55] (NET.east) -- (ACHGW.west);
\draw[->,black!55] (IPGW.east) -- (CAPIP.west);
\draw[->,black!55] (LIQ.east) -- (CAPIP.west);
\draw[->,black!55] (AVB.east) -- (CAPIP.west);
\draw[->,black!55] (FRAUD.east) -- (CAPIP.west);
\draw[->,black!55] (CARDGW.east) -- (CAPCARD.west);
\draw[->,black!55] (AVB.east) -- (CAPCARD.west);
\draw[->,black!55] (FRAUD.east) -- (CAPCARD.west);
\draw[->,black!55] (ACHGW.east) -- (CAPACH.west);
\draw[->,black!55] (BATCH.east) -- (CAPACH.west);
\draw[->,black!55] (AVB.east) -- (CAPMOB.west);
\draw[->,black!55] (IAM.east) -- (CAPMOB.west);
\draw[->,red!70!black,thick] (RTPE.east) -- (CAPTREAS.west);
\draw[->,black!55] (LIQ.east) -- (CAPTREAS.west);
\draw[->,black!55] (IPGW.east) -- (CAPTREAS.west);
\end{tikzpicture}}
\caption{Bank~A dependency graph. Red: restrained RTPE. Orange: fallback services. Dashed arrows: capped alternatives. Levels $\ell$ are the least fixed point at central parameters.}
\label{fig:graph-bank}
\end{figure}
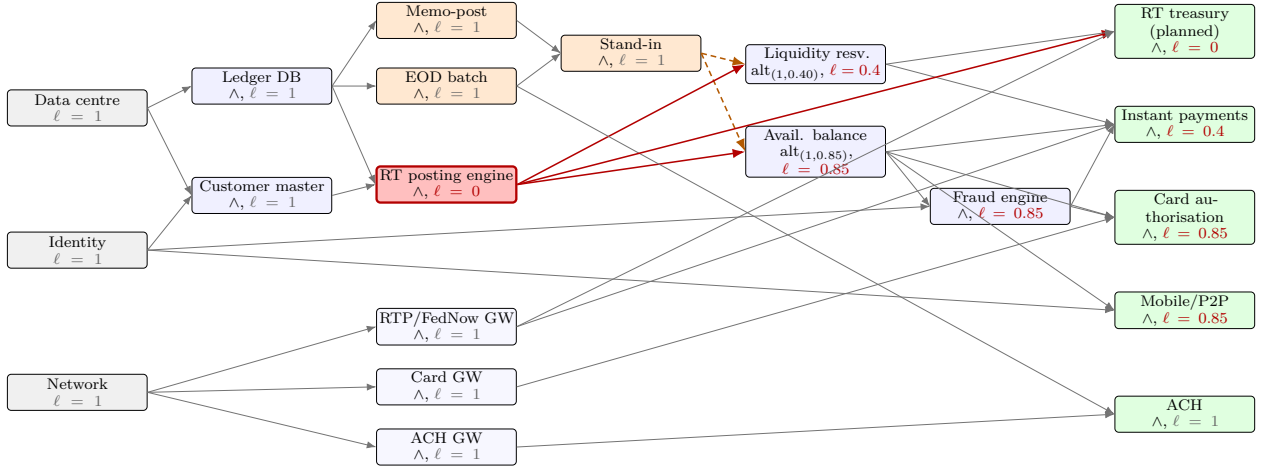

\paragraph{Channels.} The acute loss \eqref{eq:acute} has $a\sim\LN(\$2\,\mathrm{m/h},0.35)$, anchored to ITIC's finding that banking is among the verticals with large-enterprise hourly costs above \$5\,m and that 41\% of large firms exceed \$1\,m per hour \citep{itic2024}. The backlog term is $b=w\cdot\tfrac12\psi\lambda(1+\psi/(\kappa-1))$ with $\lambda=450{,}000$/h, $\psi=0.8$ and catch-up multiple $\kappa\sim\PERT(1.3,1.6,2.2)$. Customer redress accrues as a hinge after $4$\,h, scaled to the TSB redress programme \citep{fcatsb2022}. A regulatory-enforcement step at $12$\,h has expected size $p_{\mathrm{enf}}\cdot\mathrm{fine}$, with the fine anchored to the TSB (\pounds48.65\,m) and RBS (\pounds56\,m) penalties \citep{fcatsb2022,fcarbs2014}. Functional loss combines lost instant-payment margin and stale-balance losses:
\[ \pi_{\mathrm{fun}}=\underbrace{\mathrm{Vol}_{\mathrm{IP}}(1-\ell_{\mathrm{IP}})\,u_{\mathrm{IP}}(1-\varrho_{\mathrm{IP}})}_{\text{instant payments}}\Big(1-\frac{\chi_s}{2}\Big)+\underbrace{\mathrm{Vol}_{\mathrm{debit}}(1-\ell_{\mathrm{CARD}})\,\mathrm{bp}_{\mathrm{stale}}}_{\text{stale-balance risk}}\Big(1-\frac{\chi_r}{2}\Big), \]
with survival weights from \cref{thm:partition}. Retail and SME churn probabilities are
\[\chi_r=1-e^{-\beta^{\mathrm{out}}H-\beta^{\mathrm{deg}}_r(1-\ell_{\mathrm{MOB}})\tau},\qquad \chi_s=1-e^{-1.5\beta^{\mathrm{out}}H-\beta^{\mathrm{deg}}_s(1-\ell_{\mathrm{IP}})\tau}.\]
Here $H$ is the switch-over plus excess-incident outage hours. The per-hour hazard $\beta^{\mathrm{out}}\sim\LN(10^{-4},0.6)$ is scaled to TSB's reported loss of about 80{,}000 customers out of 5.2\,m after a multi-week disruption \citep{fcatsb2022}. Funding loss applies a deposit-franchise spread $\varsigma_d\sim\PERT(1.0\%,1.7\%,2.5\%)$, consistent with an industry net interest margin of $3.22\%$ \citep{fdic2024}, to lost deposits over a five-year life. The delayed real-time treasury product has $V_0\sim\LN(\$400\,\mathrm{m},0.30)$, $I=\$100$\,m, $\sigma=0.30$ and $t_L=0.5$\,yr.

\paragraph{Remediation.} The design-around re-implements real-time posting outside the claimed method: about $\rBankKSLOC$ KSLOC (log-sd $0.35$) with $\sum\mathrm{SF}=\rBankSF$ ($E=\rBankEexp$, $F=\rBankFexp$). Effort multipliers are \rBankEMlist{} (EAF $=\rBankEAF$). The nominal plan is $\mathrm{PM}=\rBankPMnom$ person-months and $\mathrm{TDEV}=\rBankTDEVnom$ months; at 75\% compression it is $\rBankPMcr$ PM over $\rBankTDEVcr$ months. Network durations (months) are $A\sim\PERT(2,3,5)$, $C\sim\PERT(3,4,7)$, $D\sim\PERT(3,5,9)$ and $E\sim\PERT(1,2,4)$. Substitution licenses a vendor real-time ledger: selection $\PERT(2,3,5)$, integration $\PERT(6,9,15)$, \$12\,m upfront, \$18\,m integration and \$4\,m per year. Industry experience with core replacement (multi-year programmes) supports these as fast-track scenarios.

\begin{small}
\setlength{\tabcolsep}{3pt}
\begin{longtable}{@{}p{1.9cm}p{4.5cm}p{3.3cm}p{1.0cm}p{5.0cm}@{}}
\caption{Bank~A calibration. Anchors cite public sources; ``assumption'' marks expert priors.}\label{tab:calib-bank}\\
\toprule
Symbol & Description & Distribution & Unit & Anchor\\ \midrule\endfirsthead
\toprule Symbol & Description & Distribution & Unit & Anchor\\ \midrule\endhead
\bottomrule\endfoot
\multicolumn{5}{l}{\textit{Restraint horizon}}\\
\texttt{T\_R} & restraint horizon (years) & $\mathrm{LN}(2,\,0.4)$ & yr & scenario assumption\\
\multicolumn{5}{l}{\textit{Switch-over outage}}\\
\texttt{p0} & P(severe outage at withdrawal of RTPE) & $\mathrm{PERT}(0.15,0.3,0.55)$ &  & change-related outage share (Uptime 2024/25); TSB 2018\\
\texttt{D0} & switch-over outage duration & $\mathrm{LN}(6,\,1)$ & h & assumption; heavy tail per RBS 2012/TSB 2018\\
\multicolumn{5}{l}{\textit{Outage economics}}\\
\texttt{a\_hr} & linear hourly loss of full core outage & $\mathrm{LN}(2\,m,\,0.35)$ & USD/h & ITIC 2024: $>$\$1M/h for 41\% of large firms; banking $>$\$5M/h\\
\texttt{K0} & fixed incident-response cost & $\mathrm{LN}(500\,k,\,0.5)$ & USD & assumption\\
\texttt{kappa} & catch-up capacity multiple c/lambda & $\mathrm{PERT}(1.3,1.6,2.2)$ &  & assumption\\
\texttt{w\_tx} & delay cost per posting-hour & $\mathrm{LN}(0.08,\,0.5)$ & USD & assumption (interest claims, fee waivers, complaints)\\
\texttt{h\_red} & customer redress per hour beyond 4 h & $\mathrm{LN}(800\,k,\,0.5)$ & USD/h & TSB redress GBP 32.7m (FCA 2022)\\
\texttt{p\_enf} & P(enforcement | outage $>$ 12 h) & $\mathrm{PERT}(0.1,0.25,0.5)$ &  & assumption\\
\texttt{fine} & enforcement penalty if imposed & $\mathrm{LN}(40\,m,\,0.6)$ & USD & TSB GBP 48.65m (2022); RBS GBP 56m (2014)\\
\multicolumn{5}{l}{\textit{Incident process}}\\
\texttt{lam\_bf} & but-for severe-incident rate & $1$ & 1/yr & assumption\\
\texttt{lam\_mult} & fallback incident-rate multiple & $\mathrm{LN}(2.5,\,0.4)$ &  & architecture CTMC (\S\ref{sec:arch}): 2.53\\
\texttt{D\_inc\_med} & median incident duration & $3$ & h & assumption\\
\texttt{D\_inc\_sig} & log-sd incident duration & $0.9$ &  & assumption\\
\multicolumn{5}{l}{\textit{Fallback capability}}\\
\texttt{cap\_avb} & stand-in available-balance quality cap & $\mathrm{PERT}(0.75,0.85,0.95)$ &  & architecture assessment\\
\texttt{cap\_liq} & stand-in liquidity-reservation cap & $\mathrm{PERT}(0.25,0.4,0.6)$ &  & architecture assessment\\
\multicolumn{5}{l}{\textit{Economic scale}}\\
\texttt{ip\_vol} & instant payments per year & $\mathrm{LN}(20\,m,\,0.2)$ & 1/yr & TCH RTP 2024: 343m tx; $\sim$6\% share\\
\texttt{ip\_margin} & margin per instant payment & $\mathrm{LN}(0.6,\,0.3)$ & USD & assumption\\
\texttt{ip\_recap} & margin recaptured via same-day ACH & $\mathrm{PERT}(0.2,0.35,0.5)$ &  & assumption\\
\texttt{debit\_vol} & card/P2P debit volume & $\mathrm{LN}(150\,bn,\,0.15)$ & USD/yr & assumption (6m customers)\\
\texttt{stale\_bps} & incremental loss on stale-balance volume & $\mathrm{LN}(6,\,0.6)$ & bp & assumption\\
\multicolumn{5}{l}{\textit{Workaround operations}}\\
\texttt{fte} & manual exception FTE & $\mathrm{PERT}(120,220,400)$ & FTE & assumption\\
\texttt{fte\_cost} & loaded cost per FTE-year & $110\,k$ & USD & assumption\\
\multicolumn{5}{l}{\textit{Customer response}}\\
\texttt{m\_r} & retail annual margin ex deposit spread & $\mathrm{LN}(160,\,0.25)$ & USD & assumption\\
\texttt{m\_s} & SME annual margin ex deposit spread & $\mathrm{LN}(2400,\,0.3)$ & USD & assumption\\
\texttt{ret\_r} & retail annual retention & $0.9$ &  & assumption\\
\texttt{ret\_s} & SME annual retention & $0.92$ &  & assumption\\
\texttt{beta\_out} & churn hazard per outage hour & $\mathrm{LN}(1.0\text{\text{e-}}4,\,0.6)$ & 1/h & TSB: $\sim$80k of 5.2m left in 2018\\
\texttt{beta\_deg\_r} & retail churn hazard per unit capability loss & $\mathrm{LN}(0.015,\,0.6)$ & 1/yr & assumption\\
\texttt{beta\_deg\_s} & SME churn hazard per unit capability loss & $\mathrm{LN}(0.02,\,0.6)$ & 1/yr & assumption\\
\texttt{dep\_mult} & balance-flight multiple of churn & $\mathrm{LN}(1,\,0.3)$ &  & assumption\\
\texttt{spread} & deposit franchise spread & $\mathrm{PERT}(0.01,0.017,0.025)$ & 1/yr & FDIC QBP 2024: NIM 3.22\%\\
\multicolumn{5}{l}{\textit{Delayed launch}}\\
\texttt{V0} & PV at launch of real-time treasury product & $\mathrm{LN}(400\,m,\,0.3)$ & USD & assumption\\
\texttt{I\_launch} & launch investment & $100\,m$ & USD & assumption\\
\texttt{sig\_V} & project-value volatility & $0.3$ &  & assumption\\
\texttt{delta\_V} & opportunity-cost (payout+erosion) rate & $\mathrm{PERT}(0.08,0.12,0.18)$ & 1/yr & assumption\\
\texttt{t\_L} & planned launch date & $0.5$ & yr & assumption\\
\end{longtable}
\end{small}

\subsection{Case B: mobile network operator (Operator~T)}\label{sec:tel}
\paragraph{Platform.} Operator~T is a large mobile network operator with $30$\,m postpaid, $12$\,m prepaid and $1.5$\,m premium-plan subscribers, and about $2{,}500$ enterprise 5G accounts. Its 5G standalone core uses a converged charging function (CHF; 3GPP TS~32.290/32.291) with a \emph{predictive quota-reservation} feature (CHF-PQR). The feature sizes quota grants dynamically, drives usage notifications and spend control, and enables slice-level charging for premium and enterprise products.

\paragraph{Restraint and fallback.} CHF-PQR is withdrawn. Online credit control falls back to an Nchf-to-Gy interworking function (IWF) towards the legacy 4G online charging system, with static quotas. After scale-out, the IWF supports a share $q_{\mathrm{IWF}}\sim\PERT(0.85,0.92,0.98)$ of peak load. Usage notification is capped at $0.5$, and slice charging is unavailable. When credit control is unavailable, the Credit-Control-Failure-Handling policy (CCFH, RFC~8506) decides the outcome: \textsc{continue} (fail-open) grants service without charging, and \textsc{terminate} (fail-closed) denies it \citep{rfc8506}. Both policies are evaluated. Levels (\cref{fig:graph-tel}) are prepaid data $0.92$, spend control $0.5$, premium plans and enterprise slices $0$, and postpaid $1$ (offline charging is unaffected).

\begin{figure}[t]
\centering
\resizebox{\textwidth}{!}{%
\begin{tikzpicture}[>=Latex,font=\scriptsize,
  svc/.style={draw,rounded corners=2pt,text width=2.45cm,minimum height=0.66cm,align=center,inner sep=1.5pt},
  infra/.style={svc,fill=gray!12}, core/.style={svc,fill=blue!6}, chan/.style={svc,fill=blue!3},
  restr/.style={svc,fill=red!25,draw=red!70!black,very thick}, fb/.style={svc,fill=orange!18},
  capnode/.style={svc,fill=green!12}]
\node[infra] (CLOUD) at (0.0,4.8) {Telco cloud\\[-1pt]{\color{black!60}$\ell=1$}};
\node[infra] (TRANS) at (0.0,1.2) {Transport\\[-1pt]{\color{black!60}$\ell=1$}};
\node[infra] (UDM) at (0.0,3.0) {UDM\\[-1pt]{\color{black!60}$\ell=1$}};
\node[infra] (NRF) at (0.0,6.6) {NRF\\[-1pt]{\color{black!60}$\ell=1$}};
\node[core] (SMF) at (3.625,3.6) {SMF\\[-1pt]$\wedge$, {\color{black!60}$\ell=1$}};
\node[core] (UPF) at (3.625,1.0) {UPF\\[-1pt]$\wedge$, {\color{black!60}$\ell=1$}};
\node[core] (PCF) at (3.625,2.3) {PCF\\[-1pt]$\wedge$, {\color{black!60}$\ell=1$}};
\node[restr] (CHFPQR) at (3.625,7.8) {CHF quota feature\\[-1pt]$\wedge$, {\color{red!70!black}$\ell=0$}};
\node[fb] (OCS4G) at (3.625,6.4) {4G OCS\\[-1pt]$\wedge$, {\color{black!60}$\ell=1$}};
\node[fb] (IWF) at (7.25,6.4) {Nchf-Gy IWF\\[-1pt]$\wedge$, {\color{black!60}$\ell=1$}};
\node[core] (CHFOFF) at (3.625,5.0) {Offline charging\\[-1pt]$\wedge$, {\color{black!60}$\ell=1$}};
\node[core] (ONLINE) at (10.875,6.4) {Online credit ctrl\\[-1pt]alt$_{(1,0.92)}$, {\color{red!70!black}$\ell=0.92$}};
\node[core] (NOTIF) at (10.875,5.0) {Usage notification\\[-1pt]alt$_{(1,0.5)}$, {\color{red!70!black}$\ell=0.5$}};
\node[core] (SLICECHG) at (7.25,7.9) {Slice charging\\[-1pt]$\wedge$, {\color{red!70!black}$\ell=0$}};
\node[capnode] (CAPPRE) at (14.75,3.6) {Prepaid 5G data\\[-1pt]$\wedge$, {\color{red!70!black}$\ell=0.92$}};
\node[capnode] (CAPPOST) at (14.75,2.0) {Postpaid 5G data\\[-1pt]$\wedge$, {\color{black!60}$\ell=1$}};
\node[capnode] (CAPSPEND) at (14.75,5.2) {Spend control\\[-1pt]$\wedge$, {\color{red!70!black}$\ell=0.5$}};
\node[capnode] (CAPPREM) at (14.75,7.9) {Premium plans\\[-1pt]$\wedge$, {\color{red!70!black}$\ell=0$}};
\node[capnode] (CAPENT) at (14.75,6.6) {Enterprise slices\\[-1pt]$\wedge$, {\color{red!70!black}$\ell=0$}};
\draw[->,black!55] (CLOUD.east) -- (SMF.west);
\draw[->,black!55] (NRF.east) -- (SMF.west);
\draw[->,black!55] (UDM.east) -- (SMF.west);
\draw[->,black!55] (CLOUD.east) -- (UPF.west);
\draw[->,black!55] (TRANS.east) -- (UPF.west);
\draw[->,black!55] (CLOUD.east) -- (PCF.west);
\draw[->,black!55] (UDM.east) -- (PCF.west);
\draw[->,black!55] (CLOUD.east) -- (CHFPQR.west);
\draw[->,black!55] (CLOUD.east) -- (OCS4G.west);
\draw[->,black!55] (OCS4G.east) -- (IWF.west);
\draw[->,black!55] (CLOUD.east) -- (CHFOFF.west);
\draw[->,red!70!black,thick] (CHFPQR.east) -- (ONLINE.west);
\draw[->,densely dashed,orange!70!black,thick] (IWF.east) -- (ONLINE.west);
\draw[->,red!70!black,thick] (CHFPQR.east) -- (NOTIF.west);
\draw[->,densely dashed,orange!70!black,thick] (IWF.east) -- (NOTIF.west);
\draw[->,red!70!black,thick] (CHFPQR.east) -- (SLICECHG.west);
\draw[->,black!55] (SMF.east) -- (CAPPRE.west);
\draw[->,black!55] (UPF.east) -- (CAPPRE.west);
\draw[->,black!55] (PCF.east) -- (CAPPRE.west);
\draw[->,black!55] (ONLINE.east) -- (CAPPRE.west);
\draw[->,black!55] (SMF.east) -- (CAPPOST.west);
\draw[->,black!55] (UPF.east) -- (CAPPOST.west);
\draw[->,black!55] (PCF.east) -- (CAPPOST.west);
\draw[->,black!55] (CHFOFF.east) -- (CAPPOST.west);
\draw[->,black!55] (NOTIF.east) -- (CAPSPEND.west);
\draw[->,black!55] (PCF.east) -- (CAPSPEND.west);
\draw[->,black!55] (ONLINE.east) -- (CAPPREM.west);
\draw[->,black!55] (SLICECHG.east) -- (CAPPREM.west);
\draw[->,black!55] (NOTIF.east) -- (CAPPREM.west);
\draw[->,black!55] (SLICECHG.east) -- (CAPENT.west);
\draw[->,black!55] (PCF.east) -- (CAPENT.west);
\end{tikzpicture}}
\caption{Operator~T dependency graph (5G SA core). Red: restrained CHF-PQR. Orange: legacy fallback path (4G OCS via Nchf--Gy IWF).}
\label{fig:graph-tel}
\end{figure}
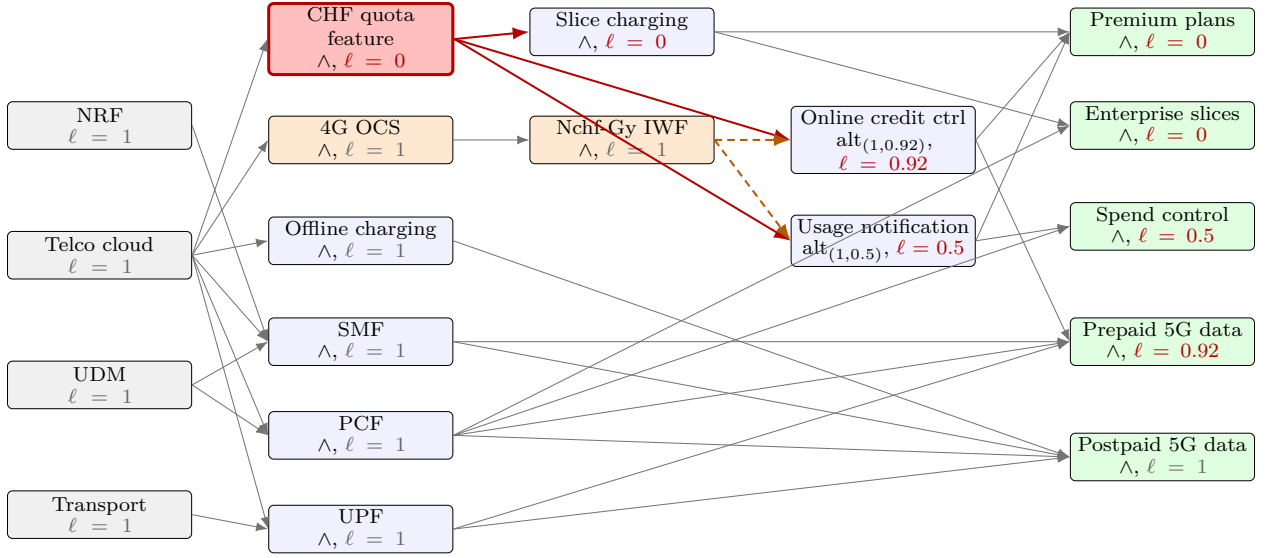

\paragraph{Capacity exceedance.} Credit-control load follows a normalised diurnal profile, $\lambda(h)=0.25+0.75\big[\tfrac12(1+\cos(2\pi(h-21)/24))\big]^{1.5}$, peaking at 21:00. With fallback capacity $q$, the share of usage above capacity and the share of time in denial are
\[ e_u(q)=\frac{\int_0^{24}(\lambda(h)-q)^+\dd h}{\int_0^{24}\lambda(h)\dd h},\qquad e_t(q)=\frac1{24}\int_0^{24}\frac{(\lambda(h)-q)^+}{\lambda(h)}\dd h. \]
Let $R_{pd}=N_{\mathrm{pre}}\cdot\mathrm{ARPU}_{\mathrm{pre}}\cdot12\cdot s_{\mathrm{data}}$ be prepaid data revenue. Under \textsc{continue}, excess usage leaks at rate $R_{pd}e_u$ and outages are invisible to subscribers. Under \textsc{terminate}, revenue $R_{pd}e_u(1-\varrho_{\mathrm{retry}})$ is lost, peak denial adds $\omega_{\mathrm{den}}\,e_t\cdot8760\,\tau$ outage-equivalent hours to churn exposure, and outages trigger goodwill credits.

\paragraph{Channels.} The acute loss has $a=R_{pd}/8760+c_{\mathrm{cc}}$ under \textsc{continue} and $a=(1-\varrho_{\mathrm{retry}})R_{pd}/8760+c_{\mathrm{cc}}$ under \textsc{terminate}, where $c_{\mathrm{cc}}$ is the contact-centre surge cost. The re-authorisation backlog term is $b=v_{\mathrm{sess}}\cdot\tfrac12\lambda_S(1+1/(\kappa-1))$, with $\lambda_S=20$\,m session re-authorisations per hour. Enterprise SLA steps are 10\% of affected monthly recurring charges at $0.73$\,h (a 99.9\% monthly allowance) and 15\% at $3.65$\,h (99.5\%). Under \textsc{terminate}, a goodwill credit step $p_{\mathrm{credit}}\cdot\mathrm{credit}\cdot N_{\mathrm{pre}}$ applies at 4\,h, anchored to the \$5 per-account credit after the February~2024 AT\&T outage and the five-day credit after the July~2022 Rogers outage \citep{fccatt2024,crtc2024}. Functional loss is
\begin{multline*}
\pi_{\mathrm{fun}}=\big(R_{pd}\,s_{\mathrm{leak}}+\pi_{\mathrm{peak}}\big)\Big(1-\frac{\chi_{\mathrm{pre}}}{2}\Big)+\pi_{\mathrm{prem}}(1-\ell_{\mathrm{PREM}})\Big(1-\frac{\chi_{\mathrm{prem}}}{2}\Big)\\
+N_{\mathrm{ent}}\,\mathrm{MRC}\cdot12\cdot s_{\mathrm{slice}}(1-\ell_{\mathrm{ENT}})(1-\varrho_{\mathrm{ent}}),
\end{multline*}
which covers static-quota leakage, peak loss, premium-plan feature margin and enterprise slice revenue net of manual-billing recapture. Churn uses prepaid monthly churn $2.73\%$ and postpaid $0.85\%$, consistent with FY2024 disclosures (postpaid phone churn $0.76$--$0.89\%$) \citep{tmobile2024,att2024}. Contribution margins are $\PERT(40\%,50\%,60\%)$, consistent with reported EBITDA margins of $36.6$--$48.0\%$. The delayed launch is an enterprise network-slicing product with $V_0\sim\LN(\$700\,\mathrm{m},0.35)$, $I=\$250$\,m, $\sigma=0.45$ and $t_L=0.25$\,yr.

\paragraph{Remediation.} The design-around replaces predictive reservation with a non-practising quota algorithm: about $\rTelKSLOC$ KSLOC, $\sum\mathrm{SF}=\rTelSF$ ($E=\rTelEexp$), EAF $=\rTelEAF$ (\rTelEMlist). The nominal plan is $\rTelPMnom$ PM over $\rTelTDEVnom$ months; crashed, it is $\rTelPMcr$ PM over $\rTelTDEVcr$ months. Network durations are $A\sim\PERT(1,2,4)$, $C\sim\PERT(2,3,6)$, $D\sim\PERT(1,2,4)$ and $E\sim\PERT(1,2,4)$. Substitution adopts a third-party CHF: selection $\PERT(1,2,4)$, integration $\PERT(6,9,15)$, \$15\,m upfront, \$12\,m integration and \$6\,m per year.

\begin{small}
\setlength{\tabcolsep}{3pt}
\begin{longtable}{@{}p{1.9cm}p{4.5cm}p{3.3cm}p{1.0cm}p{5.0cm}@{}}
\caption{Operator~T calibration.}\label{tab:calib-tel}\\
\toprule
Symbol & Description & Distribution & Unit & Anchor\\ \midrule\endfirsthead
\toprule Symbol & Description & Distribution & Unit & Anchor\\ \midrule\endhead
\bottomrule\endfoot
\multicolumn{5}{l}{\textit{Restraint horizon}}\\
\texttt{T\_R} & restraint horizon (years) & $\mathrm{LN}(2,\,0.4)$ & yr & scenario assumption\\
\multicolumn{5}{l}{\textit{Switch-over outage}}\\
\texttt{p0} & P(charging outage at feature withdrawal) & $\mathrm{PERT}(0.1,0.25,0.45)$ &  & Rogers 2022/AT\&T 2024: change-induced\\
\texttt{D0} & switch-over outage duration & $\mathrm{LN}(2.5,\,0.9)$ & h & assumption\\
\multicolumn{5}{l}{\textit{Economic scale}}\\
\texttt{arpu\_post} & postpaid phone ARPU & $\mathrm{LN}(55,\,0.05)$ & USD/mo & AT\&T 4Q24 USD 56.72\\
\texttt{arpu\_pre} & prepaid ARPU & $\mathrm{LN}(35,\,0.1)$ & USD/mo & assumption\\
\texttt{margin} & contribution margin on service revenue & $\mathrm{PERT}(0.4,0.5,0.6)$ &  & EBITDA margins 36.6-48.0\% (FY24 filings)\\
\texttt{data\_share} & data share of prepaid revenue & $\mathrm{PERT}(0.55,0.65,0.75)$ &  & assumption\\
\texttt{leak} & static-quota revenue leakage share & $\mathrm{LN}(0.006,\,0.5)$ &  & assumption\\
\texttt{prem\_rate} & premium-plan incremental margin at risk & $\mathrm{LN}(40\,m,\,0.4)$ & USD/yr & assumption (1.5m subs x USD 5/mo x 50\%)\\
\texttt{ent\_n} & enterprise 5G accounts & $2500$ &  & assumption\\
\texttt{ent\_mrc} & enterprise monthly recurring charge & $\mathrm{LN}(6000,\,0.4)$ & USD/mo & assumption\\
\texttt{ent\_slice} & share of enterprise MRC needing slice charging & $\mathrm{PERT}(0.1,0.2,0.35)$ &  & assumption\\
\multicolumn{5}{l}{\textit{Customer response}}\\
\texttt{churn\_post} & postpaid monthly churn & $0.0085$ & 1/mo & FY24: 0.76-0.89\%\\
\texttt{churn\_pre} & prepaid monthly churn & $0.0273$ & 1/mo & T-Mobile FY24 2.73\%\\
\texttt{rho\_retry} & usage recaptured after denial & $\mathrm{PERT}(0.5,0.7,0.9)$ &  & assumption\\
\texttt{omega\_den} & outage-equivalence of peak denial & $\mathrm{PERT}(0.1,0.25,0.5)$ &  & assumption\\
\texttt{beta\_out} & prepaid churn hazard per visible outage hour & $\mathrm{LN}(2.0\text{\text{e-}}4,\,0.6)$ & 1/h & Opensignal 2025 (directional)\\
\texttt{beta\_prem} & premium-subscriber churn hazard (feature loss) & $\mathrm{LN}(0.02,\,0.6)$ & 1/yr & assumption\\
\multicolumn{5}{l}{\textit{Fallback capability}}\\
\texttt{cap\_iwf} & interworking capacity cap after scale-out (share of peak) & $\mathrm{PERT}(0.85,0.92,0.98)$ &  & architecture assessment\\
\multicolumn{5}{l}{\textit{Workaround operations}}\\
\texttt{iwf\_capex} & one-off fallback OCS/IWF scale-out cost & $\mathrm{LN}(10\,m,\,0.4)$ & USD & assumption\\
\texttt{ent\_recap} & manual-billing recapture & $\mathrm{PERT}(0.3,0.5,0.7)$ &  & assumption\\
\texttt{fte} & manual operations FTE & $\mathrm{PERT}(60,120,220)$ & FTE & assumption\\
\texttt{fte\_cost} & loaded cost per FTE-year & $120\,k$ & USD & assumption\\
\multicolumn{5}{l}{\textit{Outage economics}}\\
\texttt{ent\_aff} & share of enterprise accounts hit by an outage & $\mathrm{PERT}(0.4,0.6,0.8)$ &  & assumption\\
\texttt{cc\_hr} & contact-centre surge cost per outage hour & $\mathrm{LN}(150\,k,\,0.4)$ & USD/h & assumption\\
\texttt{K0} & fixed incident-response cost & $\mathrm{LN}(400\,k,\,0.5)$ & USD & assumption\\
\texttt{kappa} & signalling catch-up capacity multiple & $\mathrm{PERT}(1.2,1.5,2)$ &  & assumption\\
\texttt{v\_sess} & cost per session-hour of re-authorisation delay & $\mathrm{LN}(0.002,\,0.5)$ & USD & assumption\\
\texttt{p\_credit} & P(goodwill credit | visible outage $>$ 4 h) & $\mathrm{PERT}(0.2,0.5,0.8)$ &  & AT\&T 2024 USD 5/acct; Rogers 2022 5 days\\
\texttt{credit} & goodwill credit per prepaid account & $\mathrm{PERT}(1,2,5)$ & USD & AT\&T 2024 USD 5/acct\\
\multicolumn{5}{l}{\textit{Incident process}}\\
\texttt{lam\_bf} & but-for charging-incident rate & $2$ & 1/yr & assumption\\
\texttt{lam\_mult} & fallback incident-rate multiple & $\mathrm{LN}(2,\,0.4)$ &  & architecture CTMC (\S\ref{sec:arch}): 1.98\\
\texttt{D\_inc\_med} & median incident duration & $1.5$ & h & assumption\\
\texttt{D\_inc\_sig} & log-sd incident duration & $0.9$ &  & assumption\\
\multicolumn{5}{l}{\textit{Delayed launch}}\\
\texttt{V0} & PV at launch of enterprise slicing product & $\mathrm{LN}(700\,m,\,0.35)$ & USD & assumption; Ericsson/ADL USD 200bn addressable by 2030\\
\texttt{I\_launch} & launch investment & $250\,m$ & USD & assumption\\
\texttt{sig\_V} & project-value volatility & $0.45$ &  & assumption\\
\texttt{delta\_V} & opportunity-cost (payout+erosion) rate & $\mathrm{PERT}(0.1,0.15,0.22)$ & 1/yr & assumption\\
\texttt{t\_L} & planned launch date & $0.25$ & yr & assumption\\
\end{longtable}
\end{small}

\subsection{Case C: airline passenger-service system (Carrier~P)}\label{sec:air}
\paragraph{Platform.} Carrier~P is a network airline carrying $48$\,m passengers a year on about $1{,}250$ daily departures (an $18$-hour operating day, $150$ seats per departure, load factor $84\%$ \citep{iata2025}). Its PSS includes a \emph{real-time availability and re-shop engine} (RAE). The RAE drives continuous (dynamic) pricing, automated re-accommodation during irregular operations (IROPS), and a planned NDC dynamic-offer API.

\paragraph{Restraint and fallback.} The RAE is withdrawn. Pricing falls back to legacy class-based availability (filed fares, booking-class buckets), which fully supports booking ($\ell_{\mathrm{BOOK}}=1$) but not continuous pricing ($\ell_{\mathrm{CP}}=0$). Re-accommodation falls back to agent tools at capability $q_{\mathrm{manual}}\sim\PERT(0.30,0.45,0.60)$. The NDC programme is blocked (\cref{fig:graph-air}).

\begin{figure}[t]
\centering
\resizebox{\textwidth}{!}{%
\begin{tikzpicture}[>=Latex,font=\scriptsize,
  svc/.style={draw,rounded corners=2pt,text width=2.45cm,minimum height=0.66cm,align=center,inner sep=1.5pt},
  infra/.style={svc,fill=gray!12}, core/.style={svc,fill=blue!6}, chan/.style={svc,fill=blue!3},
  restr/.style={svc,fill=red!25,draw=red!70!black,very thick}, fb/.style={svc,fill=orange!18},
  capnode/.style={svc,fill=green!12}]
\node[infra] (CLOUD) at (0.0,6.4) {Cloud\\[-1pt]{\color{black!60}$\ell=1$}};
\node[infra] (NET) at (0.0,1.8) {Network\\[-1pt]{\color{black!60}$\ell=1$}};
\node[infra] (PSSDB) at (0.0,4.2) {PSS database\\[-1pt]{\color{black!60}$\ell=1$}};
\node[infra] (GDS) at (0.0,0.4) {GDS links\\[-1pt]{\color{black!60}$\ell=1$}};
\node[core] (INV) at (3.5,4.2) {Inventory\\[-1pt]$\wedge$, {\color{black!60}$\ell=1$}};
\node[restr] (RAE) at (7.0,5.6) {RT availability engine\\[-1pt]$\wedge$, {\color{red!70!black}$\ell=0$}};
\node[fb] (AVS) at (7.0,4.2) {Legacy AVS\\[-1pt]$\wedge$, {\color{black!60}$\ell=1$}};
\node[core] (FARES) at (3.5,7.4) {Filed fares\\[-1pt]$\wedge$, {\color{black!60}$\ell=1$}};
\node[core] (PRICE) at (10.5,4.9) {Offer pricing\\[-1pt]alt$_{(1,1)}$, {\color{black!60}$\ell=1$}};
\node[core] (CP) at (10.5,7.4) {Continuous pricing\\[-1pt]$\wedge$, {\color{red!70!black}$\ell=0$}};
\node[fb] (MTOOL) at (7.0,2.9) {Agent re-accom. tool\\[-1pt]$\wedge$, {\color{black!60}$\ell=1$}};
\node[core] (REACC) at (10.5,2.9) {Re-accommodation\\[-1pt]alt$_{(1,0.45)}$, {\color{red!70!black}$\ell=0.45$}};
\node[core] (DCS) at (7.0,1.5) {Departure control\\[-1pt]$\wedge$, {\color{black!60}$\ell=1$}};
\node[core] (NDCAPI) at (14.0,6.2) {NDC offer API\\[-1pt]$\wedge$, {\color{red!70!black}$\ell=0$}};
\node[capnode] (CAPBOOK) at (17.5,4.4) {Booking \& ticketing\\[-1pt]$\wedge$, {\color{black!60}$\ell=1$}};
\node[capnode] (CAPCP) at (17.5,7.6) {Dynamic pricing\\[-1pt]$\wedge$, {\color{red!70!black}$\ell=0$}};
\node[capnode] (CAPIROPS) at (17.5,2.9) {Disruption recovery\\[-1pt]$\wedge$, {\color{red!70!black}$\ell=0.45$}};
\node[capnode] (CAPCHK) at (17.5,1.3) {Check-in/boarding\\[-1pt]$\wedge$, {\color{black!60}$\ell=1$}};
\node[capnode] (CAPNDC) at (17.5,6.2) {NDC retailing (planned)\\[-1pt]$\wedge$, {\color{red!70!black}$\ell=0$}};
\draw[->,black!55] (PSSDB.east) -- (INV.west);
\draw[->,black!55] (INV.east) -- (RAE.west);
\draw[->,black!55] (CLOUD.east) -- (RAE.west);
\draw[->,black!55] (INV.east) -- (AVS.west);
\draw[->,black!55] (CLOUD.east) -- (FARES.west);
\draw[->,red!70!black,thick] (RAE.east) -- (PRICE.west);
\draw[->,densely dashed,orange!70!black,thick] (AVS.east) -- (PRICE.west);
\draw[->,red!70!black,thick] (RAE.east) -- (CP.west);
\draw[->,black!55] (FARES.east) -- (CP.west);
\draw[->,black!55] (INV.east) -- (MTOOL.west);
\draw[->,red!70!black,thick] (RAE.east) -- (REACC.west);
\draw[->,densely dashed,orange!70!black,thick] (MTOOL.east) -- (REACC.west);
\draw[->,black!55] (INV.east) -- (DCS.west);
\draw[->,black!55] (NET.east) -- (DCS.west);
\draw[->,red!70!black,thick] (RAE.east) -- (NDCAPI.west);
\draw[->,black!55] (PRICE.east) -- (NDCAPI.west);
\draw[->,black!55] (INV.east) -- (CAPBOOK.west);
\draw[->,black!55] (PRICE.east) -- (CAPBOOK.west);
\draw[->,black!55] (GDS.east) -- (CAPBOOK.west);
\draw[->,black!55] (CP.east) -- (CAPCP.west);
\draw[->,black!55] (REACC.east) -- (CAPIROPS.west);
\draw[->,black!55] (DCS.east) -- (CAPIROPS.west);
\draw[->,black!55] (DCS.east) -- (CAPCHK.west);
\draw[->,black!55] (NDCAPI.east) -- (CAPNDC.west);
\end{tikzpicture}}
\caption{Carrier~P dependency graph. Red: restrained RAE. Orange: legacy fallbacks (availability server, agent re-accommodation tool).}
\label{fig:graph-air}
\end{figure}
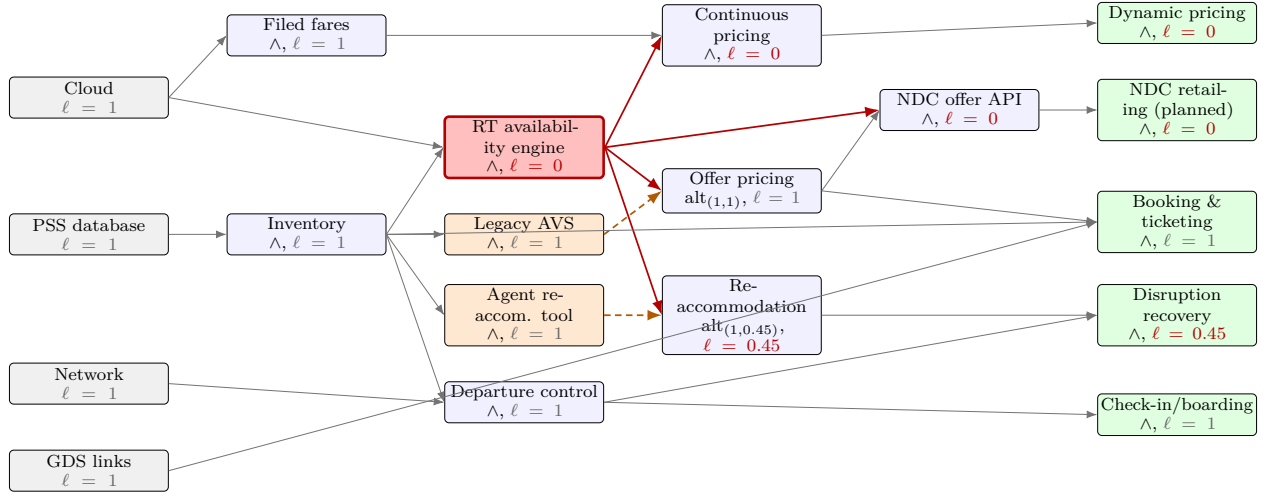

\paragraph{Channels.} Departures occur at $\lambda_f=1250/18\approx69.4$ per hour. In an outage, a share $\pi_1$ of departures in the window is cancelled. Beyond 3\,h, cascade cancellations accrue at $\pi_2\lambda_f$ per hour. Each cancellation costs $c_{\mathrm{cxl}}\sim\LN(\$60\,\mathrm{k},0.25)$, anchored to Delta's July~2024 disruption (about \$500\,m net over 7{,}000 cancellations) and Southwest's December~2022 meltdown (\$800\,m over 16{,}700) \citep{delta2024q3,southwest2022}. It also costs $126$ passengers times a downstream revenue loss of $\LN(\$120,0.5)$ each, scaled to Southwest's Q1~2023 revenue impact of \$325\,m over about 2\,m stranded passengers. Hence $a=\lambda_f\pi_1(c_{\mathrm{cxl}}+126\,\mathrm{ds})$ and the hinge at $3$\,h is $h=\lambda_f\pi_2(c_{\mathrm{cxl}}+126\,\mathrm{ds})$. Gate-delay backlog contributes $b=w_{\mathrm{flt}}\tfrac12\psi\lambda_f(1+\psi/(\kappa-1))$ with $\psi=1-\pi_1$. Here $w_{\mathrm{flt}}\sim\LN(\$3{,}500/\mathrm{h},0.30)$, consistent with A4A's CY2025 delay cost of \$98.41 per block minute, of which crew, ownership and other account for \$50.72 \citep{a4a2025}. Continuous pricing yields a revenue uplift $\mathrm{Tri}(0.5\%,1\%,2\%)$ on the share of revenue in continuous-pricing markets, consistent with the 1--2\% range reported in IATA's Dynamic Offers guidance \citep{iata2021}. IROPS events arrive at rate $\lambda_I\sim\PERT(15,25,40)$ per year. Event sizes are lognormal with median $\mathrm{pax}_I$ and log-sd $0.9$. The per-passenger cost of degraded re-accommodation is
\[ (1-\ell_{\mathrm{IROPS}})\big[k_{\mathrm{on}}c_{\mathrm{hotel}}+k_{\mathrm{ref}}\,\mathrm{fare}\,(1-s_{\mathrm{vc}})+s_{\mathrm{EU}}\,s_{\neg\mathrm{ext}}\,k_{3h}\,\mathrm{comp}_{\mathrm{EU}}+h_{\mathrm{ag}}w_{\mathrm{ag}}\big], \]
which covers overnight care, refunds under the 2024 DOT rule \citep{dot2024}, EU~261 compensation \citep{eu261} and agent time. Downstream revenue lost from stranded passengers, $(1-\ell_{\mathrm{IROPS}})k_{\mathrm{on}}\,\mathrm{ds}$, enters the attrition channel. The delayed launch is the NDC dynamic-offer programme with $V_0\sim\LN(\$600\,\mathrm{m},0.30)$, $I=\$150$\,m, $\sigma=0.35$ and $t_L=0.33$\,yr.

\paragraph{Remediation.} The design-around builds a non-practising availability computation: about $\rAirKSLOC$ KSLOC, $\sum\mathrm{SF}=\rAirSF$ ($E=\rAirEexp$), EAF $=\rAirEAF$ (\rAirEMlist). The nominal plan is $\rAirPMnom$ PM over $\rAirTDEVnom$ months; crashed, it is $\rAirPMcr$ PM over $\rAirTDEVcr$ months. Network durations are $A\sim\PERT(1.5,2.5,4)$, $C\sim\PERT(3,4,7)$, $D\sim\PERT(2,3,5)$ and $E\sim\PERT(1.5,2.5,5)$. Substitution uses a third-party offer engine: selection $\PERT(2,4,6)$, integration $\PERT(10,14,24)$, \$15\,m upfront, \$25\,m integration and \$5\,m per year. Integration into a PSS is slow, so substitution is not faster than crashing in this case.

\begin{small}
\setlength{\tabcolsep}{3pt}
\begin{longtable}{@{}p{1.9cm}p{4.5cm}p{3.3cm}p{1.0cm}p{5.0cm}@{}}
\caption{Carrier~P calibration.}\label{tab:calib-air}\\
\toprule
Symbol & Description & Distribution & Unit & Anchor\\ \midrule\endfirsthead
\toprule Symbol & Description & Distribution & Unit & Anchor\\ \midrule\endhead
\bottomrule\endfoot
\multicolumn{5}{l}{\textit{Restraint horizon}}\\
\texttt{T\_R} & restraint horizon (years) & $\mathrm{LN}(2,\,0.4)$ & yr & scenario assumption\\
\multicolumn{5}{l}{\textit{Switch-over outage}}\\
\texttt{p0} & P(PSS outage at RAE withdrawal) & $\mathrm{PERT}(0.1,0.25,0.45)$ &  & United 2012 PSS cut-over; BA 2017\\
\texttt{D0} & switch-over outage duration & $\mathrm{LN}(2,\,0.8)$ & h & assumption\\
\multicolumn{5}{l}{\textit{Economic scale}}\\
\texttt{fare} & revenue per passenger & $\mathrm{LN}(255,\,0.1)$ & USD & BTS 2024 USD 384 per domestic itinerary\\
\texttt{cp\_uplift} & continuous-pricing revenue uplift & $\mathrm{Tri}(0.005,0.01,0.02)$ &  & IATA 2021 / MIT: 1-2\%\\
\texttt{cp\_share} & revenue share in continuous-pricing markets & $\mathrm{PERT}(0.5,0.7,0.85)$ &  & assumption\\
\texttt{vc\_share} & variable cost share of a lost fare & $\mathrm{PERT}(0.15,0.25,0.35)$ &  & assumption\\
\multicolumn{5}{l}{\textit{Fallback capability}}\\
\texttt{cap\_manual} & manual re-accommodation capability cap & $\mathrm{PERT}(0.3,0.45,0.6)$ &  & architecture assessment\\
\multicolumn{5}{l}{\textit{IROPS process}}\\
\texttt{lam\_I} & significant IROPS events per year & $\mathrm{PERT}(15,25,40)$ & 1/yr & DOT 2024: 1.4\% cancellations\\
\texttt{pax\_I} & median passengers to re-accommodate per event (event log-sd 0.9) & $\mathrm{LN}(12\,k,\,0.3)$ & pax & assumption\\
\texttt{k\_on} & overnight-stranding increase per unit capability loss & $\mathrm{PERT}(0.06,0.12,0.22)$ &  & assumption\\
\texttt{k\_ref} & refund-uptake increase per unit capability loss & $\mathrm{PERT}(0.03,0.07,0.14)$ &  & DOT refund rule 2024\\
\texttt{k\_3h} & 3h+ arrival-delay increase per unit capability loss & $\mathrm{PERT}(0.05,0.1,0.18)$ &  & EU 261/2004 Art. 7\\
\texttt{hotel} & care cost per overnight passenger & $\mathrm{PERT}(150,220,320)$ & USD & EU 261 Art. 9; assumption\\
\texttt{eu\_share} & share of disrupted pax under EU 261 & $0.1$ &  & assumption\\
\texttt{non\_extra} & share not 'extraordinary circumstances' & $\mathrm{PERT}(0.25,0.4,0.6)$ &  & assumption\\
\texttt{eu\_comp} & average EU 261 compensation & $400$ & EUR & EU 261: EUR 250/400/600\\
\multicolumn{5}{l}{\textit{Workaround operations}}\\
\texttt{h\_ag} & extra agent hours per disrupted pax per unit loss & $\mathrm{PERT}(0.05,0.12,0.25)$ & h & assumption\\
\texttt{wage} & loaded agent cost per hour & $48$ & USD/h & assumption\\
\texttt{fte} & manual operations FTE (pricing/RM, agents) & $\mathrm{PERT}(80,150,260)$ & FTE & assumption\\
\texttt{fte\_cost} & loaded cost per FTE-year & $85\,k$ & USD & assumption\\
\multicolumn{5}{l}{\textit{Customer response}}\\
\texttt{downstream} & future revenue lost per stranded passenger & $\mathrm{LN}(120,\,0.5)$ & USD & Southwest Q1-2023: USD 325m / $\sim$2m stranded\\
\multicolumn{5}{l}{\textit{Outage economics}}\\
\texttt{pi1} & cancelled share of flights in outage window & $\mathrm{PERT}(0.05,0.1,0.2)$ &  & assumption\\
\texttt{pi2} & cascade cancellations per outage hour beyond 3 h (x hourly departures) & $\mathrm{PERT}(0.5,1,1.6)$ &  & Delta 2016: 2,300 cxl/3 days\\
\texttt{c\_cxl} & cost per cancellation (refunds, care, crew) & $\mathrm{LN}(60\,k,\,0.25)$ & USD & Delta 3Q24: (380+170-50)/7000; Southwest 4Q22: 800/16,700\\
\texttt{w\_flt} & airline cost per flight-hour of gate delay & $\mathrm{LN}(3500,\,0.3)$ & USD/h & A4A CY2025: crew+ownership+other USD 50.72/min\\
\texttt{kappa} & effective recovery-capacity multiple & $\mathrm{PERT}(1.4,1.8,2.5)$ &  & assumption\\
\texttt{K0} & fixed incident-response cost & $\mathrm{LN}(500\,k,\,0.5)$ & USD & assumption\\
\multicolumn{5}{l}{\textit{Incident process}}\\
\texttt{lam\_bf} & but-for PSS-incident rate & $0.5$ & 1/yr & assumption\\
\texttt{lam\_mult} & fallback incident-rate multiple & $\mathrm{LN}(2.4,\,0.45)$ &  & architecture CTMC (\S\ref{sec:arch}): 2.43\\
\texttt{D\_inc\_med} & median incident duration & $1.5$ & h & assumption\\
\texttt{D\_inc\_sig} & log-sd incident duration & $0.9$ &  & assumption\\
\multicolumn{5}{l}{\textit{Delayed launch}}\\
\texttt{V0} & PV at launch of NDC dynamic-offer programme & $\mathrm{LN}(600\,m,\,0.3)$ & USD & assumption\\
\texttt{I\_launch} & launch investment & $150\,m$ & USD & assumption\\
\texttt{sig\_V} & project-value volatility & $0.35$ &  & assumption\\
\texttt{delta\_V} & opportunity-cost (payout+erosion) rate & $\mathrm{PERT}(0.08,0.12,0.18)$ & 1/yr & assumption\\
\texttt{t\_L} & planned launch date & $0.33$ & yr & assumption\\
\end{longtable}
\end{small}

\paragraph{Dependence structure.} The Gaussian-copula correlations encode mechanistic links. In all cases, switch-over failure probability and the incident multiple have $\rho=0.4$ (both reflect fallback immaturity), and design-around size and estimation error have $0.3$. For Bank~A, hourly outage cost and redress have $0.4$; the outage and degradation churn hazards have $0.5$; and retail and SME degradation hazards have $0.6$. For Operator~T, the prepaid and premium hazards have $0.5$, and margin and data share $0.2$. For Carrier~P, the two cancellation rates have $0.5$, and the overnight, refund and 3-hour delay increments have $0.5$.

\section{Experimental Evaluation}\label{sec:experiments}
\subsection{Research questions and design}\label{sec:design}
The experiments address eight questions:
\begin{enumerate}[label=\textbf{RQ\arabic*},leftmargin=3.2em,itemsep=0pt]
\item \emph{Correctness.} Do the implemented closed forms, fixed-point iterations and estimators agree with independent numerical methods and with the theorems?
\item \emph{Architecture.} What incident multiple do the fallback architectures imply?
\item \emph{Distributions and strategy.} What are the mitigated loss distribution and the loss-minimising response for each platform?
\item \emph{Decomposition.} Which channels carry expected loss, and which carry tail risk?
\item \emph{Plug-in bias.} How biased is the conventional deterministic plug-in estimate, and does it select the right response?
\item \emph{Drivers.} Which mechanisms dominate epistemic uncertainty?
\item \emph{Security.} What security does \cref{thm:bond} imply, and how does it compare with posting the expected loss?
\item \emph{Calibration econometrics and precision.} Can customer-response effects be recovered from staggered panel data, and how precise are the Monte Carlo estimates?
\end{enumerate}
Each case is evaluated with $n=\rN$ joint draws of $(\theta,\omega)$ per strategy (and per CCFH policy for Operator~T), using common random numbers. All monetary values are in USD millions, discounted to $t=0$. The ``best'' strategy is $s^\star=\arg\min_s\overline{L(s)}$. The full pipeline runs in $\rTotalRuntime$\,s (Bank~A $\rBankRuntime$\,s, Operator~T $\rTelRuntime$\,s, Carrier~P $\rAirRuntime$\,s, plus econometrics).

\subsection{RQ1: Verification}\label{sec:verif}
\Cref{tab:verif} summarises the verification suite. Each analytical component is checked against an independent method.

\emph{Fluid backlog (\cref{thm:backlog}).} The fluid area $W(D)$ was compared with a stochastic time-stepped queue ($\lambda=2000$/h, $c=3000$/h, $\psi=0.8$, $4000$ steps per outage, $12$ replications). The maximum relative deviation is $\rVBacklogMaxErr\%$, which is within two standard errors at every $D$ (\cref{tab:v-backlog}).

\emph{Lattice semantics (\cref{thm:lattice}).} On the toy platform of \cref{ex:toy}, exhaustive enumeration over $\mathcal{C}^7$ finds the fixed points. With the bootstrap disabled there are three, and Kleene iteration from $\bot$ and $\top$ returns the least and greatest of them. In every scenario the iteration count is well below the bound $|V|(|\mathcal{C}|-1)=14$ (\cref{tab:v-fix}).

\emph{Uniformization (\cref{prop:unif}).} On a $4$-of-$6$ cluster with two crews, transient probabilities agree with the matrix exponential to $\rVUnifErr$ (maximum absolute error over $t\in\{1,10,100,1000\}$\,h). Cumulative occupation agrees with $4001$-point trapezoidal quadrature to $\rVUnifCumErr$, an error dominated by the quadrature.

\emph{Acute expectation (\cref{thm:jensen}).} The closed form \eqref{eq:acuteexp} matches adaptive quadrature over the threshold-split domain to a relative error of at most $\rVAcuteMaxErr$ (\cref{tab:v-acute}).

\emph{Delayed option (\cref{prop:delayclosed,thm:delay}).} Over nine configurations, the closed form agrees with 64-node Gauss--Hermite quadrature to a relative error of at most $\rVOptQuadErr$, and with $2\times10^6$-path Monte Carlo with $|z|\le\rVOptMaxZ$. The bounds $0\le\ell\le\ell_{\mathrm{DCF}}$ hold in all configurations (\cref{tab:v-option}).

\emph{Quantile intervals (\cref{prop:qci}).} Over $2000$ lognormal samples of size $2000$, the order-statistic interval for the $95$th percentile covers the true quantile in $\rVQcov\%$ of cases. Its exact (discrete) nominal coverage is $\rVQnom\%$.

\emph{Euler allocation (\cref{prop:euler}).} The tail contributions sum to $\CVaR_{0.95}$ to within $\rVEuler$.

\begin{table}[t]
\centering\footnotesize
\caption{Verification summary (RQ1).}\label{tab:verif}
\begin{tabularx}{\textwidth}{@{}lXl@{}}
\toprule
Component & Independent check & Max.\ discrepancy\\
\midrule
Fluid backlog $W(D)$ & stochastic time-stepped queue, 12 replications & $\rVBacklogMaxErr\%$ (within $2$ s.e.)\\
Lattice fixed points & exhaustive enumeration of $\mathcal{C}^V$ & exact (lfp least, gfp greatest)\\
Uniformization, transient & matrix exponential (\code{scipy.linalg.expm}) & $\rVUnifErr$\\
Uniformization, cumulative & $4001$-point trapezoid of $e^{\mathbf{Q}u}$ & $\rVUnifCumErr$\\
Acute expectation \eqref{eq:acuteexp} & adaptive quadrature split at thresholds & $\rVAcuteMaxErr$ (relative)\\
Delayed option \eqref{eq:Gclosed} & Gauss--Hermite (64) / Monte Carlo ($2\times10^6$) & $\rVOptQuadErr$ / $|z|\le\rVOptMaxZ$\\
Order-statistic CI & 2000 replications, $n=2000$, $p=0.95$ & coverage $\rVQcov\%$ vs.\ $\rVQnom\%$\\
Euler allocation & $\sum_k$ contributions vs.\ $\CVaR_{0.95}$ & $\rVEuler$\\
\bottomrule
\end{tabularx}
\end{table}

\begin{table}[t]
\centering\footnotesize
\begin{minipage}[t]{0.46\textwidth}\centering
\caption{Fluid backlog vs.\ stochastic simulation (transaction-hours).}\label{tab:v-backlog}
\begin{tabular}{@{}rrrrr@{}}
\toprule
$D$ (h) & $W(D)$ & sim.\ mean & s.e. & rel.\ err.\\
\midrule
0.5 & 520 & 527.5 & 9.9 & +1.44\%\\
1.0 & 2,080 & 2,060.1 & 21.9 & -0.96\%\\
2.0 & 8,320 & 8,234.7 & 68.3 & -1.03\%\\
4.0 & 33,280 & 33,014.3 & 180.8 & -0.80\%\\
\bottomrule
\end{tabular}
\end{minipage}\hfill
\begin{minipage}[t]{0.50\textwidth}\centering
\caption{Acute expectation: closed form vs.\ quadrature (USD m; $K_0=0.5$, $a=2$, $b=0.037$, step $10$ at $12$\,h, hinge $0.8$/h after $4$\,h).}\label{tab:v-acute}
\begin{tabular}{@{}rrrrr@{}}
\toprule
median & $s$ & closed form & quadrature & rel.\ err.\\
\midrule
6.0 & 1.0 & 37.7454 & 37.7454 & $-2.6\times 10^{-13}$\\
3.0 & 0.9 & 13.1849 & 13.1849 & $-1.7\times 10^{-12}$\\
1.5 & 0.9 & 5.8458 & 5.8458 & $-5.3\times 10^{-12}$\\
2.0 & 0.8 & 7.0917 & 7.0917 & $-7.2\times 10^{-14}$\\
\bottomrule
\end{tabular}
\end{minipage}
\end{table}

\begin{table}[t]
\centering\footnotesize
\caption{Fixed points of the toy platform (\cref{ex:toy}). Where they differ, entries show lfp/gfp. \#FP is the number of fixed points found by exhaustive enumeration. Iterations are Kleene steps from $\bot$/$\top$.}\label{tab:v-fix}
\begin{tabular}{@{}lccccccccccc@{}}
\toprule
Scenario & DNS & AUTH & CFG & BOOT & API & APP & CACHE & \#FP & iter.\ lfp/gfp & bound\\
\midrule
nominal & 1 & 1 & 1 & 1 & 1 & 1 & 1 & 1 & 5/0 & 14\\
no\_bootstrap & 1 & \textbf{0}/\textbf{1} & \textbf{0}/\textbf{1} & 0 & \textbf{0}/\textbf{1} & \textbf{0.5}/\textbf{1} & 1 & 3 & 2/1 & 14\\
dns\_down & 0 & 0 & 1 & 1 & 0 & 0.5 & 1 & 1 & 2/4 & 14\\
\bottomrule
\end{tabular}
\end{table}

\begin{table}[t]
\centering\footnotesize
\caption{Delayed option value $G$ (USD m): closed form \eqref{eq:Gclosed} vs.\ Gauss--Hermite quadrature and Monte Carlo ($2\times10^6$ paths); forced-delay loss $\ell$ and static bound $\ell_{\mathrm{DCF}}$.}\label{tab:v-option}
\begin{tabular}{@{}lrrrrrrrr@{}}
\toprule
Case & $\Delta$ (yr) & $G$ closed & quad.\ rel.\ err. & MC mean & MC s.e. & $z$ & $\ell(\Delta)$ & $\ell_{\mathrm{DCF}}(\Delta)$\\
\midrule
Bank A & 0.5 & 278.7 & $+3.3\times 10^{-10}$ & 278.7 & 0.057 & -0.45 & 21.3 & 21.3\\
Bank A & 1.0 & 258.7 & $-3.3\times 10^{-7}$ & 258.7 & 0.077 & +0.29 & 41.3 & 41.3\\
Bank A & 2.0 & 222.5 & $-7.5\times 10^{-6}$ & 222.4 & 0.099 & -0.91 & 77.5 & 77.7\\
Operator T & 0.5 & 405.6 & $+1.4\times 10^{-5}$ & 405.6 & 0.148 & +0.08 & 44.4 & 45.6\\
Operator T & 1.0 & 368.4 & $-5.2\times 10^{-5}$ & 368.7 & 0.197 & +1.25 & 81.6 & 87.7\\
Operator T & 2.0 & 308.0 & $-4.7\times 10^{-5}$ & 308.1 & 0.247 & +0.44 & 142.0 & 162.2\\
Carrier P & 0.5 & 418.0 & $-4.9\times 10^{-8}$ & 418.0 & 0.100 & -0.08 & 32.0 & 32.0\\
Carrier P & 1.0 & 388.1 & $-1.8\times 10^{-6}$ & 388.2 & 0.136 & +1.07 & 61.9 & 62.0\\
Carrier P & 2.0 & 334.6 & $-9.8\times 10^{-6}$ & 334.9 & 0.175 & +1.93 & 115.4 & 116.5\\
\bottomrule
\end{tabular}
\end{table}

\subsection{RQ2: Architecture-implied incident multiples}\label{sec:arch}
\Cref{tab:arch} reports the CTMC analysis. In all three platforms, common-mode software failures dominate the failure frequency of the but-for architecture. Redundant clusters contribute frequencies two to three orders of magnitude smaller. The fallbacks are worse for three reasons. They add non-redundant components (stand-in processors, batch schedulers, legacy availability servers). Manual workaround activity raises the change-driven common-mode rate (by $\times1.35$ for Bank~A). And interworking introduces a second software stack (Operator~T). The implied multiples are $\hat m=\rBankArchMult$ (Bank~A), $\rTelArchMult$ (Operator~T) and $\rAirArchMult$ (Carrier~P). Expected downtime rises from $\rBankArchDTbf$ to $\rBankArchDTfb$, from $\rTelArchDTbf$ to $\rTelArchDTfb$, and from $\rAirArchDTbf$ to $\rAirArchDTfb$ minutes per year. The 30-day interval reliability $\Prob(\Always^{\le 720\mathrm{h}}\mathsf{up})$ falls from $\rBankArchRbf$ to $\rBankArchRfb$, $\rTelArchRbf$ to $\rTelArchRfb$, and $\rAirArchRbf$ to $\rAirArchRfb$ (\cref{fig:rel}). These multiples are the medians of the lognormal priors on $m$. The priors' dispersion (log-sd $0.40$--$0.45$) represents uncertainty about block rates.

\begin{table}[t]
\centering\footnotesize
\caption{CTMC reliability analysis of but-for and fallback architectures. $A$: steady-state availability; $f$: failure frequency (per year); MDT: mean down time (h); $R(30\,\mathrm{d})$: CSL interval reliability; last column: expected downtime (min/yr).}\label{tab:arch}
\begin{tabular}{@{}lrrrrr@{}}
\toprule
Block & $A$ & $f$ (1/yr) & MDT (h) & $R(30\,\mathrm{d})$ & min/yr\\
\midrule
\multicolumn{6}{l}{\textbf{Bank A (core banking)}: implied multiple $\hat m = 2.525$}\\
\quad but-for: RTPE cluster (4-of-6) & 0.99999986 & 0.0012 & 1.00 & 0.99990 & \\
\quad but-for: Ledger DB pair (1-of-2) & 0.99999645 & 0.0078 & 4.00 & 0.99936 & \\
\quad but-for: Core software, common mode & 0.99966450 & 0.9797 & 3.00 & 0.92261 & \\
\quad \textit{but-for system} & \textit{0.9996608} & \textit{0.989} & \textit{3.01} & \textit{0.9219} & \textit{178}\\
\quad fallback: Memo-post stand-in (2-of-2) & 0.99970004 & 0.8757 & 3.00 & 0.93053 & \\
\quad fallback: Batch scheduler (1-of-1) & 0.99986668 & 0.2920 & 4.00 & 0.97629 & \\
\quad fallback: Ledger DB pair (1-of-2) & 0.99999645 & 0.0078 & 4.00 & 0.99936 & \\
\quad fallback: Core software, common mode x1.35 & 0.99954712 & 1.3224 & 3.00 & 0.89696 & \\
\quad \textit{fallback system} & \textit{0.9991105} & \textit{2.497} & \textit{3.12} & \textit{0.8143} & \textit{467}\\
\midrule
\multicolumn{6}{l}{\textbf{Operator T (mobile network operator)}: implied multiple $\hat m = 1.975$}\\
\quad but-for: CHF pods (8-of-10) & 0.99999998 & 0.0004 & 0.50 & 0.99997 & \\
\quad but-for: Charging software, common mode & 0.99965765 & 1.9993 & 1.50 & 0.84842 & \\
\quad \textit{but-for system} & \textit{0.9996576} & \textit{2.000} & \textit{1.50} & \textit{0.8484} & \textit{180}\\
\quad fallback: Nchf-Gy IWF (3-of-4) & 0.99999569 & 0.0252 & 1.50 & 0.99794 & \\
\quad fallback: 4G OCS (2-of-3) & 0.99999401 & 0.0262 & 2.00 & 0.99785 & \\
\quad fallback: Charging software, common mode & 0.99965765 & 1.9993 & 1.50 & 0.84842 & \\
\quad fallback: IWF/OCS software, common mode & 0.99967476 & 1.8994 & 1.50 & 0.85542 & \\
\quad \textit{fallback system} & \textit{0.9993222} & \textit{3.949} & \textit{1.50} & \textit{0.7227} & \textit{356}\\
\midrule
\multicolumn{6}{l}{\textbf{Carrier P (airline PSS)}: implied multiple $\hat m = 2.432$}\\
\quad but-for: RAE nodes (3-of-4) & 0.99999700 & 0.0175 & 1.50 & 0.99857 & \\
\quad but-for: PSS software, common mode & 0.99991781 & 0.4800 & 1.50 & 0.96132 & \\
\quad \textit{but-for system} & \textit{0.9999148} & \textit{0.497} & \textit{1.50} & \textit{0.9599} & \textit{45}\\
\quad fallback: Legacy AVS (1-of-1) & 0.99993334 & 0.2920 & 2.00 & 0.97629 & \\
\quad fallback: Agent re-accommodation tool (2-of-2) & 0.99992500 & 0.4380 & 1.50 & 0.96464 & \\
\quad fallback: PSS software, common mode & 0.99991781 & 0.4800 & 1.50 & 0.96132 & \\
\quad \textit{fallback system} & \textit{0.9997762} & \textit{1.210} & \textit{1.62} & \textit{0.9053} & \textit{118}\\
\bottomrule
\end{tabular}
\end{table}

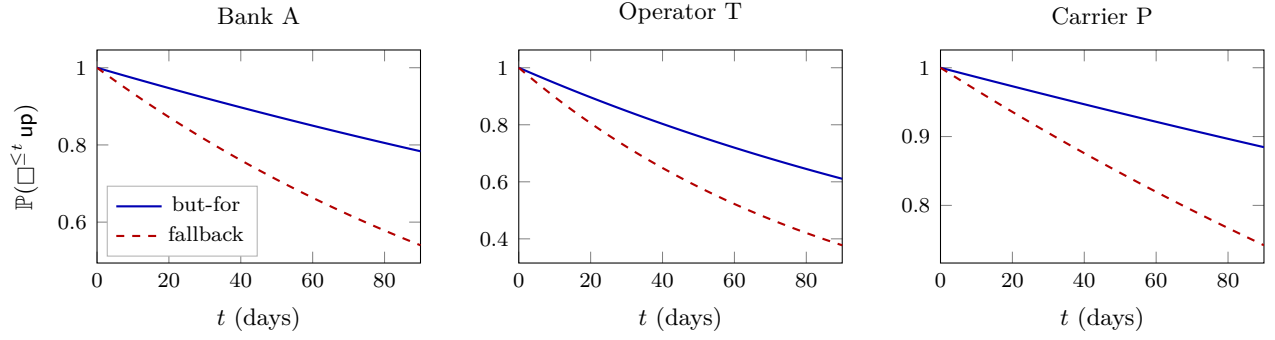
\begin{figure}[t]
\centering
\begin{tikzpicture}
\begin{groupplot}[group style={group size=3 by 1,horizontal sep=1.3cm},width=0.355\textwidth,height=4.4cm,xlabel={$t$ (days)},xmin=0,xmax=90,legend style={font=\scriptsize,at={(0.03,0.03)},anchor=south west}]
\nextgroupplot[title={Bank A},ylabel={$\Prob(\Always^{\le t}\mathsf{up})$}]
\addplot[blue!70!black,thick] coordinates {(0.0,1.00000) (2.5,0.99325) (5.0,0.98654) (7.5,0.97988) (10.0,0.97327) (12.5,0.96670) (15.0,0.96017) (17.5,0.95369) (20.0,0.94725) (22.5,0.94086) (25.0,0.93451) (27.5,0.92820) (30.0,0.92193) (32.5,0.91571) (35.0,0.90953) (37.5,0.90339) (40.0,0.89729) (42.5,0.89123) (45.0,0.88521) (47.5,0.87924) (50.0,0.87330) (52.5,0.86741) (55.0,0.86155) (57.5,0.85574) (60.0,0.84996) (62.5,0.84422) (65.0,0.83852) (67.5,0.83286) (70.0,0.82724) (72.5,0.82165) (75.0,0.81611) (77.5,0.81060) (80.0,0.80513) (82.5,0.79969) (85.0,0.79429) (87.5,0.78893) (90.0,0.78360)};
\addplot[red!70!black,thick,dashed] coordinates {(0.0,1.00000) (2.5,0.98303) (5.0,0.96635) (7.5,0.94995) (10.0,0.93383) (12.5,0.91799) (15.0,0.90241) (17.5,0.88709) (20.0,0.87204) (22.5,0.85724) (25.0,0.84270) (27.5,0.82840) (30.0,0.81434) (32.5,0.80052) (35.0,0.78694) (37.5,0.77358) (40.0,0.76046) (42.5,0.74755) (45.0,0.73487) (47.5,0.72240) (50.0,0.71014) (52.5,0.69809) (55.0,0.68624) (57.5,0.67460) (60.0,0.66315) (62.5,0.65190) (65.0,0.64083) (67.5,0.62996) (70.0,0.61927) (72.5,0.60876) (75.0,0.59843) (77.5,0.58828) (80.0,0.57829) (82.5,0.56848) (85.0,0.55883) (87.5,0.54935) (90.0,0.54003)};
\legend{but-for,fallback}
\nextgroupplot[title={Operator T}]
\addplot[blue!70!black,thick] coordinates {(0.0,1.00000) (2.5,0.98639) (5.0,0.97297) (7.5,0.95973) (10.0,0.94667) (12.5,0.93379) (15.0,0.92108) (17.5,0.90855) (20.0,0.89618) (22.5,0.88399) (25.0,0.87196) (27.5,0.86009) (30.0,0.84839) (32.5,0.83684) (35.0,0.82546) (37.5,0.81422) (40.0,0.80314) (42.5,0.79222) (45.0,0.78143) (47.5,0.77080) (50.0,0.76031) (52.5,0.74997) (55.0,0.73976) (57.5,0.72969) (60.0,0.71976) (62.5,0.70997) (65.0,0.70031) (67.5,0.69078) (70.0,0.68138) (72.5,0.67211) (75.0,0.66296) (77.5,0.65394) (80.0,0.64504) (82.5,0.63626) (85.0,0.62760) (87.5,0.61906) (90.0,0.61064)};
\addplot[red!70!black,thick,dashed] coordinates {(0.0,1.00000) (2.5,0.97330) (5.0,0.94731) (7.5,0.92202) (10.0,0.89740) (12.5,0.87344) (15.0,0.85012) (17.5,0.82742) (20.0,0.80533) (22.5,0.78382) (25.0,0.76289) (27.5,0.74253) (30.0,0.72270) (32.5,0.70340) (35.0,0.68462) (37.5,0.66634) (40.0,0.64855) (42.5,0.63123) (45.0,0.61438) (47.5,0.59797) (50.0,0.58201) (52.5,0.56647) (55.0,0.55134) (57.5,0.53662) (60.0,0.52229) (62.5,0.50835) (65.0,0.49478) (67.5,0.48156) (70.0,0.46871) (72.5,0.45619) (75.0,0.44401) (77.5,0.43216) (80.0,0.42062) (82.5,0.40939) (85.0,0.39846) (87.5,0.38782) (90.0,0.37746)};
\nextgroupplot[title={Carrier P}]
\addplot[blue!70!black,thick] coordinates {(0.0,1.00000) (2.5,0.99660) (5.0,0.99321) (7.5,0.98983) (10.0,0.98646) (12.5,0.98311) (15.0,0.97976) (17.5,0.97643) (20.0,0.97311) (22.5,0.96980) (25.0,0.96650) (27.5,0.96321) (30.0,0.95994) (32.5,0.95667) (35.0,0.95342) (37.5,0.95018) (40.0,0.94694) (42.5,0.94372) (45.0,0.94051) (47.5,0.93731) (50.0,0.93412) (52.5,0.93095) (55.0,0.92778) (57.5,0.92463) (60.0,0.92148) (62.5,0.91835) (65.0,0.91522) (67.5,0.91211) (70.0,0.90901) (72.5,0.90591) (75.0,0.90283) (77.5,0.89976) (80.0,0.89670) (82.5,0.89365) (85.0,0.89061) (87.5,0.88758) (90.0,0.88456)};
\addplot[red!70!black,thick,dashed] coordinates {(0.0,1.00000) (2.5,0.99175) (5.0,0.98356) (7.5,0.97544) (10.0,0.96739) (12.5,0.95941) (15.0,0.95149) (17.5,0.94364) (20.0,0.93585) (22.5,0.92812) (25.0,0.92046) (27.5,0.91287) (30.0,0.90533) (32.5,0.89786) (35.0,0.89045) (37.5,0.88310) (40.0,0.87581) (42.5,0.86858) (45.0,0.86142) (47.5,0.85431) (50.0,0.84726) (52.5,0.84026) (55.0,0.83333) (57.5,0.82645) (60.0,0.81963) (62.5,0.81286) (65.0,0.80615) (67.5,0.79950) (70.0,0.79290) (72.5,0.78636) (75.0,0.77987) (77.5,0.77343) (80.0,0.76705) (82.5,0.76072) (85.0,0.75444) (87.5,0.74821) (90.0,0.74204)};
\end{groupplot}
\end{tikzpicture}
\caption{CSL interval reliability $R(t)$ of the but-for and fallback architectures, computed by uniformization on the series-system CTMC with absorbing down states.}
\label{fig:rel}
\end{figure}

\subsection{RQ3: Loss distributions and loss-minimising responses}
\Cref{tab:strategies} compares strategies, and \cref{fig:cdf} shows their loss distributions. Degrade-and-wait ($S_0$) is worst in every case, so mitigation pays. The loss-minimising response depends on the platform. For Bank~A, substitution ($S_3$) completes within the restraint horizon with probability $\rBankSthreePTd\%$, which cuts $\E\tau$ from $\rBankSzeroTau$ to $\rBankSthreeTau$ years and expected loss from $\rBankSzeroMean$ to $\rBankSthreeMean$. For Operator~T, substitution under fail-open (\textsc{continue}) is best ($\rTelSthreecMean$). For Carrier~P, the crashed design-around ($S_2$, $\rAirStwoMean$) beats substitution ($\rAirSthreeMean$), because PSS integration is slow and expensive. CCFH policy matters even at fixed strategy. \textsc{terminate} converts leakage into denial-driven churn and goodwill credits, and its upper tail is heavier: $P_{95}$ is $\rTelSthreetPnf$ against $\rTelSthreecPnf$ under $S_3$.

\begin{table}[t]
\centering\footnotesize
\caption{Strategy comparison (USD m). $\tilde T_d$: median remediation time (yr); $\Prob(T_d<T_R)$ in \%; $\E\tau$: expected degraded period (yr). $^\star$: loss-minimising strategy.}\label{tab:strategies}
\setlength{\tabcolsep}{4pt}
\begin{tabular}{@{}lrrrrrrrr@{}}
\toprule
Strategy & $\tilde T_d$ & $\Prob(T_d<T_R)$ & $\E\tau$ & $\E L$ & $P_{50}$ & $P_{95}$ & $\CVaR_{0.95}$ & Plug-in\\
\midrule
\multicolumn{9}{l}{\textbf{Bank A (core banking)}}\\
$S_0$ degrade-and-wait & $\infty$ & 0.0 & 2.17 & 603.9 & 520.7 & 1,276.6 & 1,650.3 & 433.7\\
$S_1$ design-around (nominal) & 2.77 & 21.6 & 1.99 & 578.6 & 514.5 & 1,155.2 & 1,462.5 & 447.5\\
$S_2$ design-around (crashed) & 2.28 & 36.2 & 1.87 & 552.9 & 494.9 & 1,079.0 & 1,361.4 & 454.6\\
$S_3$ substitution$^\star$ & 1.59 & 69.8 & 1.50 & 463.8 & 416.2 & 870.2 & 1,104.7 & 359.1\\
\midrule
\multicolumn{9}{l}{\textbf{Operator T (mobile network operator)}}\\
$S_0$ degrade-and-wait / \textsc{continue} & $\infty$ & 0.0 & 2.16 & 623.5 & 574.3 & 1,132.8 & 1,351.3 & 527.7\\
$S_0$ degrade-and-wait / \textsc{terminate} & $\infty$ & 0.0 & 2.16 & 624.4 & 570.6 & 1,161.4 & 1,402.3 & 503.6\\
$S_1$ design-around (nominal) / \textsc{continue} & 2.00 & 47.9 & 1.75 & 528.5 & 509.8 & 838.2 & 958.9 & 504.7\\
$S_1$ design-around (nominal) / \textsc{terminate} & 2.00 & 47.9 & 1.75 & 530.0 & 504.8 & 871.3 & 1,015.9 & 482.2\\
$S_2$ design-around (crashed) / \textsc{continue} & 1.66 & 64.3 & 1.56 & 480.6 & 462.8 & 746.3 & 852.8 & 427.1\\
$S_2$ design-around (crashed) / \textsc{terminate} & 1.66 & 64.3 & 1.56 & 482.3 & 459.1 & 777.2 & 905.7 & 408.4\\
$S_3$ substitution / \textsc{continue}$^\star$ & 1.42 & 77.7 & 1.38 & 446.8 & 432.0 & 669.4 & 755.1 & 392.9\\
$S_3$ substitution / \textsc{terminate} & 1.42 & 77.7 & 1.38 & 448.6 & 428.8 & 699.9 & 806.8 & 376.8\\
\midrule
\multicolumn{9}{l}{\textbf{Carrier P (airline PSS)}}\\
$S_0$ degrade-and-wait & $\infty$ & 0.0 & 2.17 & 385.0 & 353.7 & 708.3 & 847.6 & 302.8\\
$S_1$ design-around (nominal) & 2.60 & 26.0 & 1.96 & 363.8 & 350.5 & 599.6 & 688.7 & 312.9\\
$S_2$ design-around (crashed)$^\star$ & 2.15 & 41.3 & 1.82 & 345.3 & 334.1 & 550.3 & 632.5 & 317.9\\
$S_3$ substitution & 2.16 & 42.5 & 1.81 & 363.5 & 354.3 & 562.6 & 636.2 & 340.7\\
\bottomrule
\end{tabular}
\end{table}

\begin{figure}[t]
\centering
\begin{tikzpicture}
\begin{groupplot}[group style={group size=3 by 1,horizontal sep=1.1cm},width=0.36\textwidth,height=4.6cm,xlabel={Loss (USD m)},ymin=0,ymax=1,xmin=0,legend style={font=\tiny,at={(0.97,0.03)},anchor=south east}]
\nextgroupplot[title={Bank A},ylabel={$F_L$}]
\addplot[gray,thick] coordinates {(119.74,0.0050) (152.83,0.0150) (173.29,0.0250) (188.68,0.0350) (201.41,0.0450) (212.96,0.0550) (223.63,0.0650) (233.35,0.0750) (242.25,0.0850) (250.79,0.0950) (258.65,0.1050) (266.55,0.1150) (274.10,0.1250) (281.76,0.1350) (289.17,0.1450) (296.11,0.1550) (302.90,0.1650) (309.53,0.1750) (316.40,0.1850) (323.15,0.1950) (329.64,0.2050) (336.10,0.2150) (342.57,0.2250) (348.91,0.2350) (355.32,0.2450) (361.71,0.2550) (367.92,0.2650) (374.20,0.2750) (380.55,0.2850) (386.74,0.2950) (392.86,0.3050) (399.43,0.3150) (405.88,0.3250) (411.92,0.3350) (418.00,0.3450) (424.40,0.3550) (430.80,0.3650) (437.05,0.3750) (443.33,0.3850) (449.75,0.3950) (456.12,0.4050) (462.61,0.4150) (469.29,0.4250) (476.10,0.4350) (482.87,0.4450) (489.37,0.4550) (496.55,0.4650) (503.34,0.4750) (510.32,0.4850) (517.24,0.4950) (524.29,0.5050) (531.55,0.5150) (539.09,0.5250) (546.59,0.5350) (554.34,0.5450) (562.58,0.5550) (570.22,0.5650) (578.25,0.5750) (586.24,0.5850) (594.43,0.5950) (602.93,0.6050) (611.87,0.6150) (620.78,0.6250) (629.85,0.6350) (638.98,0.6450) (648.92,0.6550) (658.77,0.6650) (668.71,0.6750) (678.89,0.6850) (689.24,0.6950) (700.26,0.7050) (711.40,0.7150) (722.77,0.7250) (734.95,0.7350) (747.43,0.7450) (760.30,0.7550) (773.36,0.7650) (786.97,0.7750) (801.45,0.7850) (816.65,0.7950) (832.37,0.8050) (849.08,0.8150) (867.44,0.8250) (886.94,0.8350) (907.31,0.8450) (928.68,0.8550) (951.83,0.8650) (977.26,0.8750) (1003.26,0.8850) (1032.50,0.8950) (1064.52,0.9050) (1100.43,0.9150) (1142.26,0.9250) (1188.37,0.9350) (1244.30,0.9450) (1311.35,0.9550) (1398.89,0.9650) (1520.14,0.9750) (1699.92,0.9850) (2122.57,0.9950)};
\addplot[orange!85!black,thick] coordinates {(126.83,0.0050) (161.19,0.0150) (182.53,0.0250) (198.63,0.0350) (211.55,0.0450) (222.99,0.0550) (233.91,0.0650) (243.78,0.0750) (252.80,0.0850) (261.39,0.0950) (269.33,0.1050) (277.07,0.1150) (284.45,0.1250) (291.87,0.1350) (299.20,0.1450) (306.15,0.1550) (312.61,0.1650) (319.24,0.1750) (325.61,0.1850) (331.94,0.1950) (338.41,0.2050) (344.55,0.2150) (350.68,0.2250) (356.64,0.2350) (362.56,0.2450) (368.53,0.2550) (374.49,0.2650) (380.28,0.2750) (386.34,0.2850) (392.37,0.2950) (398.27,0.3050) (404.02,0.3150) (409.72,0.3250) (415.53,0.3350) (421.25,0.3450) (426.99,0.3550) (433.12,0.3650) (438.83,0.3750) (444.55,0.3850) (450.42,0.3950) (456.18,0.4050) (462.30,0.4150) (468.20,0.4250) (473.99,0.4350) (480.14,0.4450) (486.16,0.4550) (492.16,0.4650) (498.43,0.4750) (504.84,0.4850) (511.33,0.4950) (517.69,0.5050) (524.48,0.5150) (530.90,0.5250) (537.52,0.5350) (544.39,0.5450) (551.33,0.5550) (558.19,0.5650) (565.09,0.5750) (572.15,0.5850) (579.40,0.5950) (587.05,0.6050) (594.50,0.6150) (602.24,0.6250) (609.87,0.6350) (617.76,0.6450) (626.54,0.6550) (634.75,0.6650) (643.47,0.6750) (651.76,0.6850) (660.84,0.6950) (670.00,0.7050) (679.89,0.7150) (689.50,0.7250) (699.27,0.7350) (710.13,0.7450) (721.42,0.7550) (733.06,0.7650) (744.93,0.7750) (757.21,0.7850) (770.29,0.7950) (783.82,0.8050) (797.80,0.8150) (812.42,0.8250) (827.97,0.8350) (845.22,0.8450) (863.59,0.8550) (881.96,0.8650) (902.41,0.8750) (924.68,0.8850) (949.01,0.8950) (976.59,0.9050) (1006.48,0.9150) (1041.70,0.9250) (1080.97,0.9350) (1127.38,0.9450) (1184.41,0.9550) (1254.68,0.9650) (1348.96,0.9750) (1498.22,0.9850) (1855.03,0.9950)};
\addplot[blue!70!black,thick] coordinates {(130.67,0.0050) (165.20,0.0150) (186.72,0.0250) (202.58,0.0350) (215.23,0.0450) (226.55,0.0550) (236.91,0.0650) (246.37,0.0750) (255.18,0.0850) (262.97,0.0950) (270.34,0.1050) (277.90,0.1150) (285.11,0.1250) (292.04,0.1350) (298.62,0.1450) (304.90,0.1550) (311.13,0.1650) (317.30,0.1750) (323.26,0.1850) (328.99,0.1950) (334.70,0.2050) (340.43,0.2150) (346.01,0.2250) (351.62,0.2350) (357.26,0.2450) (362.66,0.2550) (367.82,0.2650) (373.35,0.2750) (378.65,0.2850) (384.05,0.2950) (389.45,0.3050) (394.50,0.3150) (399.85,0.3250) (405.07,0.3350) (410.10,0.3450) (415.16,0.3550) (420.64,0.3650) (425.98,0.3750) (431.34,0.3850) (436.68,0.3950) (442.06,0.4050) (447.61,0.4150) (452.88,0.4250) (458.42,0.4350) (463.81,0.4450) (469.17,0.4550) (474.88,0.4650) (480.53,0.4750) (486.33,0.4850) (492.10,0.4950) (497.62,0.5050) (503.57,0.5150) (509.39,0.5250) (515.24,0.5350) (521.37,0.5450) (527.61,0.5550) (533.78,0.5650) (540.09,0.5750) (546.66,0.5850) (553.31,0.5950) (559.65,0.6050) (566.25,0.6150) (573.06,0.6250) (580.26,0.6350) (587.97,0.6450) (594.95,0.6550) (602.89,0.6650) (610.79,0.6750) (618.75,0.6850) (626.82,0.6950) (635.16,0.7050) (644.17,0.7150) (653.29,0.7250) (662.63,0.7350) (672.68,0.7450) (682.52,0.7550) (693.07,0.7650) (704.13,0.7750) (715.17,0.7850) (727.37,0.7950) (739.61,0.8050) (752.57,0.8150) (766.28,0.8250) (781.36,0.8350) (796.52,0.8450) (812.77,0.8550) (830.55,0.8650) (849.52,0.8750) (870.53,0.8850) (892.79,0.8950) (918.19,0.9050) (945.16,0.9150) (976.81,0.9250) (1011.27,0.9350) (1054.76,0.9450) (1106.43,0.9550) (1168.75,0.9650) (1256.53,0.9750) (1395.30,0.9850) (1725.28,0.9950)};
\addplot[red!75!black,very thick] coordinates {(140.05,0.0050) (170.51,0.0150) (187.31,0.0250) (198.76,0.0350) (208.53,0.0450) (216.60,0.0550) (223.89,0.0650) (231.07,0.0750) (237.47,0.0850) (243.50,0.0950) (249.05,0.1050) (254.58,0.1150) (259.75,0.1250) (264.82,0.1350) (269.81,0.1450) (274.57,0.1550) (279.08,0.1650) (283.53,0.1750) (287.93,0.1850) (292.15,0.1950) (296.34,0.2050) (300.60,0.2150) (304.77,0.2250) (308.83,0.2350) (312.91,0.2450) (316.99,0.2550) (320.94,0.2650) (324.83,0.2750) (328.73,0.2850) (332.64,0.2950) (336.63,0.3050) (340.60,0.3150) (344.29,0.3250) (348.17,0.3350) (352.10,0.3450) (356.08,0.3550) (360.02,0.3650) (364.08,0.3750) (368.13,0.3850) (372.25,0.3950) (376.30,0.4050) (380.28,0.4150) (384.47,0.4250) (388.63,0.4350) (392.75,0.4450) (396.92,0.4550) (401.21,0.4650) (405.42,0.4750) (409.68,0.4850) (414.01,0.4950) (418.41,0.5050) (422.90,0.5150) (427.40,0.5250) (431.83,0.5350) (436.35,0.5450) (441.19,0.5550) (445.87,0.5650) (450.62,0.5750) (455.66,0.5850) (460.68,0.5950) (465.82,0.6050) (471.06,0.6150) (476.45,0.6250) (482.15,0.6350) (487.92,0.6450) (493.56,0.6550) (499.55,0.6650) (505.61,0.6750) (512.10,0.6850) (518.46,0.6950) (524.93,0.7050) (531.69,0.7150) (538.72,0.7250) (545.87,0.7350) (553.44,0.7450) (561.20,0.7550) (569.24,0.7650) (577.58,0.7750) (586.43,0.7850) (595.58,0.7950) (605.11,0.8050) (615.61,0.8150) (626.12,0.8250) (637.11,0.8350) (649.10,0.8450) (661.80,0.8550) (676.09,0.8650) (690.12,0.8750) (706.65,0.8850) (724.49,0.8950) (744.38,0.9050) (766.24,0.9150) (791.22,0.9250) (818.43,0.9350) (851.21,0.9450) (891.27,0.9550) (943.99,0.9650) (1015.18,0.9750) (1129.47,0.9850) (1394.05,0.9950)};
\legend{$S_0$,$S_1$,$S_2$,$S_3$}
\nextgroupplot[title={Operator T}]
\addplot[gray,thick] coordinates {(184.01,0.0050) (221.39,0.0150) (242.72,0.0250) (259.41,0.0350) (273.44,0.0450) (285.50,0.0550) (296.41,0.0650) (306.44,0.0750) (315.99,0.0850) (324.27,0.0950) (332.60,0.1050) (340.46,0.1150) (348.12,0.1250) (355.64,0.1350) (362.79,0.1450) (369.72,0.1550) (376.29,0.1650) (383.15,0.1750) (389.62,0.1850) (395.67,0.1950) (401.86,0.2050) (407.89,0.2150) (413.82,0.2250) (419.73,0.2350) (425.76,0.2450) (431.92,0.2550) (437.77,0.2650) (443.63,0.2750) (449.51,0.2850) (455.11,0.2950) (460.88,0.3050) (466.52,0.3150) (472.14,0.3250) (477.90,0.3350) (483.75,0.3450) (489.29,0.3550) (494.98,0.3650) (500.81,0.3750) (506.61,0.3850) (512.50,0.3950) (518.46,0.4050) (524.34,0.4150) (529.99,0.4250) (535.70,0.4350) (541.60,0.4450) (547.44,0.4550) (553.39,0.4650) (559.30,0.4750) (565.31,0.4850) (571.28,0.4950) (577.62,0.5050) (583.83,0.5150) (590.04,0.5250) (596.16,0.5350) (602.70,0.5450) (609.27,0.5550) (615.93,0.5650) (622.82,0.5750) (629.65,0.5850) (636.74,0.5950) (643.87,0.6050) (650.66,0.6150) (657.90,0.6250) (665.11,0.6350) (672.57,0.6450) (680.17,0.6550) (687.93,0.6650) (695.81,0.6750) (704.41,0.6850) (713.23,0.6950) (722.20,0.7050) (730.95,0.7150) (740.21,0.7250) (749.81,0.7350) (759.25,0.7450) (769.01,0.7550) (779.26,0.7650) (789.89,0.7750) (801.30,0.7850) (812.60,0.7950) (824.77,0.8050) (837.37,0.8150) (850.72,0.8250) (864.52,0.8350) (878.93,0.8450) (894.55,0.8550) (911.16,0.8650) (929.29,0.8750) (947.44,0.8850) (967.86,0.8950) (990.29,0.9050) (1016.42,0.9150) (1043.60,0.9250) (1075.54,0.9350) (1111.70,0.9450) (1156.56,0.9550) (1213.50,0.9650) (1286.20,0.9750) (1399.54,0.9850) (1631.21,0.9950)};
\addplot[orange!85!black,thick] coordinates {(188.42,0.0050) (226.27,0.0150) (246.97,0.0250) (263.52,0.0350) (276.96,0.0450) (288.17,0.0550) (298.07,0.0650) (307.43,0.0750) (316.04,0.0850) (323.67,0.0950) (330.94,0.1050) (337.78,0.1150) (344.22,0.1250) (350.47,0.1350) (356.38,0.1450) (362.17,0.1550) (367.62,0.1650) (372.93,0.1750) (378.07,0.1850) (383.07,0.1950) (387.97,0.2050) (392.73,0.2150) (397.48,0.2250) (401.91,0.2350) (406.48,0.2450) (410.87,0.2550) (415.17,0.2650) (419.64,0.2750) (424.01,0.2850) (428.14,0.2950) (432.22,0.3050) (436.48,0.3150) (440.58,0.3250) (444.65,0.3350) (448.71,0.3450) (452.70,0.3550) (456.65,0.3650) (460.57,0.3750) (464.49,0.3850) (468.30,0.3950) (472.35,0.4050) (476.21,0.4150) (480.14,0.4250) (484.14,0.4350) (488.02,0.4450) (491.97,0.4550) (495.83,0.4650) (499.74,0.4750) (503.62,0.4850) (507.65,0.4950) (511.81,0.5050) (515.86,0.5150) (519.92,0.5250) (524.00,0.5350) (528.00,0.5450) (532.31,0.5550) (536.42,0.5650) (540.42,0.5750) (544.68,0.5850) (549.02,0.5950) (553.38,0.6050) (557.74,0.6150) (562.19,0.6250) (566.80,0.6350) (571.33,0.6450) (576.00,0.6550) (580.91,0.6650) (585.81,0.6750) (590.83,0.6850) (595.80,0.6950) (600.93,0.7050) (606.17,0.7150) (611.52,0.7250) (617.16,0.7350) (622.98,0.7450) (628.96,0.7550) (634.99,0.7650) (641.36,0.7750) (647.88,0.7850) (654.77,0.7950) (661.92,0.8050) (669.11,0.8150) (676.65,0.8250) (684.53,0.8350) (692.98,0.8450) (701.87,0.8550) (711.45,0.8650) (721.69,0.8750) (732.65,0.8850) (744.17,0.8950) (757.31,0.9050) (771.34,0.9150) (787.27,0.9250) (805.19,0.9350) (826.48,0.9450) (851.42,0.9550) (882.16,0.9650) (922.30,0.9750) (980.97,0.9850) (1113.07,0.9950)};
\addplot[blue!70!black,thick] coordinates {(189.35,0.0050) (224.88,0.0150) (244.37,0.0250) (258.81,0.0350) (270.17,0.0450) (279.87,0.0550) (288.55,0.0650) (296.04,0.0750) (302.90,0.0850) (309.50,0.0950) (315.55,0.1050) (321.24,0.1150) (326.60,0.1250) (331.65,0.1350) (336.45,0.1450) (341.00,0.1550) (345.56,0.1650) (350.14,0.1750) (354.27,0.1850) (358.36,0.1950) (362.59,0.2050) (366.44,0.2150) (370.03,0.2250) (373.73,0.2350) (377.49,0.2450) (381.14,0.2550) (384.67,0.2650) (388.14,0.2750) (391.66,0.2850) (395.12,0.2950) (398.60,0.3050) (401.96,0.3150) (405.31,0.3250) (408.58,0.3350) (411.89,0.3450) (415.12,0.3550) (418.39,0.3650) (421.71,0.3750) (424.99,0.3850) (428.23,0.3950) (431.58,0.4050) (434.76,0.4150) (438.01,0.4250) (441.14,0.4350) (444.41,0.4450) (447.80,0.4550) (451.15,0.4650) (454.42,0.4750) (457.70,0.4850) (461.05,0.4950) (464.58,0.5050) (467.95,0.5150) (471.35,0.5250) (474.80,0.5350) (478.27,0.5450) (481.71,0.5550) (485.41,0.5650) (489.00,0.5750) (492.64,0.5850) (496.29,0.5950) (499.97,0.6050) (503.75,0.6150) (507.44,0.6250) (511.22,0.6350) (514.93,0.6450) (518.85,0.6550) (523.05,0.6650) (527.23,0.6750) (531.50,0.6850) (535.82,0.6950) (540.21,0.7050) (544.74,0.7150) (549.52,0.7250) (554.28,0.7350) (559.28,0.7450) (564.34,0.7550) (569.64,0.7650) (575.05,0.7750) (580.61,0.7850) (586.41,0.7950) (592.44,0.8050) (598.69,0.8150) (605.12,0.8250) (611.86,0.8350) (619.18,0.8450) (627.05,0.8550) (635.37,0.8650) (644.40,0.8750) (654.01,0.8850) (664.51,0.8950) (675.61,0.9050) (688.48,0.9150) (702.22,0.9250) (717.59,0.9350) (735.55,0.9450) (758.11,0.9550) (785.09,0.9650) (819.97,0.9750) (875.08,0.9850) (990.82,0.9950)};
\addplot[red!75!black,very thick] coordinates {(202.44,0.0050) (233.31,0.0150) (249.08,0.0250) (260.60,0.0350) (269.79,0.0450) (277.69,0.0550) (284.72,0.0650) (291.11,0.0750) (297.12,0.0850) (302.37,0.0950) (307.42,0.1050) (312.08,0.1150) (316.40,0.1250) (320.65,0.1350) (324.62,0.1450) (328.52,0.1550) (332.35,0.1650) (336.25,0.1750) (339.81,0.1850) (343.26,0.1950) (346.53,0.2050) (349.74,0.2150) (353.00,0.2250) (356.15,0.2350) (359.26,0.2450) (362.39,0.2550) (365.43,0.2650) (368.45,0.2750) (371.53,0.2850) (374.43,0.2950) (377.40,0.3050) (380.30,0.3150) (383.09,0.3250) (385.88,0.3350) (388.67,0.3450) (391.49,0.3550) (394.35,0.3650) (397.20,0.3750) (399.96,0.3850) (402.68,0.3950) (405.39,0.4050) (408.21,0.4150) (411.05,0.4250) (413.77,0.4350) (416.53,0.4450) (419.30,0.4550) (422.01,0.4650) (424.83,0.4750) (427.63,0.4850) (430.49,0.4950) (433.40,0.5050) (436.28,0.5150) (439.16,0.5250) (442.23,0.5350) (445.21,0.5450) (448.16,0.5550) (451.13,0.5650) (454.24,0.5750) (457.47,0.5850) (460.48,0.5950) (463.60,0.6050) (466.71,0.6150) (469.87,0.6250) (473.27,0.6350) (476.56,0.6450) (480.04,0.6550) (483.50,0.6650) (487.16,0.6750) (490.73,0.6850) (494.30,0.6950) (498.17,0.7050) (501.95,0.7150) (505.88,0.7250) (509.96,0.7350) (514.16,0.7450) (518.37,0.7550) (522.86,0.7650) (527.67,0.7750) (532.38,0.7850) (537.43,0.7950) (542.49,0.8050) (547.90,0.8150) (553.39,0.8250) (559.38,0.8350) (565.52,0.8450) (571.95,0.8550) (578.66,0.8650) (586.25,0.8750) (594.21,0.8850) (602.52,0.8950) (611.76,0.9050) (621.96,0.9150) (633.11,0.9250) (646.09,0.9350) (660.90,0.9450) (678.57,0.9550) (700.28,0.9650) (729.09,0.9750) (771.77,0.9850) (864.41,0.9950)};
\addplot[red!75!black,thick,dashed] coordinates {(194.26,0.0050) (223.70,0.0150) (239.35,0.0250) (250.44,0.0350) (259.93,0.0450) (267.67,0.0550) (274.73,0.0650) (281.11,0.0750) (287.00,0.0850) (292.55,0.0950) (297.77,0.1050) (302.50,0.1150) (307.10,0.1250) (311.39,0.1350) (315.59,0.1450) (319.47,0.1550) (323.27,0.1650) (327.08,0.1750) (330.63,0.1850) (334.17,0.1950) (337.76,0.2050) (341.36,0.2150) (344.77,0.2250) (348.02,0.2350) (351.30,0.2450) (354.50,0.2550) (357.67,0.2650) (361.00,0.2750) (364.05,0.2850) (367.08,0.2950) (370.08,0.3050) (373.13,0.3150) (376.08,0.3250) (379.06,0.3350) (382.14,0.3450) (385.04,0.3550) (388.07,0.3650) (391.03,0.3750) (394.14,0.3850) (397.16,0.3950) (400.10,0.4050) (403.10,0.4150) (406.12,0.4250) (409.10,0.4350) (411.99,0.4450) (414.98,0.4550) (418.02,0.4650) (421.09,0.4750) (424.15,0.4850) (427.25,0.4950) (430.26,0.5050) (433.38,0.5150) (436.52,0.5250) (439.76,0.5350) (442.88,0.5450) (446.18,0.5550) (449.45,0.5650) (452.78,0.5750) (456.08,0.5850) (459.53,0.5950) (462.99,0.6050) (466.54,0.6150) (470.21,0.6250) (473.82,0.6350) (477.53,0.6450) (481.26,0.6550) (485.02,0.6650) (488.90,0.6750) (492.91,0.6850) (497.05,0.6950) (501.37,0.7050) (505.71,0.7150) (510.04,0.7250) (514.71,0.7350) (519.40,0.7450) (524.36,0.7550) (529.15,0.7650) (534.48,0.7750) (539.81,0.7850) (545.35,0.7950) (551.13,0.8050) (557.37,0.8150) (563.52,0.8250) (570.57,0.8350) (577.67,0.8450) (585.23,0.8550) (593.13,0.8650) (601.90,0.8750) (610.80,0.8850) (620.70,0.8950) (631.66,0.9050) (643.85,0.9150) (657.26,0.9250) (672.53,0.9350) (690.09,0.9450) (711.40,0.9550) (737.24,0.9650) (772.19,0.9750) (824.52,0.9850) (945.32,0.9950)};
\legend{$S_0$/C,$S_1$/C,$S_2$/C,$S_3$/C,$S_3$/T}
\nextgroupplot[title={Carrier P}]
\addplot[gray,thick] coordinates {(104.78,0.0050) (128.83,0.0150) (142.58,0.0250) (153.39,0.0350) (162.61,0.0450) (170.24,0.0550) (177.23,0.0650) (183.78,0.0750) (189.84,0.0850) (195.32,0.0950) (200.59,0.1050) (205.70,0.1150) (210.53,0.1250) (215.18,0.1350) (219.53,0.1450) (223.95,0.1550) (228.07,0.1650) (232.28,0.1750) (236.23,0.1850) (240.22,0.1950) (244.25,0.2050) (248.23,0.2150) (252.04,0.2250) (255.93,0.2350) (259.79,0.2450) (263.66,0.2550) (267.27,0.2650) (270.91,0.2750) (274.69,0.2850) (278.27,0.2950) (281.72,0.3050) (285.44,0.3150) (289.01,0.3250) (292.51,0.3350) (296.07,0.3450) (299.68,0.3550) (303.33,0.3650) (306.88,0.3750) (310.50,0.3850) (314.07,0.3950) (317.66,0.4050) (321.41,0.4150) (325.02,0.4250) (328.71,0.4350) (332.40,0.4450) (336.14,0.4550) (339.96,0.4650) (343.94,0.4750) (347.81,0.4850) (351.69,0.4950) (355.69,0.5050) (359.58,0.5150) (363.49,0.5250) (367.54,0.5350) (371.62,0.5450) (375.88,0.5550) (380.14,0.5650) (384.63,0.5750) (389.06,0.5850) (393.63,0.5950) (398.02,0.6050) (402.62,0.6150) (407.18,0.6250) (411.96,0.6350) (416.63,0.6450) (421.59,0.6550) (426.39,0.6650) (431.72,0.6750) (437.24,0.6850) (442.50,0.6950) (448.08,0.7050) (453.69,0.7150) (459.21,0.7250) (465.34,0.7350) (471.44,0.7450) (477.90,0.7550) (484.31,0.7650) (491.01,0.7750) (497.95,0.7850) (505.55,0.7950) (512.94,0.8050) (520.73,0.8150) (529.24,0.8250) (537.80,0.8350) (547.15,0.8450) (556.80,0.8550) (567.63,0.8650) (578.85,0.8750) (590.68,0.8850) (603.80,0.8950) (617.66,0.9050) (634.16,0.9150) (652.32,0.9250) (671.53,0.9350) (695.08,0.9450) (722.87,0.9550) (757.92,0.9650) (805.10,0.9750) (875.59,0.9850) (1025.92,0.9950)};
\addplot[orange!85!black,thick] coordinates {(109.65,0.0050) (133.70,0.0150) (148.38,0.0250) (159.63,0.0350) (169.15,0.0450) (177.28,0.0550) (184.35,0.0650) (190.63,0.0750) (196.64,0.0850) (202.49,0.0950) (207.91,0.1050) (213.10,0.1150) (217.93,0.1250) (222.37,0.1350) (226.75,0.1450) (231.03,0.1550) (235.24,0.1650) (239.50,0.1750) (243.42,0.1850) (247.50,0.1950) (251.33,0.2050) (255.21,0.2150) (259.09,0.2250) (262.76,0.2350) (266.40,0.2450) (270.04,0.2550) (273.54,0.2650) (277.07,0.2750) (280.63,0.2850) (283.98,0.2950) (287.45,0.3050) (290.70,0.3150) (294.17,0.3250) (297.46,0.3350) (300.73,0.3450) (303.97,0.3550) (307.27,0.3650) (310.33,0.3750) (313.59,0.3850) (316.74,0.3950) (319.96,0.4050) (323.20,0.4150) (326.35,0.4250) (329.62,0.4350) (332.74,0.4450) (335.99,0.4550) (339.21,0.4650) (342.30,0.4750) (345.57,0.4850) (348.86,0.4950) (352.15,0.5050) (355.31,0.5150) (358.52,0.5250) (361.79,0.5350) (365.19,0.5450) (368.50,0.5550) (371.89,0.5650) (375.19,0.5750) (378.55,0.5850) (382.02,0.5950) (385.51,0.6050) (389.17,0.6150) (392.80,0.6250) (396.54,0.6350) (400.10,0.6450) (403.68,0.6550) (407.47,0.6650) (411.26,0.6750) (415.14,0.6850) (419.30,0.6950) (423.42,0.7050) (427.72,0.7150) (432.06,0.7250) (436.42,0.7350) (440.79,0.7450) (445.28,0.7550) (449.92,0.7650) (454.80,0.7750) (459.55,0.7850) (464.64,0.7950) (470.04,0.8050) (475.57,0.8150) (481.33,0.8250) (487.37,0.8350) (493.87,0.8450) (500.70,0.8550) (507.55,0.8650) (515.15,0.8750) (523.18,0.8850) (532.14,0.8950) (541.55,0.9050) (551.84,0.9150) (563.49,0.9250) (576.20,0.9350) (591.05,0.9450) (608.26,0.9550) (630.33,0.9650) (659.63,0.9750) (704.44,0.9850) (798.37,0.9950)};
\addplot[blue!70!black,very thick] coordinates {(112.35,0.0050) (136.44,0.0150) (151.03,0.0250) (162.24,0.0350) (171.76,0.0450) (179.82,0.0550) (186.68,0.0650) (193.14,0.0750) (199.16,0.0850) (204.74,0.0950) (209.91,0.1050) (214.75,0.1150) (219.25,0.1250) (223.63,0.1350) (227.56,0.1450) (231.52,0.1550) (235.44,0.1650) (239.31,0.1750) (242.99,0.1850) (246.48,0.1950) (249.93,0.2050) (253.38,0.2150) (256.68,0.2250) (259.90,0.2350) (263.21,0.2450) (266.29,0.2550) (269.42,0.2650) (272.41,0.2750) (275.25,0.2850) (278.11,0.2950) (281.06,0.3050) (283.94,0.3150) (286.63,0.3250) (289.43,0.3350) (292.26,0.3450) (295.01,0.3550) (297.73,0.3650) (300.54,0.3750) (303.15,0.3850) (305.84,0.3950) (308.44,0.4050) (311.12,0.4150) (313.83,0.4250) (316.55,0.4350) (319.28,0.4450) (322.07,0.4550) (324.72,0.4650) (327.43,0.4750) (330.04,0.4850) (332.80,0.4950) (335.55,0.5050) (338.25,0.5150) (340.99,0.5250) (343.68,0.5350) (346.47,0.5450) (349.29,0.5550) (352.07,0.5650) (354.85,0.5750) (357.78,0.5850) (360.75,0.5950) (363.66,0.6050) (366.71,0.6150) (369.71,0.6250) (372.77,0.6350) (375.89,0.6450) (378.99,0.6550) (382.23,0.6650) (385.37,0.6750) (388.84,0.6850) (392.19,0.6950) (395.61,0.7050) (399.07,0.7150) (402.61,0.7250) (406.36,0.7350) (410.18,0.7450) (414.25,0.7550) (418.41,0.7650) (422.61,0.7750) (426.89,0.7850) (431.41,0.7950) (436.03,0.8050) (440.89,0.8150) (445.81,0.8250) (451.21,0.8350) (456.92,0.8450) (462.88,0.8550) (469.05,0.8650) (475.60,0.8750) (482.79,0.8850) (490.36,0.8950) (498.95,0.9050) (507.94,0.9150) (518.12,0.9250) (529.89,0.9350) (542.84,0.9450) (559.01,0.9550) (578.22,0.9650) (606.27,0.9750) (646.51,0.9850) (732.96,0.9950)};
\addplot[red!75!black,thick] coordinates {(129.09,0.0050) (153.48,0.0150) (168.24,0.0250) (179.89,0.0350) (189.55,0.0450) (197.77,0.0550) (205.14,0.0650) (211.66,0.0750) (217.52,0.0850) (223.13,0.0950) (228.35,0.1050) (233.21,0.1150) (237.76,0.1250) (242.24,0.1350) (246.52,0.1450) (250.58,0.1550) (254.61,0.1650) (258.48,0.1750) (262.13,0.1850) (265.79,0.1950) (269.29,0.2050) (272.82,0.2150) (276.20,0.2250) (279.50,0.2350) (282.70,0.2450) (285.95,0.2550) (289.05,0.2650) (292.12,0.2750) (295.01,0.2850) (298.07,0.2950) (300.95,0.3050) (303.81,0.3150) (306.60,0.3250) (309.39,0.3350) (312.19,0.3450) (314.89,0.3550) (317.67,0.3650) (320.46,0.3750) (323.24,0.3850) (325.96,0.3950) (328.69,0.4050) (331.33,0.4150) (334.00,0.4250) (336.61,0.4350) (339.32,0.4450) (341.91,0.4550) (344.66,0.4650) (347.35,0.4750) (350.16,0.4850) (352.86,0.4950) (355.59,0.5050) (358.27,0.5150) (361.03,0.5250) (363.78,0.5350) (366.58,0.5450) (369.36,0.5550) (372.24,0.5650) (375.11,0.5750) (377.90,0.5850) (380.75,0.5950) (383.75,0.6050) (386.77,0.6150) (389.80,0.6250) (392.84,0.6350) (395.94,0.6450) (399.08,0.6550) (402.31,0.6650) (405.55,0.6750) (408.82,0.6850) (411.96,0.6950) (415.33,0.7050) (418.85,0.7150) (422.45,0.7250) (426.01,0.7350) (429.80,0.7450) (433.59,0.7550) (437.60,0.7650) (441.69,0.7750) (445.73,0.7850) (450.02,0.7950) (454.67,0.8050) (459.12,0.8150) (463.92,0.8250) (468.97,0.8350) (474.27,0.8450) (479.96,0.8550) (485.79,0.8650) (492.21,0.8750) (499.00,0.8850) (506.52,0.8950) (514.23,0.9050) (522.88,0.9150) (532.66,0.9250) (543.48,0.9350) (555.53,0.9450) (570.49,0.9550) (588.10,0.9650) (611.69,0.9750) (646.67,0.9850) (727.22,0.9950)};
\legend{$S_0$,$S_1$,$S_2$,$S_3$}
\end{groupplot}
\end{tikzpicture}
\caption{Loss distribution functions by strategy (C/T: CCFH \textsc{continue}/\textsc{terminate}). Common random numbers make the curves directly comparable.}
\label{fig:cdf}
\end{figure}
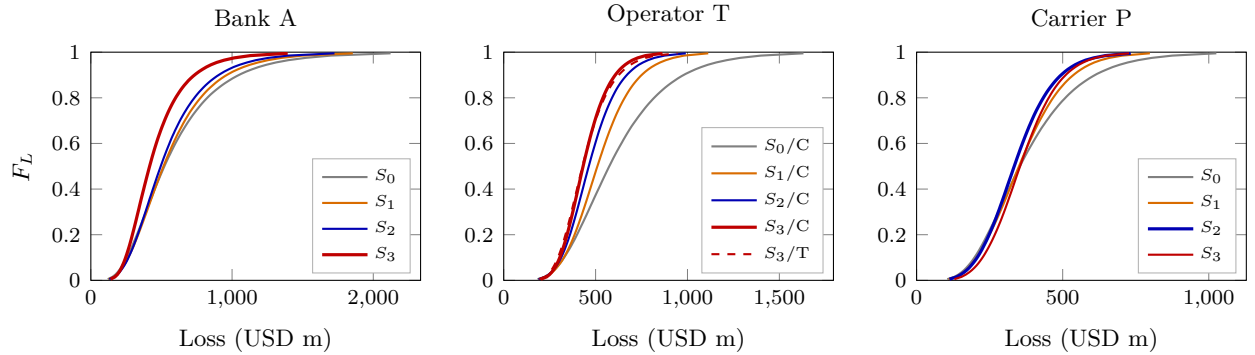

\Cref{tab:summary} and \cref{fig:hist} describe the mitigated loss $L(s^\star)$. All three distributions are right-skewed. Bank~A's is extreme (skewness $\rBankSkew$), because the customer-response channels are convex in outage hours and in the hazard parameters. Its $P_{99}$ is $\rBankPnn$, compared with a mean of $\rBankMean$. Operator~T's and Carrier~P's are moderately skewed ($\rTelSkew$, $\rAirSkew$). The mean exceeds the median in every case. By \cref{thm:bond}(c), this already implies that posting $\E L$ is not optimal whenever $\alpha^\star>F_L(\E L)$. Here $F_L(\E L)=\rBankFatEL$, $\rTelFatEL$ and $\rAirFatEL$.

\begin{table}[t]
\centering\footnotesize
\caption{Mitigated loss $L(s^\star)$ (USD m, $n=\rN$). Quantile CIs are distribution-free (\cref{prop:qci}).}\label{tab:summary}
\begin{tabular}{@{}lrrr@{}}
\toprule
Statistic & Bank A ($S_3$) & Operator T ($S_3$/\textsc{continue}) & Carrier P ($S_2$)\\
\midrule
$\mathbb{E}[L]$ & 463.8 & 446.8 & 345.3\\
95\% CI of mean & [462.8, 464.8] & [446.3, 447.4] & [344.8, 345.8]\\
Std.\ deviation & 228.5 & 122.9 & 117.0\\
Skewness & 5.17 & 0.87 & 1.05\\
$P_{5}$ & 212.8 & 274.0 & 176.1\\
$P_{50}$ & 416.2 & 432.0 & 334.1\\
$P_{90}$ & 734.2 & 607.1 & 494.7\\
$P_{95}$ [95\% CI] & 870.2 [866.1, 874.0] & 669.4 [667.7, 671.3] & 550.3 [548.9, 551.8]\\
$P_{99}$ [95\% CI] & 1,226.0 [1,215.7, 1,236.7] & 805.6 [801.6, 809.8] & 679.3 [675.9, 683.0]\\
$\mathrm{CVaR}_{0.95}$ & 1,104.7 & 755.1 & 632.5\\
$\mathrm{CVaR}_{0.99}$ & 1,536.6 & 892.9 & 769.8\\
Plug-in estimate & 359.1 & 392.9 & 317.9\\
Percentile rank of plug-in & 36.3\% & 36.0\% & 44.0\%\\
\bottomrule
\end{tabular}
\end{table}

\begin{figure}[t]
\centering
\begin{tikzpicture}
\begin{groupplot}[group style={group size=3 by 1,horizontal sep=1.25cm},width=0.36\textwidth,height=4.6cm,xlabel={Loss (USD m)},ymin=0,scaled y ticks=false,yticklabel style={/pgf/number format/fixed,/pgf/number format/precision=4},xmin=0]
\nextgroupplot[title={Bank A ($S_3$)},ylabel={Density (1/USD m)},xmax=\rBankHistXmax,ymax=\rBankHistMax]
\addplot[const plot mark mid,fill=blue!18,draw=blue!60!black] coordinates {(11.62,0.000000) (34.85,0.000000) (58.09,0.000002) (81.32,0.000011) (104.55,0.000054) (127.79,0.000144) (151.02,0.000282) (174.26,0.000544) (197.49,0.000940) (220.73,0.001350) (243.96,0.001726) (267.19,0.002054) (290.43,0.002334) (313.66,0.002506) (336.90,0.002581) (360.13,0.002503) (383.36,0.002445) (406.60,0.002349) (429.83,0.002208) (453.07,0.002044) (476.30,0.001823) (499.54,0.001664) (522.77,0.001536) (546.00,0.001363) (569.24,0.001227) (592.47,0.001084) (615.71,0.000961) (638.94,0.000865) (662.18,0.000749) (685.41,0.000690) (708.64,0.000584) (731.88,0.000505) (755.11,0.000463) (778.35,0.000401) (801.58,0.000369) (824.82,0.000334) (848.05,0.000283) (871.28,0.000248) (894.52,0.000221) (917.75,0.000192) (940.99,0.000164) (964.22,0.000161) (987.46,0.000131) (1010.69,0.000114) (1033.92,0.000104) (1057.16,0.000093) (1080.39,0.000087) (1103.63,0.000075) (1126.86,0.000065) (1150.09,0.000062) (1173.33,0.000051) (1196.56,0.000045) (1219.80,0.000040) (1243.03,0.000040) (1266.27,0.000031) (1289.50,0.000034) (1312.73,0.000030) (1335.97,0.000027) (1359.20,0.000024) (1382.44,0.000022)} \closedcycle;
\draw[black,thick] (axis cs:\rBankMeanNum,0) -- (axis cs:\rBankMeanNum,\rBankHistMax);
\draw[green!45!black,thick,dashed] (axis cs:\rBankPlugNum,0) -- (axis cs:\rBankPlugNum,\rBankHistMax);
\draw[red!70!black,thick] (axis cs:\rBankBondNum,0) -- (axis cs:\rBankBondNum,\rBankHistMax);
\draw[red!70!black,dotted,thick] (axis cs:\rBankPnnNum,0) -- (axis cs:\rBankPnnNum,\rBankHistMax);
\nextgroupplot[title={Operator T ($S_3$/C)},xmax=\rTelHistXmax,ymax=\rTelHistMax]
\addplot[const plot mark mid,fill=blue!18,draw=blue!60!black] coordinates {(7.20,0.000000) (21.61,0.000000) (36.02,0.000000) (50.42,0.000000) (64.83,0.000000) (79.24,0.000000) (93.64,0.000000) (108.05,0.000002) (122.46,0.000004) (136.87,0.000012) (151.27,0.000031) (165.68,0.000051) (180.09,0.000091) (194.49,0.000149) (208.90,0.000233) (223.31,0.000376) (237.71,0.000572) (252.12,0.000845) (266.53,0.001098) (280.93,0.001436) (295.34,0.001742) (309.75,0.002161) (324.16,0.002535) (338.56,0.002793) (352.97,0.003164) (367.38,0.003322) (381.78,0.003523) (396.19,0.003599) (410.60,0.003634) (425.00,0.003546) (439.41,0.003435) (453.82,0.003281) (468.22,0.003124) (482.63,0.002844) (497.04,0.002687) (511.44,0.002419) (525.85,0.002144) (540.26,0.001947) (554.67,0.001737) (569.07,0.001562) (583.48,0.001323) (597.89,0.001215) (612.29,0.001033) (626.70,0.000912) (641.11,0.000764) (655.51,0.000658) (669.92,0.000564) (684.33,0.000488) (698.73,0.000422) (713.14,0.000358) (727.55,0.000289) (741.96,0.000265) (756.36,0.000210) (770.77,0.000192) (785.18,0.000158) (799.58,0.000125) (813.99,0.000102) (828.40,0.000094) (842.80,0.000077) (857.21,0.000066)} \closedcycle;
\draw[black,thick] (axis cs:\rTelMeanNum,0) -- (axis cs:\rTelMeanNum,\rTelHistMax);
\draw[green!45!black,thick,dashed] (axis cs:\rTelPlugNum,0) -- (axis cs:\rTelPlugNum,\rTelHistMax);
\draw[red!70!black,thick] (axis cs:\rTelBondNum,0) -- (axis cs:\rTelBondNum,\rTelHistMax);
\draw[red!70!black,dotted,thick] (axis cs:\rTelPnnNum,0) -- (axis cs:\rTelPnnNum,\rTelHistMax);
\nextgroupplot[title={Carrier P ($S_2$)},xmax=\rAirHistXmax,ymax=\rAirHistMax]
\addplot[const plot mark mid,fill=blue!18,draw=blue!60!black] coordinates {(6.11,0.000000) (18.32,0.000000) (30.54,0.000000) (42.76,0.000001) (54.97,0.000005) (67.19,0.000020) (79.40,0.000050) (91.62,0.000103) (103.84,0.000190) (116.05,0.000292) (128.27,0.000478) (140.48,0.000640) (152.70,0.000831) (164.92,0.001026) (177.13,0.001289) (189.35,0.001549) (201.57,0.001798) (213.78,0.002121) (226.00,0.002473) (238.21,0.002662) (250.43,0.002944) (262.65,0.003152) (274.86,0.003443) (287.08,0.003578) (299.29,0.003710) (311.51,0.003770) (323.73,0.003717) (335.94,0.003674) (348.16,0.003617) (360.37,0.003398) (372.59,0.003278) (384.81,0.003056) (397.02,0.002900) (409.24,0.002560) (421.45,0.002385) (433.67,0.002170) (445.89,0.001935) (458.10,0.001705) (470.32,0.001563) (482.53,0.001370) (494.75,0.001164) (506.97,0.001088) (519.18,0.000898) (531.40,0.000808) (543.62,0.000713) (555.83,0.000576) (568.05,0.000529) (580.26,0.000429) (592.48,0.000362) (604.70,0.000309) (616.91,0.000274) (629.13,0.000248) (641.34,0.000200) (653.56,0.000177) (665.78,0.000140) (677.99,0.000128) (690.21,0.000112) (702.42,0.000100) (714.64,0.000079) (726.86,0.000073)} \closedcycle;
\draw[black,thick] (axis cs:\rAirMeanNum,0) -- (axis cs:\rAirMeanNum,\rAirHistMax);
\draw[green!45!black,thick,dashed] (axis cs:\rAirPlugNum,0) -- (axis cs:\rAirPlugNum,\rAirHistMax);
\draw[red!70!black,thick] (axis cs:\rAirBondNum,0) -- (axis cs:\rAirBondNum,\rAirHistMax);
\draw[red!70!black,dotted,thick] (axis cs:\rAirPnnNum,0) -- (axis cs:\rAirPnnNum,\rAirHistMax);
\end{groupplot}
\end{tikzpicture}
\caption{Density of the mitigated loss (truncated at $P_{99.5}$). Black: $\E L$. Green dashed: plug-in estimate. Red: critical-fractile security $B^\star$ at premium $2\%$ and $q=0.3$. Red dotted: $P_{99}$.}
\label{fig:hist}
\end{figure}
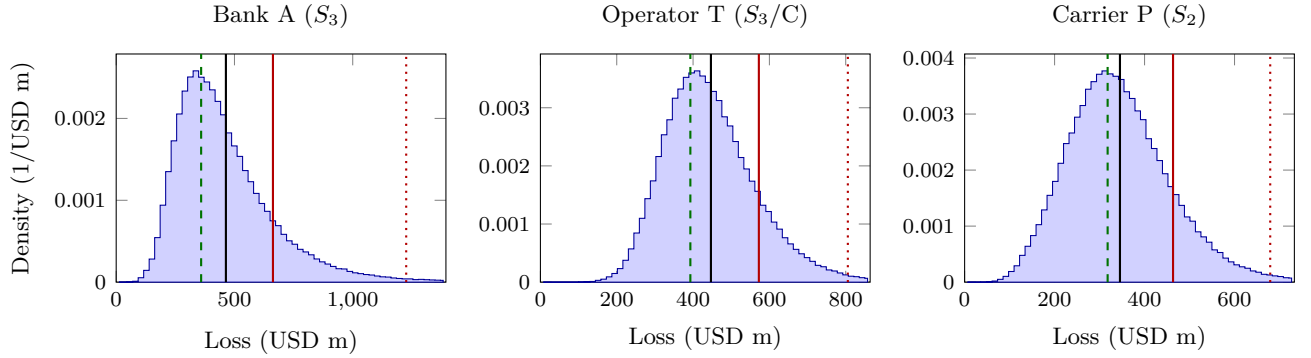

\subsection{RQ4: Channel decomposition of mean and tail}
\Cref{tab:channels,fig:stack} decompose $\E L$ and $\CVaR_{0.95}$ by channel, using the Euler allocation of \cref{prop:euler}. The platforms differ sharply in where their loss comes from.

For \emph{Bank~A}, customer attrition and deposit funding make up $\rBankAttrFundMeanShare\%$ of expected loss and $\rBankAttrFundTailShare\%$ of tail loss. The tail is a customer-flight phenomenon. Acute and incident outages contribute only $\rBankAcIncMeanShare\%$ of the mean, yet $\rBankAcIncTailShare\%$ of the tail, because outage hours enter the churn hazard. For \emph{Operator~T}, functional degradation and delayed launch dominate ($\rTelFunDelMeanShare\%$ of the mean). Leakage and peak denial on the interworking path accumulate steadily, and the enterprise-slicing launch is blocked. For \emph{Carrier~P}, functional loss (continuous pricing and degraded IROPS recovery) and the delayed NDC programme together account for $\rAirFunDelMeanShare\%$. Acute outages are a minor share of the mean in all cases. A model restricted to ``downtime cost'' would therefore miss most of the economic exposure, and this is the central empirical point of the paper.

\begin{table}[t]
\centering\footnotesize
\caption{Channel decomposition of $\E L$ and Euler allocation of $\CVaR_{0.95}$ (USD m and \% of total) for $s^\star$.}\label{tab:channels}
\setlength{\tabcolsep}{3pt}
\begin{tabular}{@{}l rrrr rrrr rrrr@{}}
\toprule
& \multicolumn{4}{c}{Bank A} & \multicolumn{4}{c}{Operator T} & \multicolumn{4}{c}{Carrier P}\\
\cmidrule(lr){2-5}\cmidrule(lr){6-9}\cmidrule(lr){10-13}
Channel & $\E L_k$ & \% & CVaR$_k$ & \% & $\E L_k$ & \% & CVaR$_k$ & \% & $\E L_k$ & \% & CVaR$_k$ & \%\\
\midrule
Switch-over outage & 12.9 & 2.8 & 70.5 & 6.4 & 1.5 & 0.3 & 2.8 & 0.4 & 2.6 & 0.7 & 7.4 & 1.2\\
Excess incidents & 35.2 & 7.6 & 100.7 & 9.1 & 10.7 & 2.4 & 18.6 & 2.5 & 11.1 & 3.2 & 47.2 & 7.5\\
Functional degradation & 30.6 & 6.6 & 37.6 & 3.4 & 188.4 & 42.2 & 275.6 & 36.5 & 192.5 & 55.8 & 334.8 & 52.9\\
Workaround operations & 37.1 & 8.0 & 42.2 & 3.8 & 31.0 & 6.9 & 36.5 & 4.8 & 22.7 & 6.6 & 31.8 & 5.0\\
Customer attrition & 153.6 & 33.1 & 421.8 & 38.2 & 91.0 & 20.4 & 207.7 & 27.5 & 7.9 & 2.3 & 13.1 & 2.1\\
Deposit funding & 120.6 & 26.0 & 346.3 & 31.3 & -- & -- & -- & -- & -- & -- & -- & --\\
Remediation cost & 30.7 & 6.6 & 32.9 & 3.0 & 30.3 & 6.8 & 32.8 & 4.3 & 16.0 & 4.6 & 29.3 & 4.6\\
Delayed launch & 43.2 & 9.3 & 52.8 & 4.8 & 93.9 & 21.0 & 181.0 & 24.0 & 92.5 & 26.8 & 169.0 & 26.7\\
\midrule
Total & 463.8 & 100.0 & 1,104.7 & 100.0 & 446.8 & 100.0 & 755.1 & 100.0 & 345.3 & 100.0 & 632.5 & 100.0\\
\bottomrule
\end{tabular}
\end{table}

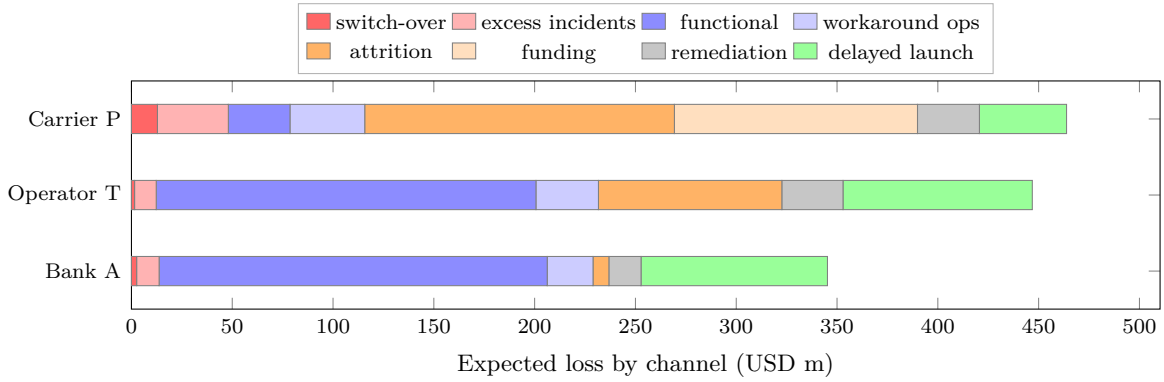
\begin{figure}[t]
\centering
\begin{tikzpicture}
\begin{axis}[xbar stacked,width=0.92\textwidth,height=4.6cm,/pgf/bar width=11pt,symbolic y coords={Air,Tel,Bank},ytick=data,yticklabels={Carrier P,Operator T,Bank A},xmin=0,xlabel={Expected loss by channel (USD m)},legend style={font=\scriptsize,at={(0.5,1.03)},anchor=south,legend columns=4},enlarge y limits=0.25]
\addplot[fill=red!60,draw=black!50] coordinates {(12.89,Bank) (1.55,Tel) (2.58,Air)};
\addplot[fill=red!30,draw=black!50] coordinates {(35.18,Bank) (10.74,Tel) (11.14,Air)};
\addplot[fill=blue!45,draw=black!50] coordinates {(30.57,Bank) (188.36,Tel) (192.52,Air)};
\addplot[fill=blue!20,draw=black!50] coordinates {(37.08,Bank) (31.00,Tel) (22.72,Air)};
\addplot[fill=orange!60,draw=black!50] coordinates {(153.61,Bank) (90.97,Tel) (7.87,Air)};
\addplot[fill=orange!25,draw=black!50] coordinates {(120.58,Bank) (0.00,Tel) (0.00,Air)};
\addplot[fill=gray!45,draw=black!50] coordinates {(30.72,Bank) (30.34,Tel) (15.95,Air)};
\addplot[fill=green!40,draw=black!50] coordinates {(43.22,Bank) (93.90,Tel) (92.47,Air)};
\legend{switch-over,excess incidents,functional,workaround ops,attrition,funding,remediation,delayed launch}
\end{axis}
\end{tikzpicture}
\caption{Channel composition of the mitigated expected loss.}
\label{fig:stack}
\end{figure}

\subsection{RQ5: Bias of plug-in estimates and strategy misranking}\label{sec:plugin}
The plug-in estimate understates the expected mitigated loss by $\rBankPlugGap\%$ (Bank~A), $\rTelPlugGap\%$ (Operator~T) and $\rAirPlugGap\%$ (Carrier~P). It lies at the $\rBankPlugRank$th, $\rTelPlugRank$th and $\rAirPlugRank$th percentiles of the loss distribution (\cref{tab:plugin}). The biases follow the theory.

\emph{Jensen bias in outage channels (\cref{thm:jensen}).} At central parameters, the expected switch-over loss conditional on an outage exceeds the loss at the median duration by a factor of $\rBankJswRatio$ (Bank~A: $\rBankJswE$ vs.\ $\rBankJswMed$), $\rTelJswRatio$ and $\rAirJswRatio$. For incidents the factors are $\rBankJincRatio$, $\rTelJincRatio$ and $\rAirJincRatio$. For Carrier~P they are largest because the 3-hour cascade hinge is crossed mainly in the tail (\cref{fig:acute}). Accordingly, the outage channels have the largest MC/plug-in ratios in \cref{tab:plugin} (up to $4.7\times$).

\emph{Fulkerson bias (\cref{prop:fulkerson}).} The modal critical path of $S_1$ gives $T_d=\rBankTdMode$\,yr for Bank~A. Mean durations give $\rBankTdMeanPath$\,yr, while the Monte Carlo mean is $\rBankTdMC$\,yr ($+\rBankFulkGap\%$ over the mean path) and the $P_{90}$ is $\rBankTdPninety$\,yr. The figures for Operator~T are $\rTelTdMode$, $\rTelTdMeanPath$, $\rTelTdMC$ ($+\rTelFulkGap\%$) and $\rTelTdPninety$, and for Carrier~P $\rAirTdMode$, $\rAirTdMeanPath$, $\rAirTdMC$ ($+\rAirFulkGap\%$) and $\rAirTdPninety$.

\emph{Convexity in the customer response.} Attrition is $1-e^{-\Lambda}$ in a hazard that is linear in lognormally distributed parameters, multiplied by lognormal margins. Its expectation exceeds its plug-in value (ratios $1.22$--$1.77$).

The more consequential finding is that \emph{the plug-in misranks the responses} (\cref{fig:rank}). For Carrier~P, the plug-in selects degrade-and-wait ($\rAirPlugBest$, plug-in $302.8$ vs.\ $317.9$ for $S_2$), whereas the full model shows that $S_0$ is the worst strategy ($385.0$ vs.\ $345.3$). The reason is that the modal $T_d$ for $S_1$ ($2.45$\,yr) exceeds the median restraint horizon, so in the plug-in world the design-around never pays off; the plug-in ignores the option value of completing remediation early. For Bank~A, the plug-in ranks $S_0$ ahead of both design-around variants, again the reverse of the full model. For Operator~T, the plug-in prefers \textsc{terminate} under every strategy, while the full model prefers \textsc{continue} under every strategy, with lighter tails. A security or a mitigation plan built on plug-in values would thus be both too small and attached to the wrong operational response.

The variance decomposition \eqref{eq:totvar} attributes $\rBankAleaShare\%$ (Bank~A), $\rTelAleaShare\%$ (Operator~T) and $\rAirAleaShare\%$ (Carrier~P) of $\Var L$ to aleatory event risk. The rest is epistemic. Bank~A's large aleatory share reflects the switch-over Bernoulli and the heavy-tailed outage durations feeding the churn hazard.

\begin{table}[t]
\centering\footnotesize
\caption{Plug-in vs.\ full-model channel means for $s^\star$ (USD m) and their ratio.}\label{tab:plugin}
\setlength{\tabcolsep}{4pt}
\begin{tabular}{@{}l rrr rrr rrr@{}}
\toprule
& \multicolumn{3}{c}{Bank A} & \multicolumn{3}{c}{Operator T} & \multicolumn{3}{c}{Carrier P}\\
\cmidrule(lr){2-4}\cmidrule(lr){5-7}\cmidrule(lr){8-10}
Channel & plug-in & MC & ratio & plug-in & MC & ratio & plug-in & MC & ratio\\
\midrule
Switch-over outage & 4.6 & 12.9 & 2.81 & 0.7 & 1.5 & 2.07 & 0.6 & 2.6 & 4.18\\
Excess incidents & 14.7 & 35.2 & 2.39 & 5.7 & 10.7 & 1.88 & 2.4 & 11.1 & 4.71\\
Functional degradation & 26.2 & 30.6 & 1.17 & 167.1 & 188.4 & 1.13 & 174.9 & 192.5 & 1.10\\
Workaround operations & 35.0 & 37.1 & 1.06 & 28.5 & 31.0 & 1.09 & 23.9 & 22.7 & 0.95\\
Customer attrition & 117.3 & 153.6 & 1.31 & 74.5 & 91.0 & 1.22 & 4.4 & 7.9 & 1.77\\
Deposit funding & 89.8 & 120.6 & 1.34 & -- & -- & -- & -- & -- & --\\
Remediation cost & 31.2 & 30.7 & 0.99 & 30.2 & 30.3 & 1.00 & 15.2 & 16.0 & 1.05\\
Delayed launch & 40.3 & 43.2 & 1.07 & 86.1 & 93.9 & 1.09 & 96.5 & 92.5 & 0.96\\
\midrule
Total & 359.1 & 463.8 & 1.29 & 392.9 & 446.8 & 1.14 & 317.9 & 345.3 & 1.09\\
\bottomrule
\end{tabular}
\end{table}

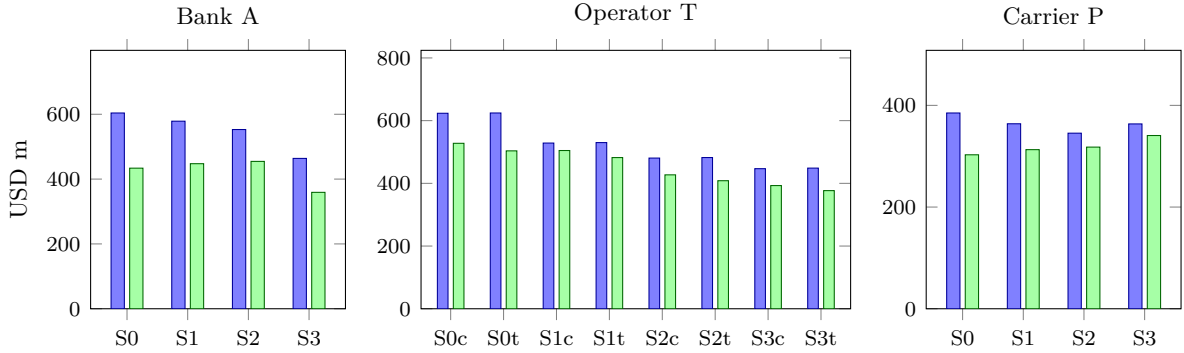
\begin{figure}[t]
\centering
\begin{tikzpicture}
\begin{groupplot}[group style={group size=3 by 1,horizontal sep=1.0cm},ybar,height=5.0cm,ymin=0,enlarge y limits={upper,value=0.32},xticklabel style={font=\scriptsize}]
\nextgroupplot[title={Bank A},ylabel={USD m},width=0.30\textwidth,/pgf/bar width=5pt,symbolic x coords={S0,S1,S2,S3},xtick=data,enlarge x limits=0.2]
\addplot[fill=blue!50,draw=blue!60!black] coordinates {(S0,603.9) (S1,578.6) (S2,552.9) (S3,463.8)};
\addplot[fill=green!35,draw=green!40!black] coordinates {(S0,433.7) (S1,447.5) (S2,454.6) (S3,359.1)};
\nextgroupplot[title={Operator T},width=0.44\textwidth,/pgf/bar width=4pt,symbolic x coords={S0c,S0t,S1c,S1t,S2c,S2t,S3c,S3t},xtick=data,enlarge x limits=0.08]
\addplot[fill=blue!50,draw=blue!60!black] coordinates {(S0c,623.5) (S0t,624.4) (S1c,528.5) (S1t,530.0) (S2c,480.6) (S2t,482.3) (S3c,446.8) (S3t,448.6)};
\addplot[fill=green!35,draw=green!40!black] coordinates {(S0c,527.7) (S0t,503.6) (S1c,504.7) (S1t,482.2) (S2c,427.1) (S2t,408.4) (S3c,392.9) (S3t,376.8)};
\nextgroupplot[title={Carrier P},width=0.30\textwidth,/pgf/bar width=5pt,symbolic x coords={S0,S1,S2,S3},xtick=data,enlarge x limits=0.2]
\addplot[fill=blue!50,draw=blue!60!black] coordinates {(S0,385.0) (S1,363.8) (S2,345.3) (S3,363.5)};
\addplot[fill=green!35,draw=green!40!black] coordinates {(S0,302.8) (S1,312.9) (S2,317.9) (S3,340.7)};
\end{groupplot}
\end{tikzpicture}
\caption{Strategy ranking under the full model (blue: Monte Carlo $\E L$) and the plug-in estimate (green). Suffix c/t: CCFH \textsc{continue}/\textsc{terminate}. The plug-in selects a different response for Operator~T and Carrier~P.}
\label{fig:rank}
\end{figure}

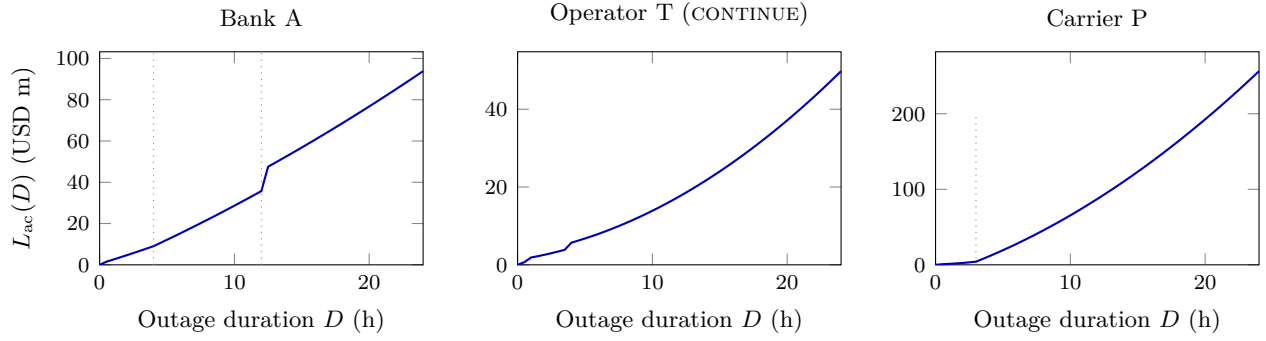
\begin{figure}[t]
\centering
\begin{tikzpicture}
\begin{groupplot}[group style={group size=3 by 1,horizontal sep=1.25cm},width=0.355\textwidth,height=4.4cm,xlabel={Outage duration $D$ (h)},xmin=0,xmax=24,ymin=0]
\nextgroupplot[title={Bank A},ylabel={$L_{\mathrm{ac}}(D)$ (USD m)}]
\addplot[blue!70!black,thick] coordinates {(0,0) (0.5,1.5084) (1,2.5336) (1.5,3.5756) (2,4.6344) (2.5,5.71) (3,6.8024) (3.5,7.9116) (4,9.0376) (4.5,10.58) (5,12.14) (5.5,13.716) (6,15.31) (6.5,16.92) (7,18.546) (7.5,20.19) (8,21.85) (8.5,23.528) (9,25.222) (9.5,26.932) (10,28.66) (10.5,30.404) (11,32.166) (11.5,33.944) (12,35.738) (12.5,47.55) (13,49.378) (13.5,51.224) (14,53.086) (14.5,54.964) (15,56.86) (15.5,58.772) (16,60.702) (16.5,62.648) (17,64.61) (17.5,66.59) (18,68.586) (18.5,70.6) (19,72.63) (19.5,74.676) (20,76.74) (20.5,78.82) (21,80.918) (21.5,83.032) (22,85.162) (22.5,87.31) (23,89.474) (23.5,91.656) (24,93.854)};
\draw[gray,dotted] (axis cs:4,0) -- (axis cs:4,200); \draw[gray,dotted] (axis cs:12,0) -- (axis cs:12,200);
\nextgroupplot[title={Operator T (\textsc{continue})}]
\addplot[blue!70!black,thick] coordinates {(0,0) (0.5,0.67699) (1,1.884) (1.5,2.221) (2,2.5879) (2.5,2.9849) (3,3.4119) (3.5,3.8689) (4,5.7059) (4.5,6.2229) (5,6.7699) (5.5,7.3468) (6,7.9538) (6.5,8.5908) (7,9.2578) (7.5,9.9548) (8,10.682) (8.5,11.439) (9,12.226) (9.5,13.043) (10,13.89) (10.5,14.767) (11,15.674) (11.5,16.611) (12,17.578) (12.5,18.575) (13,19.602) (13.5,20.659) (14,21.746) (14.5,22.863) (15,24.01) (15.5,25.187) (16,26.394) (16.5,27.631) (17,28.898) (17.5,30.195) (18,31.522) (18.5,32.878) (19,34.265) (19.5,35.682) (20,37.129) (20.5,38.606) (21,40.113) (21.5,41.65) (22,43.217) (22.5,44.814) (23,46.441) (23.5,48.098) (24,49.785)};
\nextgroupplot[title={Carrier P}]
\addplot[blue!70!black,thick] coordinates {(0,0) (0.5,0.81894) (1,1.2541) (1.5,1.8054) (2,2.473) (2.5,3.2568) (3,4.1568) (3.5,7.7813) (4,11.522) (4.5,15.379) (5,19.352) (5.5,23.442) (6,27.647) (6.5,31.969) (7,36.407) (7.5,40.961) (8,45.632) (8.5,50.418) (9,55.321) (9.5,60.34) (10,65.476) (10.5,70.727) (11,76.095) (11.5,81.579) (12,87.179) (12.5,92.895) (13,98.728) (13.5,104.68) (14,110.74) (14.5,116.92) (15,123.22) (15.5,129.63) (16,136.16) (16.5,142.81) (17,149.57) (17.5,156.45) (18,163.44) (18.5,170.56) (19,177.78) (19.5,185.13) (20,192.59) (20.5,200.16) (21,207.85) (21.5,215.66) (22,223.59) (22.5,231.63) (23,239.78) (23.5,248.06) (24,256.44)};
\draw[gray,dotted] (axis cs:3,0) -- (axis cs:3,200);
\end{groupplot}
\end{tikzpicture}
\caption{Acute loss functions at central parameters. The quadratic backlog term, the hinge at $4$\,h (redress) and the step at $12$\,h (enforcement) for Bank~A, and the $3$\,h cascade hinge for Carrier~P, are the sources of the Jensen gap in \cref{thm:jensen}.}
\label{fig:acute}
\end{figure}

\begin{figure}[t]
\centering
\begin{tikzpicture}
\begin{groupplot}[group style={group size=3 by 1,horizontal sep=1.25cm},width=0.355\textwidth,height=4.4cm,xlabel={Forced delay $\Delta$ (yr)},xmin=0,xmax=3,ymin=0,legend style={font=\scriptsize,at={(0.03,0.97)},anchor=north west}]
\nextgroupplot[title={Bank A ($V_0\gg V^\ast$)},ylabel={Loss (USD m)}]
\addplot[blue!70!black,thick] coordinates {(0,0) (0.1,4.3721) (0.2,8.6889) (0.3,12.951) (0.4,17.159) (0.5,21.314) (0.6,25.416) (0.7,29.466) (0.8,33.465) (0.9,37.412) (1,41.308) (1.1,45.153) (1.2,48.947) (1.3,52.691) (1.4,56.384) (1.5,60.026) (1.6,63.617) (1.7,67.158) (1.8,70.649) (1.9,74.089) (2,77.479) (2.1,80.819) (2.2,84.11) (2.3,87.352) (2.4,90.545) (2.5,93.69) (2.6,96.787) (2.7,99.837) (2.8,102.84) (2.9,105.8) (3,108.71)};
\addplot[red!70!black,thick,dashed] coordinates {(0,0) (0.1,4.3721) (0.2,8.6889) (0.3,12.951) (0.4,17.159) (0.5,21.314) (0.6,25.416) (0.7,29.466) (0.8,33.465) (0.9,37.413) (1,41.311) (1.1,45.159) (1.2,48.958) (1.3,52.709) (1.4,56.412) (1.5,60.068) (1.6,63.678) (1.7,67.241) (1.8,70.759) (1.9,74.232) (2,77.66) (2.1,81.045) (2.2,84.387) (2.3,87.685) (2.4,90.942) (2.5,94.156) (2.6,97.33) (2.7,100.46) (2.8,103.56) (2.9,106.61) (3,109.62)};
\legend{$\ell(\Delta)$ option,$\ell_{\mathrm{DCF}}(\Delta)$}
\nextgroupplot[title={Operator T}]
\addplot[blue!70!black,thick] coordinates {(0,0) (0.1,9.4225) (0.2,18.643) (0.3,27.559) (0.4,36.137) (0.5,44.388) (0.6,52.336) (0.7,60.008) (0.8,67.426) (0.9,74.612) (1,81.582) (1.1,88.353) (1.2,94.937) (1.3,101.35) (1.4,107.59) (1.5,113.68) (1.6,119.62) (1.7,125.42) (1.8,131.09) (1.9,136.62) (2,142.04) (2.1,147.34) (2.2,152.52) (2.3,157.59) (2.4,162.56) (2.5,167.43) (2.6,172.2) (2.7,176.87) (2.8,181.45) (2.9,185.94) (3,190.34)};
\addplot[red!70!black,thick,dashed] coordinates {(0,0) (0.1,9.4236) (0.2,18.696) (0.3,27.82) (0.4,36.797) (0.5,45.629) (0.6,54.32) (0.7,62.87) (0.8,71.282) (0.9,79.559) (1,87.702) (1.1,95.713) (1.2,103.59) (1.3,111.35) (1.4,118.98) (1.5,126.48) (1.6,133.86) (1.7,141.12) (1.8,148.27) (1.9,155.29) (2,162.21) (2.1,169.01) (2.2,175.69) (2.3,182.27) (2.4,188.74) (2.5,195.11) (2.6,201.37) (2.7,207.52) (2.8,213.58) (2.9,219.53) (3,225.39)};
\nextgroupplot[title={Carrier P}]
\addplot[blue!70!black,thick] coordinates {(0,0) (0.1,6.5582) (0.2,13.033) (0.3,19.427) (0.4,25.739) (0.5,31.971) (0.6,38.122) (0.7,44.192) (0.8,50.18) (0.9,56.085) (1,61.906) (1.1,67.642) (1.2,73.291) (1.3,78.854) (1.4,84.332) (1.5,89.723) (1.6,95.028) (1.7,100.25) (1.8,105.39) (1.9,110.44) (2,115.41) (2.1,120.3) (2.2,125.11) (2.3,129.85) (2.4,134.5) (2.5,139.08) (2.6,143.59) (2.7,148.02) (2.8,152.38) (2.9,156.67) (3,160.89)};
\addplot[red!70!black,thick,dashed] coordinates {(0,0) (0.1,6.5582) (0.2,13.033) (0.3,19.427) (0.4,25.739) (0.5,31.971) (0.6,38.124) (0.7,44.2) (0.8,50.198) (0.9,56.119) (1,61.966) (1.1,67.738) (1.2,73.437) (1.3,79.064) (1.4,84.619) (1.5,90.103) (1.6,95.517) (1.7,100.86) (1.8,106.14) (1.9,111.35) (2,116.49) (2.1,121.57) (2.2,126.58) (2.3,131.53) (2.4,136.41) (2.5,141.23) (2.6,145.99) (2.7,150.69) (2.8,155.33) (2.9,159.91) (3,164.43)};
\end{groupplot}
\end{tikzpicture}
\caption{Forced-delay loss of the planned launch: real-option value $\ell(\Delta)$ (\cref{thm:delay}) vs.\ the static commit-and-delay benchmark. The gap is largest for Operator~T, whose launch is closest to its exercise threshold ($V_0/V^\ast=\rTelVO/\rTelVstar$) and most volatile ($\sigma=0.45$).}
\label{fig:delay}
\end{figure}
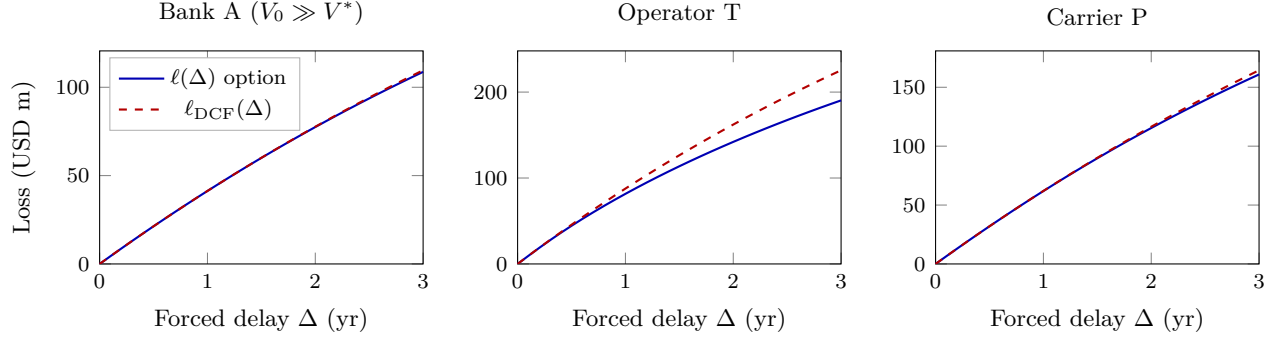

\paragraph{Delayed commercialisation.} At central values all three launches are in the money ($V_0>V^\ast$), so \cref{thm:delay}(ii) applies and the static DCF delay loss is an upper bound. For Bank~A ($V_0/V^\ast=\rBankVO/\rBankVstar$, $\beta=\rBankBeta$) and Carrier~P ($\rAirVO/\rAirVstar$, $\beta=\rAirBeta$) the bound is nearly tight. Here the option to wait is worth little, because immediate launch is strongly optimal. For Operator~T ($\rTelVO/\rTelVstar$, $\beta=\rTelBeta$), the option loss at $\Delta=2$ is $\rTelDelayOptTwo$ against $\rTelDelayDcfTwo$ for DCF. The expected delay channel under $s^\star$ is $\rTelDelayMean$ against $\rTelDelayDcfMean$ (\cref{fig:delay}). A DCF-based delay claim therefore overstates the loss for options near their exercise boundary. This is a point an opposing expert can be expected to raise, and the option model removes it.

\subsection{RQ6: Sensitivity analysis}
\Cref{tab:sobol,fig:sobol} report grouped Sobol indices of $\E[L\mid\theta]$, and \cref{fig:tornado} shows one-at-a-time swings. The first-order indices sum to $\rBankSobSumS$, $\rTelSobSumS$ and $\rAirSobSumS$, so the response is close to additive across mechanism groups. The exception is design-around execution, whose total index exceeds its first-order index (interaction with $T_R$ through $\tau=T_R\wedge T_d$).

The dominant drivers differ by platform. For Bank~A, uncertainty is concentrated in the \emph{customer response} ($S_T=\rBankSobTopST$). The SME degradation hazard alone swings the conditional mean by more than $250$ (\cref{fig:tornado}). Evidence on customer response, such as the econometrics of \cref{sec:econ}, therefore has the highest value of information. For Operator~T, uncertainty is spread across customer response ($\rTelSobTopST$), design-around execution ($\rTelSobSecST$), the delayed launch and the restraint horizon. For Carrier~P, the \emph{restraint horizon} dominates ($S_T=\rAirSobTopST$), followed by economic scale ($\rAirSobSecST$). The loss is driven by flows whose duration is set by $T_R$ when remediation is slow. Fallback-capability caps matter little for Carrier~P, whose losses are mostly binary (continuous pricing on or off), but noticeably for Operator~T ($S_T\approx0.08$).

\begin{table}[t]
\centering\footnotesize
\caption{Grouped Sobol indices of $\E[L\mid\theta]$ for $s^\star$ (Saltelli, $N=4096$). Bootstrap 95\% half-widths are $\le0.03$.}\label{tab:sobol}
\begin{tabular}{@{}l rr rr rr@{}}
\toprule
& \multicolumn{2}{c}{Bank A} & \multicolumn{2}{c}{Operator T} & \multicolumn{2}{c}{Carrier P}\\
\cmidrule(lr){2-3}\cmidrule(lr){4-5}\cmidrule(lr){6-7}
Group & $S_1$ & $S_T$ & $S_1$ & $S_T$ & $S_1$ & $S_T$\\
\midrule
Restraint horizon & 0.106 & 0.141 & 0.140 & 0.196 & 0.476 & 0.580\\
Switch-over outage & 0.036 & 0.046 & 0.001 & 0.001 & 0.003 & 0.003\\
Outage economics & 0.004 & 0.008 & 0.001 & 0.001 & 0.000 & 0.001\\
Incident process & 0.068 & 0.078 & 0.004 & 0.005 & 0.006 & 0.006\\
Fallback capability & 0.027 & 0.033 & 0.078 & 0.080 & 0.001 & 0.001\\
Economic scale & 0.008 & 0.009 & 0.085 & 0.093 & 0.219 & 0.228\\
Workaround operations & 0.003 & 0.002 & 0.004 & 0.005 & 0.002 & 0.002\\
Customer response & 0.596 & 0.639 & 0.242 & 0.240 & 0.002 & 0.003\\
Delayed launch & 0.007 & 0.009 & 0.196 & 0.207 & 0.100 & 0.112\\
Design-around execution & 0.065 & 0.091 & 0.171 & 0.234 & 0.066 & 0.152\\
IROPS process & -- & -- & -- & -- & 0.013 & 0.015\\
\bottomrule
\end{tabular}
\end{table}

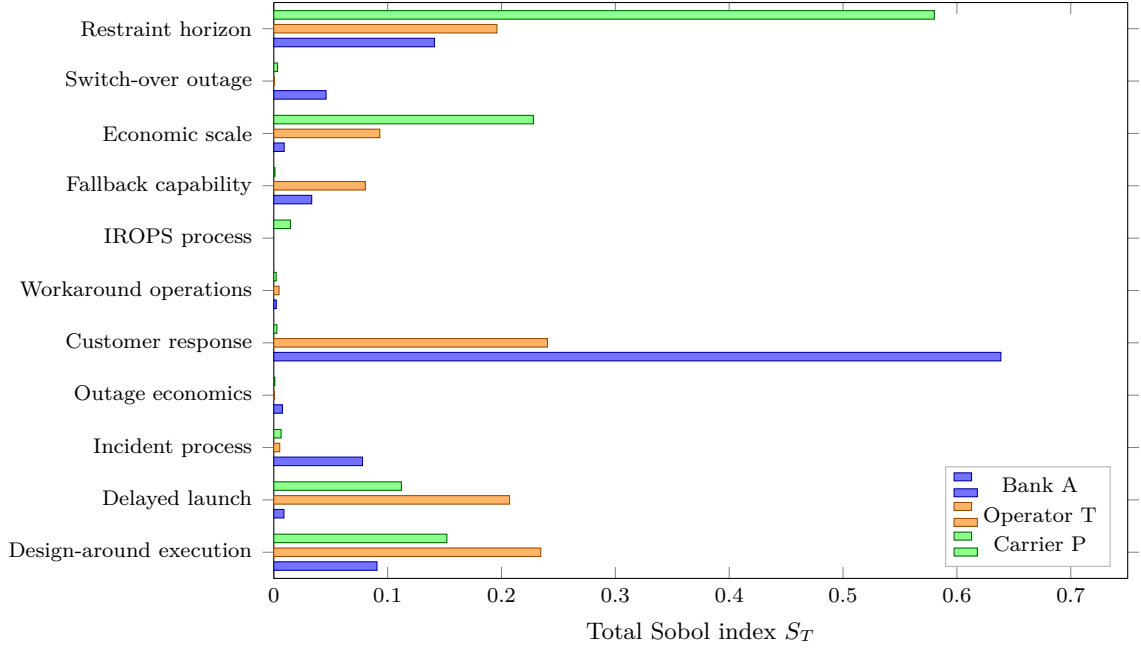
\begin{figure}[t]
\centering
\begin{tikzpicture}
\begin{axis}[xbar,width=0.78\textwidth,height=9.2cm,/pgf/bar width=3.2pt,xmin=0,xmax=0.75,
  symbolic y coords={{Restraint horizon},{Switch-over outage},{Economic scale},{Fallback capability},{IROPS process},{Workaround operations},{Customer response},{Outage economics},{Incident process},{Delayed launch},{Design-around execution}},ytick={{Restraint horizon},{Switch-over outage},{Economic scale},{Fallback capability},{IROPS process},{Workaround operations},{Customer response},{Outage economics},{Incident process},{Delayed launch},{Design-around execution}},yticklabel style={font=\scriptsize},
  y dir=reverse,enlarge y limits=0.05,xlabel={Total Sobol index $S_T$},
  legend style={font=\scriptsize,at={(0.98,0.02)},anchor=south east}]
\addplot[fill=blue!55,draw=blue!60!black] coordinates {(0.1411,{Restraint horizon}) (0.0459,{Switch-over outage}) (0.0076,{Outage economics}) (0.0779,{Incident process}) (0.0333,{Fallback capability}) (0.0091,{Economic scale}) (0.0022,{Workaround operations}) (0.6386,{Customer response}) (0.0089,{Delayed launch}) (0.0906,{Design-around execution})};
\addplot[fill=orange!60,draw=orange!60!black] coordinates {(0.1960,{Restraint horizon}) (0.0006,{Switch-over outage}) (0.0931,{Economic scale}) (0.0804,{Fallback capability}) (0.0046,{Workaround operations}) (0.2404,{Customer response}) (0.0005,{Outage economics}) (0.0052,{Incident process}) (0.2071,{Delayed launch}) (0.2345,{Design-around execution})};
\addplot[fill=green!45,draw=green!40!black] coordinates {(0.5802,{Restraint horizon}) (0.0033,{Switch-over outage}) (0.2280,{Economic scale}) (0.0010,{Fallback capability}) (0.0147,{IROPS process}) (0.0021,{Workaround operations}) (0.0027,{Customer response}) (0.0009,{Outage economics}) (0.0064,{Incident process}) (0.1121,{Delayed launch}) (0.1520,{Design-around execution})};
\legend{Bank A,Operator T,Carrier P}
\end{axis}
\end{tikzpicture}
\caption{Total Sobol indices of $\E[L\mid\theta]$ by mechanism group (first-order indices in \cref{tab:sobol}). The IROPS group exists only for Carrier~P.}
\label{fig:sobol}
\end{figure}

\begin{figure}[p]
\centering
\begin{tikzpicture}
\begin{groupplot}[group style={group size=1 by 3,vertical sep=2.3cm},xbar,/pgf/bar shift=0pt,width=0.72\textwidth,height=5.6cm,/pgf/bar width=4.5pt,ytick=data,yticklabel style={font=\scriptsize},enlarge y limits=0.07,xlabel={$\Delta\,\E[L\mid\theta]$ relative to base (USD m)},legend style={font=\scriptsize,at={(0.98,0.04)},anchor=south east}]
\nextgroupplot[title={Bank A (base \rBankTornBase)},symbolic y coords={{cap\_avb},{spread},{cap\_liq},{beta\_out},{beta\_deg\_r},{dep\_mult},{m\_s},{u\_size},{T\_R},{lam\_mult},{u\_err},{beta\_deg\_s}}]
\addplot[fill=blue!45,draw=blue!60!black] coordinates {(-21.31,{cap\_avb}) (-21.35,{spread}) (-24.24,{cap\_liq}) (-16.14,{beta\_out}) (-22.90,{beta\_deg\_r}) (-31.32,{dep\_mult}) (-31.88,{m\_s}) (-37.83,{u\_size}) (-93.04,{T\_R}) (-36.85,{lam\_mult}) (-43.97,{u\_err}) (-80.80,{beta\_deg\_s})};
\addplot[fill=red!45,draw=red!60!black] coordinates {(21.28,{cap\_avb}) (22.15,{spread}) (22.35,{cap\_liq}) (34.78,{beta\_out}) (49.26,{beta\_deg\_r}) (46.01,{dep\_mult}) (46.83,{m\_s}) (45.02,{u\_size}) (4.69,{T\_R}) (61.49,{lam\_mult}) (56.02,{u\_err}) (171.59,{beta\_deg\_s})};
\legend{swing to lower value,swing to higher value}
\nextgroupplot[title={Operator T (base \rTelTornBase)},symbolic y coords={{u\_A},{u\_E},{leak},{u\_C},{delta\_V},{prem\_rate},{T\_R},{cap\_iwf},{u\_size},{u\_err},{beta\_prem},{V0}}]
\addplot[fill=blue!45,draw=blue!60!black] coordinates {(-14.21,{u\_A}) (-14.21,{u\_E}) (-12.60,{leak}) (-17.15,{u\_C}) (-21.06,{delta\_V}) (-21.43,{prem\_rate}) (-60.18,{T\_R}) (-39.12,{cap\_iwf}) (-42.91,{u\_size}) (-48.05,{u\_err}) (-41.40,{beta\_prem}) (-61.10,{V0})};
\addplot[fill=red!45,draw=red!60!black] coordinates {(17.00,{u\_A}) (17.00,{u\_E}) (23.91,{leak}) (22.76,{u\_C}) (21.47,{delta\_V}) (35.79,{prem\_rate}) (6.94,{T\_R}) (50.57,{cap\_iwf}) (50.61,{u\_size}) (60.64,{u\_err}) (87.22,{beta\_prem}) (69.93,{V0})};
\nextgroupplot[title={Carrier P (base \rAirTornBase)},symbolic y coords={{fte},{lam\_I},{lam\_mult},{pax\_I},{u\_size},{u\_err},{delta\_V},{cp\_share},{fare},{V0},{cp\_uplift},{T\_R}}]
\addplot[fill=blue!45,draw=blue!60!black] coordinates {(-6.72,{fte}) (-7.14,{lam\_I}) (-7.24,{lam\_mult}) (-9.86,{pax\_I}) (-23.17,{u\_size}) (-28.02,{u\_err}) (-19.43,{delta\_V}) (-23.80,{cp\_share}) (-22.48,{fare}) (-37.41,{V0}) (-57.22,{cp\_uplift}) (-144.63,{T\_R})};
\addplot[fill=red!45,draw=red!60!black] coordinates {(7.60,{fte}) (7.97,{lam\_I}) (12.89,{lam\_mult}) (14.49,{pax\_I}) (4.50,{u\_size}) (0.68,{u\_err}) (20.71,{delta\_V}) (22.01,{cp\_share}) (25.55,{fare}) (51.42,{V0}) (76.06,{cp\_uplift}) (11.45,{T\_R})};
\end{groupplot}
\end{tikzpicture}
\caption{Tornado diagrams: change in the mean-field loss when one parameter moves from its median to its $P_{10}$ or $P_{90}$, with all others at their medians. Parameters are ordered by swing; identifiers are those of the calibration tables (\texttt{u\_*} are design-around drivers: size, estimation error and activity durations).}
\label{fig:tornado}
\end{figure}
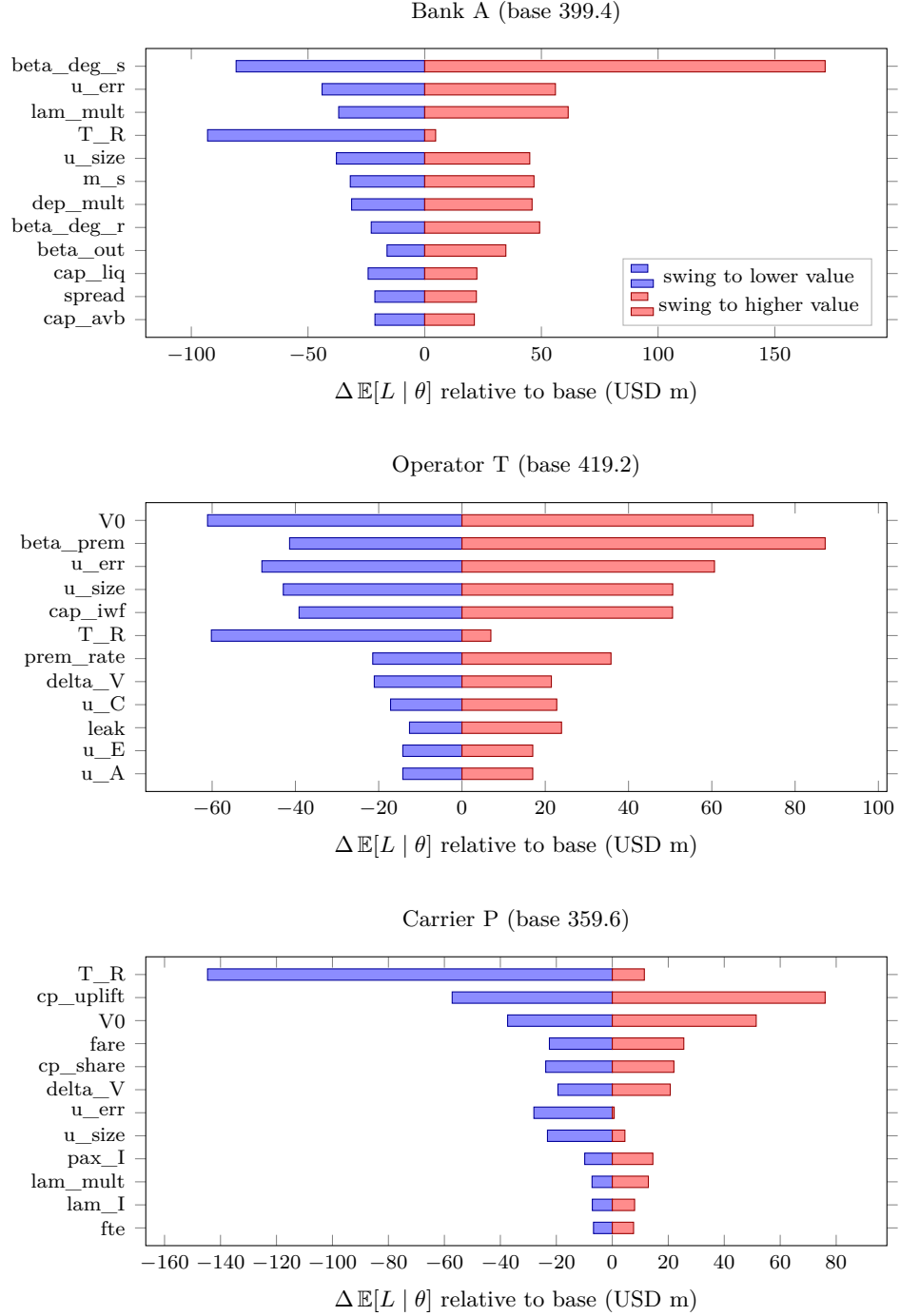

\subsection{RQ7: Security sizing}\label{sec:bondexp}
\Cref{tab:bond} applies \cref{thm:bond} on a grid of annual premium rates ($1$--$3\%$, so $c=\text{premium}\times\E T_R$) and wrongful-restraint probabilities $q\in[0.2,0.5]$, with $w=1$. At the central configuration (premium $2\%$, $q=0.3$), $\alpha^\star=\rBankBondAlpha$. The optimal security is $B^\star=\rBankBondB$ for Bank~A ($\rBankBondRatio\times\E L$), $\rTelBondB$ for Operator~T ($\rTelBondRatio\times$) and $\rAirBondB$ for Carrier~P ($\rAirBondRatio\times$). Across the grid, $B^\star/\E L$ ranges from $1.09$ to $1.94$. It increases with $q$ and decreases with the premium, as \cref{thm:bond}(a) predicts. It is largest for the most skewed distribution (Bank~A).

\Cref{fig:bondcurve} plots the expected uncompensated loss $U(B)=\E(L-B)^+$ relative to $\E L$. Posting $B=\E L$ leaves $\rBankUatEL\%$ (Bank~A), $\rTelUatEL\%$ (Operator~T) and $\rAirUatEL\%$ (Carrier~P) of the expected loss uncompensated. At $B^\star$ these fall to $\rBankUatB\%$, $\rTelUatB\%$ and $\rAirUatB\%$, at a carrying cost of $c\cdot(B^\star-\E L)$. The marginal security dollar has value $qw\Prob(L>B)$. It is worth posting while this exceeds $c$, which is exactly the first-order condition in the proof of \cref{thm:bond}.

\begin{table}[t]
\centering\footnotesize
\caption{Critical-fractile security $B^\star$ (USD m) and ratio $B^\star/\E L$ (in parentheses), with $w=1$ and $c=\text{premium}\times\E T_R$.}\label{tab:bond}
\begin{tabular}{@{}rrrrrrr@{}}
\toprule
Premium & $c$ & $q$ & $\alpha^\star$ & Bank A & Operator T & Carrier P\\
\midrule
1\% & 0.0217 & 0.2 & 0.892 & 718.1 (1.55) & 599.8 (1.34) & 487.8 (1.41)\\
1\% & 0.0217 & 0.3 & 0.928 & 798.5 (1.72) & 636.5 (1.42) & 521.4 (1.51)\\
1\% & 0.0217 & 0.4 & 0.946 & 854.4 (1.84) & 662.4 (1.48) & 544.1 (1.58)\\
1\% & 0.0217 & 0.5 & 0.957 & 899.3 (1.94) & 681.9 (1.53) & 562.1 (1.63)\\
\midrule
2\% & 0.0433 & 0.2 & 0.783 & 585.2 (1.26) & 531.7 (1.19) & 426.2 (1.23)\\
2\% & 0.0433 & 0.3 & 0.856 & 662.7 (1.43) & 572.4 (1.28) & 463.3 (1.34)\\
2\% & 0.0433 & 0.4 & 0.892 & 718.1 (1.55) & 599.8 (1.34) & 487.8 (1.41)\\
2\% & 0.0433 & 0.5 & 0.913 & 762.3 (1.64) & 620.3 (1.39) & 506.5 (1.47)\\
\midrule
3\% & 0.0650 & 0.2 & 0.675 & 505.7 (1.09) & 487.3 (1.09) & 385.4 (1.12)\\
3\% & 0.0650 & 0.3 & 0.783 & 585.2 (1.26) & 531.7 (1.19) & 426.2 (1.23)\\
3\% & 0.0650 & 0.4 & 0.838 & 640.2 (1.38) & 561.1 (1.26) & 452.6 (1.31)\\
3\% & 0.0650 & 0.5 & 0.870 & 683.1 (1.47) & 582.5 (1.30) & 472.4 (1.37)\\
\bottomrule
\end{tabular}
\end{table}

\begin{figure}[t]
\centering
\begin{tikzpicture}
\begin{axis}[width=0.62\textwidth,height=5.4cm,xlabel={Security $B/\E L$},ylabel={$\E(L-B)^+/\E L$},xmin=0,xmax=3,ymin=0,ymax=1,legend style={font=\scriptsize}]
\addplot[blue!70!black,thick] coordinates {(0.000,1.00000) (0.050,0.95000) (0.100,0.90000) (0.150,0.85000) (0.200,0.80001) (0.250,0.75005) (0.300,0.70020) (0.350,0.65058) (0.400,0.60143) (0.450,0.55312) (0.500,0.50613) (0.550,0.46090) (0.600,0.41785) (0.650,0.37735) (0.700,0.33962) (0.750,0.30484) (0.800,0.27301) (0.850,0.24404) (0.900,0.21784) (0.950,0.19427) (1.000,0.17317) (1.050,0.15431) (1.100,0.13747) (1.150,0.12247) (1.200,0.10915) (1.250,0.09733) (1.300,0.08685) (1.350,0.07755) (1.400,0.06933) (1.450,0.06203) (1.500,0.05557) (1.550,0.04985) (1.600,0.04475) (1.650,0.04022) (1.700,0.03619) (1.750,0.03261) (1.800,0.02944) (1.850,0.02662) (1.900,0.02411) (1.950,0.02188) (2.000,0.01988) (2.050,0.01809) (2.100,0.01650) (2.150,0.01507) (2.200,0.01378) (2.250,0.01262) (2.300,0.01158) (2.350,0.01064) (2.400,0.00979) (2.450,0.00903) (2.500,0.00833) (2.550,0.00771) (2.600,0.00715) (2.650,0.00663) (2.700,0.00616) (2.750,0.00572) (2.800,0.00533) (2.850,0.00498) (2.900,0.00466) (2.950,0.00437) (3.000,0.00410)};
\addplot[orange!85!black,thick] coordinates {(0.000,1.00000) (0.050,0.95000) (0.100,0.90000) (0.150,0.85000) (0.200,0.80000) (0.250,0.75000) (0.300,0.70000) (0.350,0.65002) (0.400,0.60008) (0.450,0.55023) (0.500,0.50060) (0.550,0.45139) (0.600,0.40299) (0.650,0.35587) (0.700,0.31057) (0.750,0.26776) (0.800,0.22800) (0.850,0.19179) (0.900,0.15944) (0.950,0.13112) (1.000,0.10675) (1.050,0.08610) (1.100,0.06888) (1.150,0.05471) (1.200,0.04317) (1.250,0.03386) (1.300,0.02643) (1.350,0.02054) (1.400,0.01591) (1.450,0.01230) (1.500,0.00949) (1.550,0.00730) (1.600,0.00562) (1.650,0.00432) (1.700,0.00333) (1.750,0.00256) (1.800,0.00198) (1.850,0.00153) (1.900,0.00119) (1.950,0.00093) (2.000,0.00072) (2.050,0.00056) (2.100,0.00044) (2.150,0.00035) (2.200,0.00028) (2.250,0.00023) (2.300,0.00019) (2.350,0.00015) (2.400,0.00013) (2.450,0.00011) (2.500,0.00009) (2.550,0.00008) (2.600,0.00007) (2.650,0.00006) (2.700,0.00005) (2.750,0.00005) (2.800,0.00005) (2.850,0.00004) (2.900,0.00004) (2.950,0.00004) (3.000,0.00003)};
\addplot[green!45!black,thick] coordinates {(0.000,1.00000) (0.050,0.95000) (0.100,0.90000) (0.150,0.85000) (0.200,0.80001) (0.250,0.75003) (0.300,0.70013) (0.350,0.65038) (0.400,0.60096) (0.450,0.55205) (0.500,0.50390) (0.550,0.45676) (0.600,0.41097) (0.650,0.36685) (0.700,0.32480) (0.750,0.28511) (0.800,0.24808) (0.850,0.21401) (0.900,0.18307) (0.950,0.15533) (1.000,0.13079) (1.050,0.10935) (1.100,0.09083) (1.150,0.07502) (1.200,0.06166) (1.250,0.05045) (1.300,0.04112) (1.350,0.03341) (1.400,0.02708) (1.450,0.02192) (1.500,0.01772) (1.550,0.01433) (1.600,0.01160) (1.650,0.00941) (1.700,0.00765) (1.750,0.00623) (1.800,0.00507) (1.850,0.00415) (1.900,0.00340) (1.950,0.00280) (2.000,0.00232) (2.050,0.00193) (2.100,0.00162) (2.150,0.00137) (2.200,0.00117) (2.250,0.00101) (2.300,0.00087) (2.350,0.00076) (2.400,0.00068) (2.450,0.00060) (2.500,0.00054) (2.550,0.00049) (2.600,0.00045) (2.650,0.00041) (2.700,0.00038) (2.750,0.00036) (2.800,0.00034) (2.850,0.00032) (2.900,0.00030) (2.950,0.00028) (3.000,0.00027)};
\draw[gray,dashed] (axis cs:1,0) -- (axis cs:1,1);
\draw[blue!70!black,dotted,thick] (axis cs:\rBankBondRatioNum,0) -- (axis cs:\rBankBondRatioNum,0.5);
\draw[orange!85!black,dotted,thick] (axis cs:\rTelBondRatioNum,0) -- (axis cs:\rTelBondRatioNum,0.5);
\draw[green!45!black,dotted,thick] (axis cs:\rAirBondRatioNum,0) -- (axis cs:\rAirBondRatioNum,0.5);
\legend{Bank A,Operator T,Carrier P}
\end{axis}
\end{tikzpicture}
\caption{Expected uncompensated loss as a function of the security. Dashed: $B=\E L$. Dotted: $B^\star$ at premium $2\%$, $q=0.3$.}
\label{fig:bondcurve}
\end{figure}
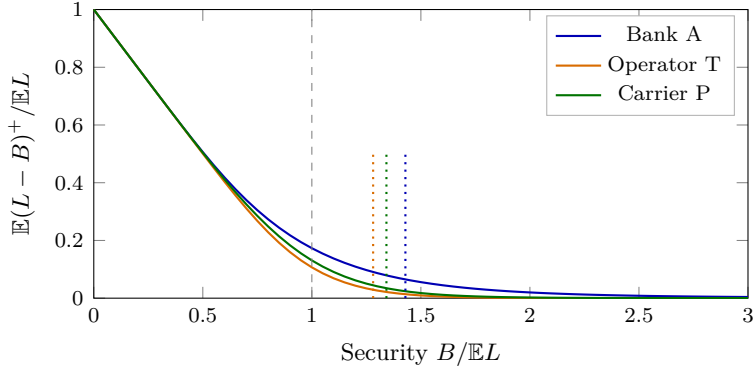

\subsection{RQ8a: Econometric calibration of customer response}\label{sec:econ}
The customer-response parameters dominate Bank~A's uncertainty. In practice they are calibrated from post-event panels, for example monthly churn by branch or region after staggered roll-outs of degraded functionality. We simulate this setting to test the estimators. The panel has $90$ units observed for $48$ months. Three cohorts of $20$ units adopt the degraded fallback in months $20$, $26$ and $32$, with long-run effects of $0.30$, $0.20$ and $0.10$ percentage points of monthly churn, ramping in as $1-e^{-(e+1)/3}$ in event time $e$. Thirty units are never treated. Errors are AR(1) with coefficient $0.3$, and the model includes unit effects, seasonality and a trend.

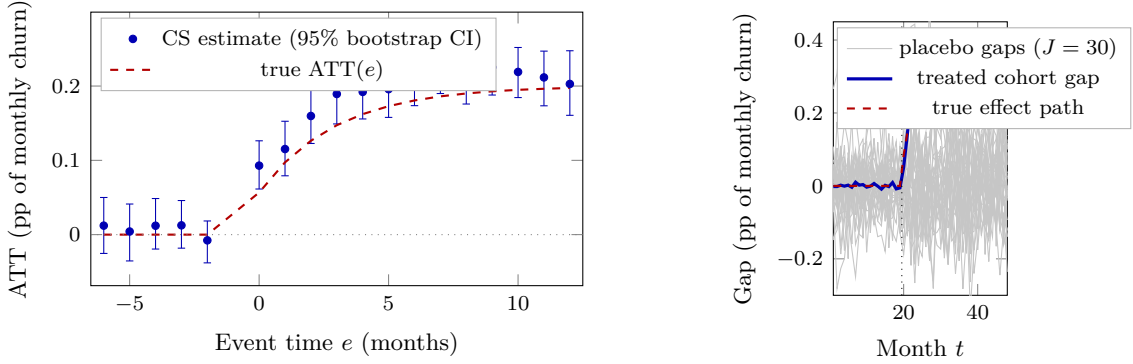
\begin{figure}[t]
\centering
\begin{minipage}[t]{0.49\textwidth}\centering
\begin{tikzpicture}
\begin{axis}[width=\textwidth,height=5.2cm,xlabel={Event time $e$ (months)},ylabel={ATT (pp of monthly churn)},xmin=-6.5,xmax=12.5,legend style={font=\scriptsize,at={(0.02,0.98)},anchor=north west}]
\addplot[only marks,mark=*,mark size=1.4pt,blue!70!black,error bars/.cd,y dir=both,y explicit] coordinates {(-6,0.0121) -= (0,0.0374) += (0,0.0379) (-5,0.0042) -= (0,0.0396) += (0,0.0370) (-4,0.0120) -= (0,0.0314) += (0,0.0366) (-3,0.0125) -= (0,0.0307) += (0,0.0334) (-2,-0.0077) -= (0,0.0303) += (0,0.0261) (0,0.0929) -= (0,0.0315) += (0,0.0334) (1,0.1152) -= (0,0.0360) += (0,0.0374) (2,0.1597) -= (0,0.0370) += (0,0.0382) (3,0.1893) -= (0,0.0403) += (0,0.0420) (4,0.1921) -= (0,0.0364) += (0,0.0372) (5,0.1958) -= (0,0.0380) += (0,0.0407) (6,0.2130) -= (0,0.0394) += (0,0.0391) (7,0.2291) -= (0,0.0392) += (0,0.0397) (8,0.2181) -= (0,0.0423) += (0,0.0417) (9,0.2252) -= (0,0.0373) += (0,0.0425) (10,0.2190) -= (0,0.0345) += (0,0.0329) (11,0.2116) -= (0,0.0382) += (0,0.0354) (12,0.2028) -= (0,0.0422) += (0,0.0447)};
\addplot[red!70!black,thick,dashed] coordinates {(-6,0.0000) (-5,0.0000) (-4,0.0000) (-3,0.0000) (-2,0.0000) (0,0.0567) (1,0.0973) (2,0.1264) (3,0.1473) (4,0.1622) (5,0.1729) (6,0.1806) (7,0.1861) (8,0.1900) (9,0.1929) (10,0.1949) (11,0.1963) (12,0.1974)};
\addplot[gray,dotted,forget plot] coordinates {(-6.5,0) (12.5,0)};
\legend{CS estimate (95\% bootstrap CI),true $\mathrm{ATT}(e)$}
\end{axis}
\end{tikzpicture}
\end{minipage}\hfill
\begin{minipage}[t]{0.49\textwidth}\centering
\begin{tikzpicture}
\begin{axis}[width=0.48\textwidth,height=5.2cm,xlabel={Month $t$},ylabel={Gap (pp of monthly churn)},xmin=1,xmax=48,ymin=-0.3,ymax=0.45,legend style={font=\scriptsize,at={(0.02,0.98)},anchor=north west}]
\addplot[gray!45,thin] coordinates {(1,0.0279) (2,-0.0050) (3,0.0068) (4,0.1116) (5,-0.0940) (6,-0.0790) (7,0.0118) (8,-0.0484) (9,0.0017) (10,0.0892) (11,0.1211) (12,0.0346) (13,0.0284) (14,0.0433) (15,0.0597) (16,-0.0998) (17,-0.1277) (18,-0.1518) (19,-0.0826) (20,-0.1439) (21,-0.1237) (22,-0.1193) (23,-0.0558) (24,-0.1653) (25,-0.1934) (26,-0.0154) (27,0.0081) (28,0.0491) (29,0.0111) (30,-0.0700) (31,-0.0573) (32,0.0130) (33,-0.1198) (34,-0.0054) (35,0.0529) (36,0.1361) (37,0.0796) (38,-0.0076) (39,0.0927) (40,0.0382) (41,-0.1315) (42,-0.0681) (43,-0.1171) (44,0.0217) (45,-0.0027) (46,0.0893) (47,-0.0350) (48,-0.0193)};
\addplot[gray!45,thin,forget plot] coordinates {(1,-0.0714) (2,-0.0327) (3,0.1098) (4,-0.0713) (5,0.0097) (6,-0.0427) (7,0.0070) (8,0.0821) (9,-0.0320) (10,0.0509) (11,-0.0044) (12,-0.0579) (13,-0.0575) (14,0.0193) (15,0.0107) (16,0.0001) (17,0.0285) (18,-0.0162) (19,0.0174) (20,0.1502) (21,0.0322) (22,-0.0051) (23,-0.0039) (24,0.0527) (25,0.1142) (26,-0.0908) (27,-0.0363) (28,-0.1264) (29,0.0794) (30,0.0353) (31,-0.0186) (32,-0.0082) (33,0.1718) (34,0.0473) (35,0.0702) (36,-0.0868) (37,-0.0124) (38,-0.0453) (39,-0.1764) (40,-0.0259) (41,0.1418) (42,-0.0098) (43,-0.0004) (44,-0.0463) (45,-0.0241) (46,-0.0893) (47,-0.0295) (48,-0.1397)};
\addplot[gray!45,thin,forget plot] coordinates {(1,0.0326) (2,0.0286) (3,-0.0761) (4,-0.0692) (5,0.0221) (6,0.0601) (7,0.0905) (8,0.0392) (9,0.0393) (10,-0.1132) (11,-0.0553) (12,-0.1065) (13,-0.0554) (14,-0.1012) (15,-0.0243) (16,0.0127) (17,0.0418) (18,0.0824) (19,0.0176) (20,-0.0204) (21,0.0187) (22,-0.1467) (23,0.0186) (24,-0.1349) (25,0.0038) (26,0.1454) (27,-0.0522) (28,0.0232) (29,-0.0903) (30,-0.0108) (31,0.0837) (32,-0.0035) (33,0.0587) (34,0.1454) (35,-0.0391) (36,0.0511) (37,0.0662) (38,-0.0098) (39,0.1646) (40,0.1172) (41,-0.0239) (42,-0.0380) (43,0.0329) (44,0.0843) (45,0.0307) (46,-0.0296) (47,-0.0452) (48,-0.0936)};
\addplot[gray!45,thin,forget plot] coordinates {(1,0.1079) (2,-0.0431) (3,0.0752) (4,0.0792) (5,-0.0583) (6,0.0016) (7,-0.0164) (8,-0.0486) (9,0.0749) (10,-0.1820) (11,-0.0462) (12,-0.0409) (13,-0.0126) (14,0.0507) (15,-0.0543) (16,-0.0154) (17,0.0113) (18,-0.0426) (19,0.0043) (20,-0.0556) (21,-0.0920) (22,0.1115) (23,-0.1274) (24,-0.1429) (25,-0.1220) (26,-0.0292) (27,-0.1348) (28,-0.2813) (29,-0.0250) (30,0.0134) (31,-0.1130) (32,0.0888) (33,-0.0893) (34,0.0568) (35,-0.0361) (36,0.1412) (37,-0.0803) (38,-0.1755) (39,-0.0980) (40,-0.0362) (41,0.0798) (42,-0.0707) (43,-0.2075) (44,-0.0417) (45,0.0193) (46,-0.1135) (47,-0.1605) (48,-0.1171)};
\addplot[gray!45,thin,forget plot] coordinates {(1,-0.0180) (2,0.0166) (3,-0.1660) (4,-0.0700) (5,-0.0705) (6,0.0259) (7,0.0762) (8,0.0559) (9,-0.1435) (10,0.0396) (11,0.0915) (12,0.0049) (13,0.0481) (14,0.0653) (15,0.0031) (16,-0.0118) (17,-0.0457) (18,-0.0358) (19,-0.0030) (20,0.0712) (21,0.1682) (22,0.0318) (23,0.0753) (24,0.1450) (25,0.0430) (26,-0.1470) (27,0.0599) (28,-0.0445) (29,-0.0529) (30,-0.0230) (31,-0.0532) (32,-0.0213) (33,0.1176) (34,0.0903) (35,0.2505) (36,-0.0158) (37,-0.0823) (38,0.1359) (39,-0.0733) (40,0.0079) (41,0.0269) (42,0.0451) (43,-0.0858) (44,-0.1403) (45,0.1245) (46,0.0653) (47,0.1574) (48,0.0579)};
\addplot[gray!45,thin,forget plot] coordinates {(1,0.0070) (2,0.0455) (3,0.0416) (4,-0.0132) (5,-0.0225) (6,0.0393) (7,0.0887) (8,-0.1252) (9,0.0062) (10,0.0060) (11,-0.1289) (12,0.0106) (13,0.0577) (14,-0.0726) (15,-0.0375) (16,0.1057) (17,0.0190) (18,-0.0284) (19,0.0254) (20,0.0634) (21,-0.0147) (22,0.0802) (23,-0.0551) (24,0.0084) (25,-0.0768) (26,-0.0394) (27,0.0211) (28,0.2365) (29,0.1795) (30,0.1492) (31,0.0293) (32,0.0022) (33,0.0675) (34,0.0650) (35,0.0279) (36,0.1486) (37,-0.1365) (38,-0.0170) (39,0.0726) (40,0.1204) (41,0.0864) (42,-0.0094) (43,-0.0727) (44,0.1063) (45,0.1285) (46,0.1719) (47,-0.0141) (48,0.0261)};
\addplot[gray!45,thin,forget plot] coordinates {(1,0.0601) (2,-0.0325) (3,0.0820) (4,0.1723) (5,0.1176) (6,-0.0694) (7,-0.0935) (8,-0.1326) (9,-0.1415) (10,0.0872) (11,0.0722) (12,-0.0267) (13,-0.0967) (14,-0.1097) (15,0.0661) (16,0.0374) (17,0.0975) (18,-0.0270) (19,-0.0676) (20,-0.0215) (21,-0.0364) (22,-0.0537) (23,-0.1129) (24,0.0420) (25,0.0384) (26,0.0686) (27,-0.0750) (28,0.1100) (29,0.1697) (30,0.1847) (31,0.1381) (32,0.0358) (33,-0.0816) (34,-0.0507) (35,0.0048) (36,-0.1009) (37,0.0052) (38,0.0057) (39,-0.0396) (40,-0.0361) (41,0.0278) (42,0.0809) (43,-0.0232) (44,0.0236) (45,0.0849) (46,0.0251) (47,0.0092) (48,-0.0078)};
\addplot[gray!45,thin,forget plot] coordinates {(1,-0.0140) (2,0.0286) (3,-0.0398) (4,0.0071) (5,0.0093) (6,0.0324) (7,-0.1246) (8,0.0573) (9,0.0361) (10,0.0451) (11,0.0677) (12,0.0753) (13,0.0292) (14,-0.0739) (15,0.0541) (16,0.0061) (17,-0.0739) (18,-0.1021) (19,0.0456) (20,-0.0235) (21,0.1138) (22,0.0326) (23,0.1201) (24,-0.0357) (25,0.0987) (26,-0.1452) (27,0.0429) (28,-0.1183) (29,0.0286) (30,0.0228) (31,-0.0147) (32,0.0235) (33,-0.0882) (34,-0.0726) (35,0.0069) (36,0.0354) (37,0.0572) (38,-0.2636) (39,-0.0157) (40,0.0570) (41,0.1506) (42,0.1694) (43,0.1522) (44,-0.0181) (45,-0.1197) (46,-0.0290) (47,0.1116) (48,0.0238)};
\addplot[gray!45,thin,forget plot] coordinates {(1,-0.1485) (2,-0.0863) (3,-0.0041) (4,0.0735) (5,0.0407) (6,0.1868) (7,-0.0784) (8,0.0754) (9,0.0352) (10,-0.0111) (11,-0.1000) (12,0.0020) (13,0.0555) (14,0.0177) (15,-0.0023) (16,-0.0914) (17,-0.0879) (18,-0.0311) (19,0.0211) (20,0.0710) (21,-0.0821) (22,0.1025) (23,0.1731) (24,0.0686) (25,0.0670) (26,-0.0553) (27,0.1419) (28,0.0827) (29,0.0972) (30,0.0266) (31,-0.1830) (32,0.0342) (33,-0.0161) (34,0.1154) (35,-0.0618) (36,-0.2824) (37,-0.1244) (38,-0.1147) (39,0.0415) (40,-0.0450) (41,-0.0501) (42,-0.0499) (43,-0.1098) (44,-0.0344) (45,-0.0156) (46,-0.0245) (47,-0.0514) (48,-0.0153)};
\addplot[gray!45,thin,forget plot] coordinates {(1,0.0717) (2,0.0285) (3,-0.0363) (4,0.0099) (5,0.0270) (6,-0.0261) (7,0.0024) (8,0.0277) (9,0.0541) (10,0.0356) (11,-0.0731) (12,-0.0507) (13,-0.0931) (14,-0.0312) (15,0.0690) (16,0.0045) (17,0.0273) (18,0.0106) (19,-0.0145) (20,0.0837) (21,-0.0084) (22,-0.0456) (23,-0.0805) (24,0.0424) (25,-0.0774) (26,-0.1062) (27,0.0384) (28,0.0139) (29,-0.0075) (30,-0.0891) (31,-0.0483) (32,-0.0303) (33,0.0371) (34,0.0474) (35,-0.0162) (36,0.0125) (37,0.0518) (38,-0.1298) (39,-0.1247) (40,0.0092) (41,0.0157) (42,0.0801) (43,0.1522) (44,0.1615) (45,0.1373) (46,-0.0242) (47,0.1444) (48,-0.0283)};
\addplot[gray!45,thin,forget plot] coordinates {(1,-0.0401) (2,-0.2815) (3,-0.2267) (4,-0.1993) (5,-0.0956) (6,-0.1732) (7,-0.2494) (8,-0.1237) (9,-0.0747) (10,-0.0032) (11,-0.0849) (12,-0.0835) (13,-0.1392) (14,-0.1352) (15,-0.1795) (16,-0.0371) (17,-0.1496) (18,0.0101) (19,-0.1287) (20,-0.1882) (21,-0.0866) (22,0.0682) (23,-0.0421) (24,-0.2196) (25,-0.0755) (26,-0.1305) (27,-0.1728) (28,-0.2612) (29,0.0239) (30,-0.0879) (31,-0.0203) (32,0.0493) (33,-0.1505) (34,-0.1350) (35,-0.1080) (36,0.0829) (37,-0.0559) (38,-0.3413) (39,-0.1733) (40,-0.0998) (41,0.0027) (42,0.0555) (43,-0.1332) (44,-0.0768) (45,0.0596) (46,-0.1243) (47,-0.0801) (48,-0.2357)};
\addplot[gray!45,thin,forget plot] coordinates {(1,0.0268) (2,0.0587) (3,-0.0163) (4,-0.0710) (5,0.0073) (6,-0.0264) (7,0.0542) (8,-0.0490) (9,0.1114) (10,0.1216) (11,-0.0456) (12,0.0055) (13,0.0205) (14,-0.0182) (15,-0.0642) (16,-0.1145) (17,-0.0127) (18,-0.0211) (19,0.0089) (20,0.0058) (21,-0.0601) (22,0.0944) (23,-0.0906) (24,0.1326) (25,0.0048) (26,-0.0797) (27,-0.0047) (28,0.0790) (29,0.0643) (30,-0.1172) (31,0.0507) (32,0.0334) (33,-0.0602) (34,-0.1872) (35,0.0809) (36,0.1890) (37,0.0597) (38,0.0016) (39,0.0031) (40,-0.1056) (41,0.0072) (42,0.1090) (43,0.0637) (44,0.0338) (45,-0.0195) (46,0.0809) (47,0.0336) (48,-0.0348)};
\addplot[gray!45,thin,forget plot] coordinates {(1,0.2660) (2,0.3089) (3,0.4146) (4,0.2829) (5,0.3146) (6,0.2165) (7,0.1896) (8,0.1053) (9,0.0358) (10,0.1370) (11,0.1956) (12,0.2033) (13,0.2490) (14,0.1042) (15,0.2346) (16,0.3118) (17,0.1902) (18,0.2970) (19,0.2463) (20,0.1798) (21,0.3515) (22,0.4369) (23,0.4436) (24,0.2673) (25,0.3570) (26,0.2238) (27,0.4982) (28,0.2255) (29,0.1586) (30,0.1304) (31,0.2260) (32,0.0650) (33,0.2069) (34,-0.0138) (35,0.4303) (36,0.2988) (37,0.1726) (38,0.1328) (39,0.1285) (40,0.1216) (41,0.0608) (42,0.1903) (43,0.3463) (44,0.1110) (45,-0.1655) (46,0.1454) (47,0.3517) (48,0.3461)};
\addplot[gray!45,thin,forget plot] coordinates {(1,-0.0328) (2,-0.1105) (3,-0.0662) (4,0.0194) (5,0.0452) (6,-0.0042) (7,-0.0124) (8,-0.1124) (9,0.0680) (10,-0.0399) (11,-0.0719) (12,-0.0372) (13,0.0148) (14,0.0426) (15,-0.0290) (16,0.0575) (17,0.0715) (18,0.0327) (19,0.0932) (20,0.0037) (21,0.0241) (22,0.1435) (23,0.0419) (24,0.0233) (25,-0.0732) (26,-0.1676) (27,0.0433) (28,-0.0045) (29,-0.0060) (30,-0.0381) (31,-0.0209) (32,-0.0083) (33,0.1030) (34,-0.0528) (35,-0.0560) (36,-0.0646) (37,-0.0897) (38,0.0258) (39,-0.0225) (40,0.0129) (41,0.1090) (42,0.2043) (43,0.1525) (44,-0.0495) (45,-0.1122) (46,-0.2446) (47,-0.1179) (48,0.0014)};
\addplot[gray!45,thin,forget plot] coordinates {(1,-0.0272) (2,0.1755) (3,0.0200) (4,0.0063) (5,-0.1996) (6,0.0682) (7,-0.0005) (8,-0.1221) (9,0.0106) (10,0.1621) (11,0.1574) (12,0.0011) (13,0.0270) (14,-0.0114) (15,-0.0280) (16,-0.0796) (17,0.1014) (18,-0.0473) (19,-0.0808) (20,0.0115) (21,-0.0619) (22,0.0401) (23,0.1570) (24,0.0668) (25,0.0060) (26,0.0002) (27,0.0424) (28,0.0391) (29,-0.0887) (30,-0.0514) (31,-0.1098) (32,-0.2162) (33,-0.0666) (34,0.0396) (35,0.1688) (36,0.1480) (37,0.0486) (38,0.0892) (39,0.1555) (40,0.1835) (41,0.0952) (42,0.0777) (43,-0.1208) (44,0.0949) (45,0.0745) (46,-0.0267) (47,0.2410) (48,-0.0248)};
\addplot[gray!45,thin,forget plot] coordinates {(1,-0.0266) (2,0.0332) (3,-0.0403) (4,0.0485) (5,0.0446) (6,-0.0615) (7,0.0553) (8,0.0337) (9,-0.0122) (10,0.0540) (11,0.0077) (12,0.0439) (13,-0.0488) (14,-0.0028) (15,-0.0972) (16,0.0260) (17,-0.0268) (18,-0.0241) (19,0.0289) (20,0.0398) (21,0.0478) (22,-0.0428) (23,0.0089) (24,0.0762) (25,-0.1325) (26,-0.0217) (27,0.0057) (28,0.0297) (29,0.0809) (30,0.0806) (31,0.0195) (32,0.0705) (33,0.1244) (34,0.2343) (35,0.1405) (36,0.1745) (37,0.1262) (38,-0.0206) (39,0.0877) (40,0.1597) (41,0.0876) (42,-0.1880) (43,0.0543) (44,0.0434) (45,0.0980) (46,0.0411) (47,-0.0492) (48,-0.0202)};
\addplot[gray!45,thin,forget plot] coordinates {(1,-0.0023) (2,-0.0309) (3,-0.0104) (4,0.0248) (5,0.0273) (6,0.0279) (7,0.0037) (8,-0.0084) (9,0.0152) (10,0.0627) (11,-0.0290) (12,-0.0211) (13,-0.0424) (14,0.0482) (15,0.0088) (16,0.0095) (17,-0.0723) (18,-0.0069) (19,0.0244) (20,-0.1161) (21,-0.0052) (22,-0.2180) (23,0.2002) (24,0.0504) (25,0.0286) (26,0.1177) (27,0.0310) (28,0.0493) (29,0.0155) (30,0.1086) (31,-0.0435) (32,0.0338) (33,-0.0743) (34,0.0314) (35,0.0493) (36,-0.0599) (37,0.1578) (38,-0.0448) (39,-0.0220) (40,0.1020) (41,-0.0958) (42,-0.1158) (43,-0.1514) (44,-0.0435) (45,-0.1208) (46,0.2240) (47,0.0585) (48,0.0357)};
\addplot[gray!45,thin,forget plot] coordinates {(1,0.0568) (2,-0.0435) (3,-0.0151) (4,0.0152) (5,-0.0123) (6,-0.0186) (7,0.0541) (8,0.0080) (9,-0.0132) (10,-0.0376) (11,-0.0374) (12,-0.0276) (13,-0.0312) (14,0.0557) (15,0.0397) (16,-0.0137) (17,-0.0324) (18,-0.0284) (19,0.0416) (20,0.0607) (21,0.0028) (22,-0.0038) (23,-0.0139) (24,-0.0744) (25,-0.0524) (26,0.1482) (27,-0.0156) (28,-0.0005) (29,-0.0102) (30,0.0725) (31,0.0313) (32,-0.0770) (33,-0.1080) (34,0.0209) (35,-0.0383) (36,-0.0395) (37,-0.1396) (38,-0.0463) (39,-0.0016) (40,0.0440) (41,0.0295) (42,-0.0209) (43,-0.0763) (44,-0.0814) (45,0.0236) (46,0.0188) (47,-0.1156) (48,0.0921)};
\addplot[gray!45,thin,forget plot] coordinates {(1,0.0387) (2,-0.0406) (3,-0.1869) (4,-0.0970) (5,0.0762) (6,-0.0329) (7,0.0482) (8,-0.0069) (9,0.0219) (10,0.0087) (11,0.0218) (12,0.0272) (13,0.0168) (14,-0.0408) (15,0.0207) (16,0.1105) (17,0.0719) (18,0.1219) (19,-0.0447) (20,-0.1630) (21,0.0114) (22,-0.0571) (23,-0.2663) (24,-0.0624) (25,0.0319) (26,0.0877) (27,0.0290) (28,-0.0779) (29,0.1707) (30,0.2474) (31,0.2869) (32,0.2482) (33,0.1002) (34,0.1583) (35,-0.0584) (36,-0.0586) (37,0.0531) (38,-0.0425) (39,-0.0566) (40,0.0293) (41,0.0143) (42,-0.0164) (43,0.1367) (44,-0.0367) (45,0.0156) (46,-0.0088) (47,-0.1525) (48,-0.0108)};
\addplot[gray!45,thin,forget plot] coordinates {(1,-0.1429) (2,-0.1214) (3,-0.0096) (4,-0.1769) (5,-0.0020) (6,0.0448) (7,0.0027) (8,0.0187) (9,-0.0314) (10,-0.0128) (11,0.0308) (12,0.1306) (13,0.0163) (14,0.0676) (15,0.0176) (16,0.0302) (17,0.0378) (18,-0.0156) (19,0.1094) (20,0.0257) (21,-0.1379) (22,0.1828) (23,-0.2149) (24,-0.0335) (25,-0.0875) (26,-0.0845) (27,-0.2024) (28,-0.1084) (29,-0.1628) (30,-0.1732) (31,-0.0635) (32,-0.1861) (33,0.0296) (34,0.0391) (35,-0.1996) (36,-0.0500) (37,-0.1849) (38,0.2204) (39,0.0182) (40,-0.1754) (41,-0.0471) (42,-0.1174) (43,-0.1398) (44,-0.0136) (45,0.0013) (46,-0.1212) (47,0.0062) (48,0.1189)};
\addplot[gray!45,thin,forget plot] coordinates {(1,0.0065) (2,-0.0617) (3,-0.0975) (4,0.0052) (5,-0.0520) (6,0.0300) (7,0.0300) (8,0.1068) (9,0.1280) (10,-0.0003) (11,0.0020) (12,0.0205) (13,0.0129) (14,0.1126) (15,-0.0540) (16,-0.1110) (17,0.0519) (18,-0.0230) (19,-0.0385) (20,0.0518) (21,-0.0785) (22,-0.1708) (23,-0.1505) (24,-0.0062) (25,-0.0585) (26,0.0312) (27,-0.2059) (28,0.0332) (29,0.0352) (30,0.0992) (31,0.0030) (32,0.1392) (33,0.0291) (34,0.1797) (35,-0.1554) (36,-0.0768) (37,0.0479) (38,0.1353) (39,0.1027) (40,0.0703) (41,0.0709) (42,0.0181) (43,-0.1244) (44,0.0655) (45,0.3070) (46,0.0630) (47,-0.1615) (48,-0.1140)};
\addplot[gray!45,thin,forget plot] coordinates {(1,-0.0431) (2,-0.0561) (3,0.0945) (4,0.0514) (5,-0.0884) (6,0.0469) (7,0.0991) (8,-0.0631) (9,-0.0329) (10,-0.0032) (11,0.0637) (12,-0.0168) (13,0.0073) (14,0.0013) (15,-0.0672) (16,-0.0681) (17,-0.0109) (18,0.0045) (19,0.0844) (20,0.0170) (21,-0.0188) (22,0.1162) (23,0.0750) (24,0.0130) (25,-0.0051) (26,-0.1076) (27,0.0293) (28,-0.1342) (29,-0.1275) (30,-0.1090) (31,-0.1769) (32,-0.0113) (33,-0.0807) (34,-0.2406) (35,-0.1220) (36,0.0130) (37,-0.0261) (38,0.0180) (39,-0.2401) (40,-0.1469) (41,0.0085) (42,-0.0019) (43,0.0292) (44,-0.0349) (45,-0.0822) (46,0.0655) (47,0.1304) (48,0.1393)};
\addplot[gray!45,thin,forget plot] coordinates {(1,0.0303) (2,-0.0525) (3,0.0294) (4,0.0520) (5,-0.1004) (6,-0.0311) (7,0.0136) (8,0.0800) (9,-0.0309) (10,0.0075) (11,-0.0050) (12,0.0071) (13,-0.0111) (14,-0.0796) (15,-0.0128) (16,0.0742) (17,0.0424) (18,0.0320) (19,-0.0298) (20,0.1744) (21,0.1151) (22,0.0203) (23,0.1123) (24,0.1966) (25,0.0041) (26,-0.0610) (27,0.0228) (28,0.0200) (29,0.0240) (30,-0.0286) (31,-0.1981) (32,-0.0326) (33,0.0547) (34,0.0058) (35,0.1243) (36,-0.1317) (37,-0.1191) (38,0.0166) (39,0.0359) (40,0.0931) (41,0.0834) (42,0.1644) (43,-0.0491) (44,0.0230) (45,-0.0278) (46,0.0924) (47,0.0318) (48,0.0200)};
\addplot[gray!45,thin,forget plot] coordinates {(1,0.0376) (2,-0.0394) (3,0.0162) (4,-0.0620) (5,-0.0269) (6,0.0171) (7,0.0492) (8,-0.0054) (9,-0.0635) (10,0.0336) (11,-0.0290) (12,-0.0492) (13,0.1196) (14,-0.0207) (15,0.1995) (16,-0.0223) (17,-0.0439) (18,-0.0349) (19,-0.0416) (20,0.0483) (21,-0.0054) (22,0.0297) (23,0.1825) (24,0.0545) (25,0.0586) (26,-0.0658) (27,0.1195) (28,0.1940) (29,0.0732) (30,-0.0946) (31,-0.0615) (32,0.1860) (33,0.1002) (34,0.1959) (35,0.0411) (36,0.0355) (37,0.1535) (38,0.0655) (39,0.0071) (40,0.1826) (41,0.0848) (42,0.1012) (43,0.0175) (44,-0.0637) (45,-0.0060) (46,0.1271) (47,0.1048) (48,0.1271)};
\addplot[gray!45,thin,forget plot] coordinates {(1,-0.0546) (2,0.0273) (3,0.0403) (4,0.0413) (5,-0.0221) (6,-0.1595) (7,0.0027) (8,-0.0080) (9,0.0383) (10,-0.0172) (11,0.0567) (12,-0.0953) (13,-0.0276) (14,0.0653) (15,0.0457) (16,0.1437) (17,0.0140) (18,0.0082) (19,-0.0334) (20,0.0781) (21,-0.0886) (22,0.0393) (23,-0.1415) (24,-0.0355) (25,-0.1531) (26,0.0752) (27,-0.1402) (28,0.0013) (29,-0.0786) (30,-0.1071) (31,0.0344) (32,0.0278) (33,-0.0118) (34,-0.1562) (35,-0.0344) (36,0.0351) (37,-0.1659) (38,0.1352) (39,0.0826) (40,-0.1169) (41,0.0352) (42,-0.1023) (43,0.0766) (44,0.1317) (45,0.2407) (46,-0.0940) (47,-0.0753) (48,-0.1385)};
\addplot[gray!45,thin,forget plot] coordinates {(1,-0.0722) (2,0.1037) (3,-0.0047) (4,0.0076) (5,0.0171) (6,-0.0025) (7,0.0137) (8,-0.0050) (9,-0.0417) (10,0.0083) (11,-0.0133) (12,0.0183) (13,0.0125) (14,-0.0542) (15,0.0100) (16,-0.0028) (17,0.0064) (18,-0.0118) (19,-0.0277) (20,-0.0034) (21,0.0248) (22,-0.1300) (23,0.0040) (24,0.0899) (25,0.0218) (26,0.0270) (27,0.0894) (28,0.2131) (29,-0.0604) (30,-0.0244) (31,0.0290) (32,-0.1156) (33,0.0035) (34,-0.0785) (35,-0.0329) (36,-0.1297) (37,0.0247) (38,0.1842) (39,0.0942) (40,-0.0107) (41,-0.1360) (42,-0.0436) (43,0.1227) (44,0.0213) (45,-0.1024) (46,0.0392) (47,0.0445) (48,0.1134)};
\addplot[gray!45,thin,forget plot] coordinates {(1,-0.0353) (2,0.0198) (3,0.0143) (4,-0.0824) (5,-0.0942) (6,0.0431) (7,0.0270) (8,0.0003) (9,0.0105) (10,-0.0683) (11,-0.0664) (12,0.0226) (13,0.0355) (14,0.0599) (15,0.0255) (16,0.0157) (17,-0.0260) (18,0.1005) (19,-0.0539) (20,0.0262) (21,-0.0123) (22,0.0901) (23,0.0420) (24,0.0994) (25,-0.0090) (26,0.1329) (27,0.1342) (28,0.2128) (29,0.0838) (30,0.1003) (31,0.1065) (32,0.1476) (33,0.0966) (34,-0.0001) (35,0.1296) (36,-0.1826) (37,-0.0562) (38,-0.1616) (39,0.0866) (40,0.0164) (41,0.0809) (42,0.1166) (43,0.1815) (44,0.1021) (45,0.0305) (46,0.0698) (47,0.0359) (48,-0.0401)};
\addplot[gray!45,thin,forget plot] coordinates {(1,0.0371) (2,0.0983) (3,0.0326) (4,-0.0745) (5,0.0915) (6,0.0093) (7,-0.1604) (8,-0.1231) (9,-0.0231) (10,-0.0081) (11,0.0026) (12,0.1052) (13,0.0333) (14,0.0505) (15,0.0397) (16,-0.1698) (17,-0.0106) (18,-0.0125) (19,0.0496) (20,-0.0941) (21,-0.0011) (22,0.0602) (23,0.1863) (24,0.1047) (25,-0.0472) (26,0.0630) (27,-0.1254) (28,-0.1073) (29,0.0260) (30,0.0319) (31,0.0710) (32,0.1945) (33,0.2039) (34,0.1984) (35,-0.0692) (36,-0.0673) (37,0.0317) (38,-0.1219) (39,-0.0776) (40,-0.0649) (41,-0.0098) (42,0.0871) (43,0.1616) (44,0.0489) (45,0.0946) (46,0.0192) (47,0.1898) (48,0.1005)};
\addplot[gray!45,thin,forget plot] coordinates {(1,0.0305) (2,-0.0543) (3,-0.0930) (4,-0.0164) (5,-0.0075) (6,0.0832) (7,-0.0261) (8,-0.0399) (9,-0.0398) (10,0.0376) (11,-0.0824) (12,0.0369) (13,0.0008) (14,0.0147) (15,0.0879) (16,-0.0128) (17,0.1355) (18,0.0004) (19,-0.1304) (20,-0.0292) (21,-0.0892) (22,-0.1241) (23,-0.0971) (24,-0.1445) (25,-0.0066) (26,0.0555) (27,0.0268) (28,0.0239) (29,0.0188) (30,-0.1594) (31,-0.1137) (32,-0.2006) (33,-0.0606) (34,-0.0062) (35,-0.0643) (36,0.0364) (37,0.0219) (38,-0.1868) (39,-0.0170) (40,0.1084) (41,0.0143) (42,-0.1064) (43,0.0379) (44,-0.0136) (45,-0.0273) (46,0.1460) (47,0.0293) (48,-0.0332)};
\addplot[gray!45,thin,forget plot] coordinates {(1,-0.0357) (2,0.0726) (3,-0.0358) (4,-0.0008) (5,-0.0353) (6,0.0473) (7,-0.0092) (8,0.0554) (9,-0.0537) (10,-0.0247) (11,-0.0118) (12,-0.0032) (13,-0.0418) (14,0.0443) (15,0.0416) (16,0.0437) (17,-0.0008) (18,-0.0138) (19,-0.0074) (20,0.1187) (21,0.0985) (22,-0.0785) (23,-0.0586) (24,0.0302) (25,0.0998) (26,0.1474) (27,-0.0875) (28,-0.0799) (29,-0.0376) (30,-0.0685) (31,-0.0505) (32,-0.0612) (33,-0.0943) (34,-0.0311) (35,-0.2059) (36,0.0592) (37,0.2307) (38,-0.0682) (39,-0.0411) (40,-0.0380) (41,-0.0179) (42,0.0690) (43,-0.0255) (44,0.0451) (45,0.0361) (46,0.0039) (47,0.1566) (48,0.1195)};
\addplot[blue!70!black,very thick] coordinates {(1,-0.0017) (2,-0.0018) (3,0.0020) (4,-0.0007) (5,0.0021) (6,-0.0038) (7,0.0102) (8,0.0021) (9,0.0037) (10,-0.0034) (11,-0.0008) (12,0.0066) (13,-0.0008) (14,-0.0089) (15,0.0016) (16,-0.0020) (17,0.0092) (18,-0.0082) (19,-0.0054) (20,0.0515) (21,0.1422) (22,0.2385) (23,0.2922) (24,0.2738) (25,0.3122) (26,0.3012) (27,0.3055) (28,0.2547) (29,0.3014) (30,0.2729) (31,0.3054) (32,0.3377) (33,0.3128) (34,0.2884) (35,0.3054) (36,0.3215) (37,0.3065) (38,0.2938) (39,0.2998) (40,0.3072) (41,0.3102) (42,0.3319) (43,0.3398) (44,0.2610) (45,0.3029) (46,0.3211) (47,0.4009) (48,0.3843)};
\addplot[red!70!black,thick,dashed] coordinates {(1,0.0000) (2,0.0000) (3,0.0000) (4,0.0000) (5,0.0000) (6,0.0000) (7,0.0000) (8,0.0000) (9,0.0000) (10,0.0000) (11,0.0000) (12,0.0000) (13,0.0000) (14,0.0000) (15,0.0000) (16,0.0000) (17,0.0000) (18,0.0000) (19,0.0000) (20,0.0850) (21,0.1460) (22,0.1896) (23,0.2209) (24,0.2433) (25,0.2594) (26,0.2709) (27,0.2792) (28,0.2851) (29,0.2893) (30,0.2923) (31,0.2945) (32,0.2961) (33,0.2972) (34,0.2980) (35,0.2986) (36,0.2990) (37,0.2993) (38,0.2995) (39,0.2996) (40,0.2997) (41,0.2998) (42,0.2999) (43,0.2999) (44,0.2999) (45,0.2999) (46,0.3000) (47,0.3000) (48,0.3000)};
\draw[black,dotted] (axis cs:19.5,-0.3) -- (axis cs:19.5,0.45);
\legend{placebo gaps ($J=30$), treated cohort gap, true effect path}
\end{axis}
\end{tikzpicture}
\end{minipage}
\caption{Left: Callaway--Sant'Anna event study with never-treated controls (cluster bootstrap, $399$ draws). Right: synthetic control for the first cohort, with in-space placebo gaps (RMSPE ratio $\rSynthRatio$, permutation $p=\rSynthP$, $\rSynthDonors$ donors with positive weight).}
\label{fig:econ}
\end{figure}

Over $\rEconR$ replications the true simple ATT is $\rEconTruth$\,pp. Static TWFE is biased by $\rEconTwfeBias$\,pp ($\rEconTwfeRel\%$; RMSE $\rEconTwfeRmse$). The bias arises from the forbidden comparisons of already-treated units with dynamic effects \citep{goodmanbacon2021,dechaisemartin2020}. The Callaway--Sant'Anna estimator is unbiased to simulation precision (bias $\rEconCsBias$, RMSE $\rEconCsRmse$). Its cluster-bootstrap 95\% interval covers the truth in $\rEconCsCov$ of $\rEconCovReps$ replications, which is slightly conservative. \Cref{fig:econ} shows one replication. Pre-trends are flat, and the event-time profile tracks the ramp. The synthetic control isolates the first cohort with a post/pre RMSPE ratio of $\rSynthRatio$. Its permutation $p$-value $\rSynthP=1/31$ is the smallest attainable with 30 placebos. An estimated effect maps into the hazard of \cref{ass:hazard} as $\beta^{\mathrm{deg}}=-12\ln(1-\widehat{\mathrm{ATT}})/\bar x$, where $\bar x$ is the treated cohort's capability loss. The simulated ATT gives $\rEconHazard/\bar x$ per year, which is of the same order as the priors in \cref{tab:calib-bank}. TWFE would understate it by about a sixth.

\subsection{RQ8b: Monte Carlo precision}
At $n=\rN$, the 95\% confidence half-width is $\pm\rBankMeanHW\%$ of the mean for Bank~A, $\pm\rTelMeanHW\%$ for Operator~T and $\pm\rAirMeanHW\%$ for Carrier~P. For $P_{95}$ it is $\pm\rBankPnfHW\%$, $\pm\rTelPnfHW\%$ and $\pm\rAirPnfHW\%$ (distribution-free). \Cref{fig:conv} shows the running estimates for Bank~A, the most skewed case. Both the mean and the $P_{95}$ stabilise within their confidence bands after about $2\times10^4$ draws. The sampling error is therefore negligible relative to the epistemic spread, and the reported $B^\star$ is stable to the first decimal in USD millions.

\begin{figure}[t]
\centering
\begin{tikzpicture}
\begin{groupplot}[group style={group size=2 by 1,horizontal sep=1.6cm},width=0.47\textwidth,height=4.6cm,xmode=log,xlabel={Draws $n$},xmin=100,xmax=200000]
\nextgroupplot[title={Bank A: running mean},ylabel={USD m}]
\addplot[name path=mlo,draw=none,forget plot] coordinates {(100,417.66) (121,426.36) (147,429.39) (179,441.27) (218,438.82) (264,452.16) (321,452.84) (391,454.03) (475,447.54) (577,448.63) (702,448.38) (853,450.97) (1036,452.42) (1259,455.41) (1531,454.29) (1860,457.21) (2260,456.44) (2747,456.99) (3338,457.47) (4056,457.8) (4929,457.72) (5990,458.53) (7279,458.63) (8846,461.44) (10749,461.08) (13063,461) (15874,461.68) (19289,463.3) (23440,463.61) (28484,463.63) (34614,463.32) (42062,463.76) (51113,463.62) (62112,463.61) (75478,463.35) (91719,463.02) (1.1146e+05,463.17) (1.3544e+05,463.36) (1.6458e+05,463.25) (2e+05,462.84)};
\addplot[name path=mhi,draw=none,forget plot] coordinates {(100,492.11) (121,493.74) (147,489.06) (179,497.53) (218,489.79) (264,500.8) (321,497.94) (391,494.71) (475,483.1) (577,481.88) (702,478.12) (853,478.79) (1036,479.4) (1259,480.24) (1531,476.71) (1860,477.91) (2260,474.71) (2747,473.38) (3338,472.22) (4056,471.16) (4929,470) (5990,469.82) (7279,468.71) (8846,470.74) (10749,469.63) (13063,468.66) (15874,468.58) (19289,469.99) (23440,470.12) (28484,469.41) (34614,468.48) (42062,468.43) (51113,467.81) (62112,467.37) (75478,466.71) (91719,466.09) (1.1146e+05,465.92) (1.3544e+05,465.83) (1.6458e+05,465.48) (2e+05,464.84)};
\addplot[blue!15,forget plot] fill between[of=mlo and mhi];
\addplot[blue!70!black,thick] coordinates {(100,454.88) (121,460.05) (147,459.23) (179,469.4) (218,464.31) (264,476.48) (321,475.39) (391,474.37) (475,465.32) (577,465.26) (702,463.25) (853,464.88) (1036,465.91) (1259,467.82) (1531,465.5) (1860,467.56) (2260,465.58) (2747,465.19) (3338,464.84) (4056,464.48) (4929,463.86) (5990,464.17) (7279,463.67) (8846,466.09) (10749,465.36) (13063,464.83) (15874,465.13) (19289,466.65) (23440,466.87) (28484,466.52) (34614,465.9) (42062,466.09) (51113,465.71) (62112,465.49) (75478,465.03) (91719,464.55) (1.1146e+05,464.54) (1.3544e+05,464.59) (1.6458e+05,464.36) (2e+05,463.84)};
\nextgroupplot[title={Bank A: running $P_{95}$}]
\addplot[name path=qlo,draw=none,forget plot] coordinates {(100,653.33) (121,691.58) (147,709.16) (179,737.79) (218,744.99) (264,774.06) (321,807.12) (391,807.12) (475,782.72) (577,807.12) (702,808.87) (853,812.98) (1036,811.33) (1259,820.68) (1531,824.64) (1860,834.35) (2260,829.99) (2747,834.32) (3338,835.64) (4056,840.34) (4929,841.69) (5990,841.81) (7279,847.34) (8846,850.54) (10749,850.18) (13063,854.17) (15874,854.34) (19289,858.7) (23440,858.79) (28484,858.02) (34614,860.91) (42062,863.23) (51113,863.65) (62112,862.5) (75478,863.65) (91719,863.32) (1.1146e+05,864.87) (1.3544e+05,865.22) (1.6458e+05,865.96) (2e+05,866.07)};
\addplot[name path=qhi,draw=none,forget plot] coordinates {(100,1297.7) (121,1248.1) (147,918.06) (179,1023.1) (218,1023.1) (264,1015.1) (321,1015.1) (391,1012.4) (475,942.32) (577,969.07) (702,939.76) (853,939.76) (1036,924.95) (1259,933.6) (1531,918.06) (1860,920.55) (2260,894.34) (2747,892.07) (3338,890.25) (4056,889.45) (4929,887.33) (5990,885.2) (7279,883.18) (8846,884.75) (10749,882) (13063,883.18) (15874,881.39) (19289,882.31) (23440,881.2) (28484,878.98) (34614,880.51) (42062,882) (51113,880.39) (62112,877.57) (75478,877.51) (91719,874.92) (1.1146e+05,875.78) (1.3544e+05,874.59) (1.6458e+05,874.54) (2e+05,874.01)};
\addplot[red!15,forget plot] fill between[of=qlo and qhi];
\addplot[red!70!black,thick] coordinates {(100,735.73) (121,774.06) (147,765.34) (179,788.05) (218,783.4) (264,828.05) (321,850.54) (391,848.97) (475,828.03) (577,847.78) (702,847.11) (853,848.91) (1036,847.71) (1259,851.49) (1531,848.98) (1860,860.15) (2260,848.55) (2747,850.32) (3338,856.18) (4056,860.62) (4929,863.84) (5990,859.8) (7279,863.95) (8846,867.17) (10749,865.59) (13063,868.62) (15874,868.01) (19289,870.5) (23440,869.97) (28484,867.97) (34614,871.25) (42062,872.38) (51113,871.98) (62112,869.73) (75478,870.68) (91719,869.09) (1.1146e+05,870.59) (1.3544e+05,870.07) (1.6458e+05,870.44) (2e+05,870.16)};
\end{groupplot}
\end{tikzpicture}
\caption{Monte Carlo convergence for Bank~A, with 95\% CI bands (normal for the mean, order-statistic for $P_{95}$).}
\label{fig:conv}
\end{figure}
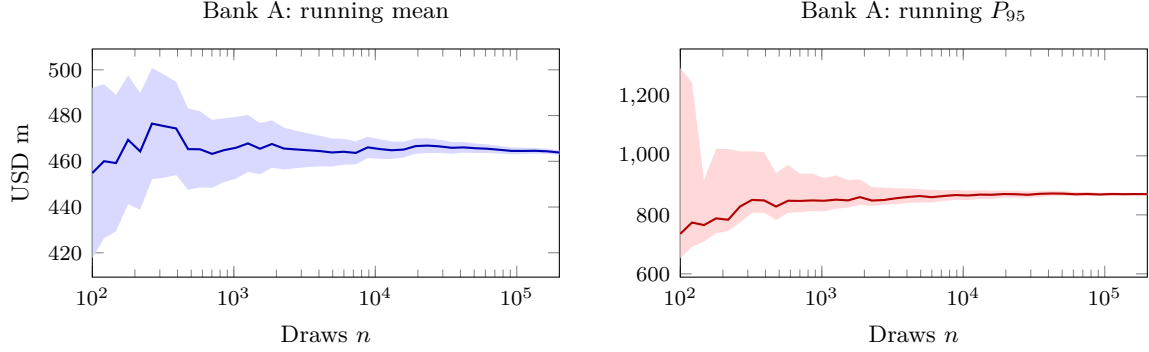

\section{Discussion}\label{sec:discussion}
\subsection{What the model changes}
The experiments support four conclusions that do not depend on the particular calibration.

First, \emph{downtime is not the main loss}. Across three platforms with very different economics, acute and incident outages account for only $\rTelAcIncMeanShare$--$\rBankAcIncMeanShare\%$ of expected loss. Most of the exposure lies in prolonged degraded operation (functional loss, workaround operations), in the customer response to it, and in blocked launches. A model that prices only outage hours, however carefully, misses most of the quantity of interest.

Second, \emph{point estimates are biased in a predictable direction}. Three mechanisms bias deterministic estimates downwards: convexity of outage costs (\cref{thm:jensen}), the max-structure of project networks (\cref{prop:fulkerson}), and convexity of the hazard-to-attrition map. In our cases the downward bias is $\rAirPlugGap$--$\rBankPlugGap\%$ of expected loss. The bias in the tail is larger, because the plug-in has no tail.

Third, \emph{the plug-in can choose the wrong response}. A deterministic analysis treats the remediation time as a known number and the restraint as a known interval. It therefore cannot represent the value of a remediation that completes early with positive probability, and it can recommend waiting when building is optimal (Carrier~P). It cannot see tail risk either, so it chooses the credit-control failure policy with the lower mode but the heavier tail (Operator~T). A loss estimate is only as defensible as the mitigation plan behind it. Because the model optimises over responses, the reported loss is a \emph{mitigated} loss.

Fourth, \emph{the right security is a quantile, not a mean}. Whenever $F_L(\E L)<\alpha^\star$ (right skew), posting the expected loss is sub-optimal in the precise sense of \cref{thm:bond}(c). Here $F_L(\E L)\approx0.54$--$0.60$, against $\alpha^\star\approx0.86$ at central assumptions. The critical-fractile rule is transparent: it follows from three inputs (the carrying cost, the probability that the restraint proves wrongful, and the weight on uncompensated loss) and the loss distribution.

\subsection{Robustness of the security rule to distributional ambiguity}\label{sec:robust}
A natural objection is that the loss distribution is itself model-dependent. Suppose only the mean and variance of $L$ are trusted. The classical distribution-free newsvendor argument \citep{scarf1958,gallego1993} then gives a security that depends on these two moments alone.

\begin{proposition}[Moment-robust security]\label{prop:robust}
Let $\mathcal{F}(\mu,\sigma)$ be the set of distributions on $\R$ with mean $\mu$ and standard deviation $\sigma>0$. Then for every $B\in\R$,
\begin{equation}\label{eq:scarf}
\sup_{F\in\mathcal{F}(\mu,\sigma)}\E_F(L-B)^+=\tfrac12\Big(\sqrt{\sigma^2+(B-\mu)^2}-(B-\mu)\Big).
\end{equation}
If $0<c<qw$ and $\alpha^\star=1-c/(qw)$, the minimiser of the worst-case cost $\sup_{F\in\mathcal{F}(\mu,\sigma)}\big[cB+qw\,\E_F(L-B)^+\big]$ is
\begin{equation}\label{eq:Brob}
B^{\mathrm{rob}}=\mu+\frac{\sigma}{2}\left(\sqrt{\frac{\alpha^\star}{1-\alpha^\star}}-\sqrt{\frac{1-\alpha^\star}{\alpha^\star}}\right)=\mu+\sigma\,\frac{2\alpha^\star-1}{2\sqrt{\alpha^\star(1-\alpha^\star)}} .
\end{equation}
In particular $B^{\mathrm{rob}}\ge\mu$ if and only if $\alpha^\star\ge\frac12$.
\end{proposition}
\begin{proof}
\emph{Upper bound.} Write $(L-B)^+=\frac12\big(|L-B|+(L-B)\big)$. By Jensen's inequality, $\E|L-B|\le\big(\E(L-B)^2\big)^{1/2}=\big(\sigma^2+(\mu-B)^2\big)^{1/2}$, so $\E(L-B)^+\le\frac12\big(\sqrt{\sigma^2+(B-\mu)^2}+\mu-B\big)$.

\emph{Attainment.} Let $d=\sqrt{\sigma^2+(B-\mu)^2}$. Put mass $p$ at $x_1=B-d$ and mass $1-p$ at $x_2=B+d$, with $p=\frac12\big(1-(\mu-B)/d\big)\in(0,1)$. The mean is $B+d(1-2p)=\mu$, and the second moment about $B$ is $d^2=\sigma^2+(\mu-B)^2$, so the variance is $\sigma^2$. For this distribution, $|L-B|=d$ almost surely, so Jensen holds with equality. Hence $\E(L-B)^+=\frac12(d+\mu-B)$.

\emph{Minimisation.} Since $B\mapsto\sqrt{\sigma^2+(B-\mu)^2}$ is a Euclidean norm of an affine map, it is convex. So the worst-case cost $g(B)=cB+\frac{qw}{2}\big(\sqrt{\sigma^2+(B-\mu)^2}-(B-\mu)\big)$ is convex and differentiable. Setting $g'(B)=0$ gives $(B-\mu)/\sqrt{\sigma^2+(B-\mu)^2}=1-2c/(qw)=2\alpha^\star-1=:k\in(-1,1)$. Squaring and solving, $(B-\mu)^2(1-k^2)=k^2\sigma^2$, with the sign of $B-\mu$ equal to that of $k$. Since $1-k^2=4\alpha^\star(1-\alpha^\star)$, this gives \eqref{eq:Brob}. The sign statement follows from that of $k$.
\end{proof}

\begin{remark}
Restricting $\mathcal{F}$ to distributions on $[0,\infty)$ can only lower the supremum in \eqref{eq:scarf}, so $B^{\mathrm{rob}}$ is conservative for non-negative losses. The rule $B^{\mathrm{rob}}/\mu=1+\mathrm{CV}\cdot\kappa(\alpha^\star)$, with $\kappa(\alpha)=(2\alpha-1)/(2\sqrt{\alpha(1-\alpha)})$, needs only the coefficient of variation. At $\alpha^\star=\rBankBondAlpha$, $\kappa=1.01$.
\end{remark}

Applied to the simulated moments, \eqref{eq:Brob} gives $B^{\mathrm{rob}}=\rBankBrob$ for Bank~A ($\sigma=\rBankSD$; $\rBankBrobRatio\times\E L$), $\rTelBrob$ for Operator~T ($\rTelBrobRatio\times$) and $\rAirBrob$ for Carrier~P ($\rAirBrobRatio\times$). The critical-fractile values are $\rBankBondB$, $\rTelBondB$ and $\rAirBondB$. For Operator~T and Carrier~P the two rules agree to within $0.2\%$. For Bank~A, the moment-robust rule is about $5\%$ higher, because its loss distribution has a heavier tail than the two-moment summary conveys. Evaluated under the simulated distribution, the expected cost of using $B^{\mathrm{rob}}$ instead of $B^\star$ is at most $\rBankRobPremium\%$ higher (Bank~A: $\rBankJtrueRob$ vs.\ $\rBankJtrueStar$), because $J$ is flat near its minimum. The worst-case cost over $\mathcal{F}(\mu,\sigma)$ is $\rBankJrobWC$, $\rTelJrobWC$ and $\rAirJrobWC$, against true optima of $\rBankJtrueStar$, $\rTelJtrueStar$ and $\rAirJtrueStar$. This gap is the price of knowing only two moments. The security conclusion is therefore robust: in all three cases it lies $28$--$50\%$ above the expected loss, whether derived from the full distribution or from its first two moments.

\subsection{Properties relevant to independent review}
A quantitative loss model used to support high-stakes decisions will be reviewed adversarially. The TEIM was designed to have six properties that make such review productive.

\emph{Causal traceability.} Every loss dollar is attributed to a channel and, through the dependency graph, to the withdrawn component. By the causal-locality theorem (\cref{thm:locality}), services not reachable from the restrained component are unaffected, so the model cannot attribute loss to unrelated systems.

\emph{No double counting.} The partition theorem (\cref{thm:partition}) shows that the functional, attrition and funding channels, as constructed with survival weights, sum exactly to the loss of the degraded customer base. The approximation error of the usual additive shortcut is bounded explicitly.

\emph{Conservative treatment of speculative value.} The delay channel uses an option value that is provably no larger than the static DCF claim for launches that are optimal immediately (\cref{thm:delay}). This answers the standard objection that delayed-launch losses are speculative or overstated.

\emph{Mitigation built in.} The reported loss is that of the loss-minimising feasible response, including the cost of that response.

\emph{Testability and error rates.} Every theorem used for computation has a numerical check (\cref{tab:verif}). Monte Carlo error is reported for every statistic, and quantile intervals are distribution-free.

\emph{Sensitivity transparency.} The Sobol decomposition tells a reviewer which assumptions matter. An expert who disputes a parameter group with a small total index (for example the fallback-capability caps for Carrier~P) cannot move the result materially. Disputes are correctly focused on high-index groups (the customer response for Bank~A, the restraint horizon for Carrier~P).

\subsection{Threats to validity}
\emph{Internal validity (model specification).} Repair times in the CTMC are exponential. Real repair-time distributions are often lognormal or bimodal, and phase-type expansions would capture them without leaving the Markov framework. The fluid backlog approximation is exact in the deterministic limit. It understates waiting near critical load, but the stochastic check (\cref{tab:v-backlog}) shows deviations below $1\%$ in the operating range. The proportional-hazard response is linear in outage hours and exposure. Threshold or saturation effects (for example customers who leave after a second outage) would increase convexity and thus strengthen the plug-in bias result. Capability-to-revenue mappings are linear in the capability shortfall within each segment.

\emph{Construct validity.} Capability levels are discretised on a finite chain. This is required for the fixed-point semantics, and it is coarse relative to continuous degradation. A finer chain changes the levels but not the theory, since \cref{thm:lattice} holds for any finite chain. The CCFH policies are stylised as two discrete choices.

\emph{External validity.} The case data are synthetic. They are calibrated to public benchmarks, incident disclosures and regulatory reports, but they are not any single firm's data. The qualitative findings follow from structural properties established by theorems: the dominance of non-outage channels where prolonged degradation and customer response are present, the downward bias of plug-ins, misranking and the quantile security rule. The magnitudes do not generalise, and in applications they must be recalibrated to the firm.

\emph{Statistical conclusion validity.} Monte Carlo half-widths are below $0.5\%$ for all reported statistics. Sobol indices have bootstrap half-widths up to $0.03$, so rankings among groups with indices below $0.05$ should not be interpreted. The econometric experiment uses a simulated design, so its coverage results are design-specific.

\subsection{Limitations and extensions}
Five extensions would strengthen the framework. (i)~\emph{Endogenous restraint horizon.} In practice $T_R$ may depend on the remediation path, for instance if a demonstrated design-around shortens the restraint. This couples $T_R$ and $T_d$ and requires a stopping-game formulation. (ii)~\emph{Competitive response.} Customers lost through attrition go to competitors, whose pricing response could amplify or dampen the loss. A structural demand model would replace the reduced-form hazard. (iii)~\emph{Bayesian updating.} As post-restraint data arrive (incident counts, churn, remediation progress), the prior over $\theta$ can be updated and the security re-evaluated. For incident rates with gamma priors, the updating is closed-form by Poisson--gamma conjugacy. (iv)~\emph{Dependence uncertainty.} The Gaussian copula with stated rank correlations is a modelling choice. Bounds over all copulas consistent with the marginals (Fr\'echet--Hoeffding) would quantify the dependence risk in the tail. (v)~\emph{Multiple restrained components.} The lattice semantics already handles several simultaneous zero bases, but the remediation model would need a portfolio of interacting programmes competing for shared engineering capacity.

\section{Conclusion}\label{sec:conclusion}
We have developed a technical--economic impact model for losses caused by the forced withdrawal of a software component from an enterprise platform. The model is formal where the engineering is: lattice fixed points for degraded dependencies, CTMC reliability with CSL properties, fluid queues and stochastic activity networks for remediation. It is formal where the economics is too: present values, real options, customer lifetime value, coherent risk measures and decision theory. The main theoretical results include the existence and computability of degraded-capability semantics, causal locality, convexity of backlog costs and the resulting Jensen bias, exact channel partition, option-value bounds on forced-delay loss, and a critical-fractile security rule with a CVaR identity and a moment-robust counterpart. All are proved, and those used for computation are verified numerically.

Three anonymised case studies, a core-banking ledger, a 5G converged-charging function and an airline passenger-service system, show that the economically relevant loss is dominated by prolonged degraded operation, customer response and delayed commercialisation rather than downtime. Conventional plug-in estimates are biased low and can select the wrong mitigation, and the optimal security is a high quantile of the loss distribution, $1.3$--$1.4$ times its mean at central assumptions. The framework is reproducible, auditable and modular, so it can be recalibrated to a specific platform as firm data become available.

\appendix
\section{Notation}\label{app:notation}
\begin{longtable}{@{}p{0.22\textwidth}p{0.72\textwidth}@{}}
\toprule
Symbol & Meaning\\
\midrule
\endhead
$\mathcal{C}$ & Capability chain $\{0=c_0<\dots<c_m=1\}$ (\cref{def:chain})\\
$\mathcal{P}=(V,E,\{f_v\},b,K)$ & Platform: services, dependencies, aggregation operators, bases, revenue-bearing capabilities (\cref{def:platform})\\
$\mathrm{alt}_q$, $\mathrm{kth}_k$ & Capped alternative (fallback) and $k$-of-$n$ operators\\
$\Phi_{b}$, $\lfp$, $\gfp$ & Platform operator on $\mathcal{C}^V$ and its least/greatest fixed points\\
$r$, $T_R$ & Restrained component and random restraint horizon\\
$s\in\mathcal{S}$; $S_0,\dots,S_3$ & Response strategies: degrade-and-wait, nominal design-around, crashed design-around, substitution\\
$T_d(s)$, $\tau_s=T_R\wedge T_d(s)$ & Remediation completion time and effective degraded period\\
$\rho$ & Discount rate (industry cost of capital)\\
$L(s)=\sum_k L_k(s)$ & Present-value restraint-induced loss and its channels $k\in\{\mathrm{ac},\mathrm{inc},\mathrm{fun},\mathrm{ops},\mathrm{att},\mathrm{fund},\mathrm{rem},\mathrm{del}\}$\\
$\mathbf{Q}$, $\pi$, $A$, $f$ & CTMC generator, stationary law, availability, failure frequency\\
$R(t)=\Prob(\Always^{\le t}\mathsf{up})$ & CSL interval reliability\\
$\lambda,c,\psi$; $W(D)$ & Arrival rate, recovery capacity, deferral fraction; fluid backlog area for outage length $D$\\
$L_{\mathrm{ac}}(D)$ & Acute loss function of an outage of length $D$\\
$m$ & Fallback incident-rate multiple (architecture-implied)\\
$\Lambda_j=\beta^{\mathrm{out}}_jH+\beta^{\mathrm{deg}}_jx_j\tau_s$ & Cumulative excess-churn hazard of segment $j$\\
$V_t$, $I$, $\delta$, $\sigma$, $r$ & Launch value, investment cost, convenience yield, volatility, risk-free rate\\
$\beta$, $V^\ast$ & Positive root of the fundamental quadratic and exercise threshold\\
$\ell(\Delta)$, $\ell_{\mathrm{DCF}}(\Delta)$ & Option loss and static DCF loss from a forced delay $\Delta$\\
$\mathrm{PM}$, $\mathrm{TDEV}$ & COCOMO~II effort (person-months) and schedule (months)\\
$q$, $w$, $c$, $J(B)$ & Probability of wrongful restraint, weight on uncompensated loss, carrying cost, expected cost of security $B$\\
$\alpha^\star=1-c/(qw)$ & Critical fractile\\
$\VaR_\alpha$, $\CVaR_\alpha$ & Value-at-risk and conditional value-at-risk at level $\alpha$\\
$B^\star$, $B^{\mathrm{rob}}$ & Critical-fractile and moment-robust securities\\
$S_1$, $S_T$ & First-order and total Sobol indices (context distinguishes these from strategies)\\
\bottomrule
\end{longtable}

\section{Glossary of platform terms}\label{app:glossary}
\begin{description}[leftmargin=1.2em,itemsep=1pt,font=\normalfont\bfseries]
\item[Core-banking ledger] The system of record for accounts, balances and postings. \emph{Real-time posting} updates balances as transactions occur. \emph{Batch posting} applies them at scheduled cycles, so intra-day balances are provisional.
\item[Stand-in processing] A fallback that authorises card or payment transactions against a cached balance when the ledger is unavailable, at the risk of overdrafts and fraud.
\item[CHF, OCS, CCFH] The 5G converged charging function (CHF) and the legacy 4G online charging system (OCS) grant usage quota in real time. Credit-control failure handling (CCFH) defines what the network does when charging is unreachable: \textsc{continue} (serve without charging, causing revenue leakage) or \textsc{terminate} (deny service) \citep{rfc8506}.
\item[IWF] Interworking function translating 5G service-based charging (Nchf) into 4G Diameter (Gy) so that a 4G OCS can serve 5G sessions.
\item[PSS, RAE] Passenger-service system (reservations, inventory, departure control). The real-time availability and re-shop engine (RAE) computes availability and offers continuously.
\item[Continuous pricing] Fare computation at arbitrary price points instead of a finite set of booking-class buckets.
\item[IROPS] Irregular operations: disruption of the schedule by weather, technical or crew events, which requires re-accommodation of passengers.
\item[NDC] IATA's New Distribution Capability, an XML/JSON offer-and-order standard for distributing dynamic offers through intermediaries.
\end{description}

\section{Reproducibility and audit checklist}\label{app:repro}
\begin{enumerate}[leftmargin=1.8em,itemsep=1pt]
\item \emph{Inputs.} Every parameter has a distribution, a unit and a stated basis (\cref{tab:calib-bank,tab:calib-tel,tab:calib-air}). Changing an input is a one-line change in the case module.
\item \emph{Seeds.} Parameter sampling, event sampling and the Sobol design use fixed seeds (\cref{sec:impl}). Re-running reproduces all tables bit-for-bit on the same library versions.
\item \emph{Verification.} \code{verify.py} reproduces \cref{tab:verif,tab:v-backlog,tab:v-acute,tab:v-fix,tab:v-option}.
\item \emph{Experiments.} \code{experiments.py} produces \code{results.json}. The fragment generators then produce every table, figure data file and numerical macro in \cref{sec:experiments}.
\item \emph{Audit trail.} For any draw, the engine can report the full channel vector, the fixed-point capability vector, $T_R$, $T_d$ and the incident and outage realisations. A reviewer can therefore trace any tail scenario to its causes.
\item \emph{Stress tests.} Any parameter can be pinned to a value, which yields the conditional loss distribution under a disputed assumption. The tornado and Sobol analyses are special cases of this.
\end{enumerate}

\end{document}